\pdfoutput=1

\documentclass[12pt,floats,onecolumn,a4paper]{article}

\usepackage[noblocks]{authblk}
\usepackage{bm}
\usepackage[numbers]{natbib}
\usepackage{subcaption}
\usepackage{hyperref}
\usepackage{enumerate}

\usepackage{color}
\usepackage{amsfonts}
\usepackage{amsmath}
\usepackage{graphicx}

\usepackage{mathtools}
\mathtoolsset{showonlyrefs}
\usepackage{enumerate}
\usepackage{amssymb}
\usepackage{amsthm}
\newcommand*{\bfs}{\bm{s}}
\newcommand*{\bftau}{\bm{\tau}}
\newcommand*{\bfone}{\mathbf{1}}

\newcommand{\appref}[1]{\hyperref[#1]{{Appendix~\ref*{#1}}}}
\newcommand{\be}{\begin{eqnarray} \begin{aligned}}
\newcommand{\ee}{\end{aligned} \end{eqnarray} }
\newcommand{\benn}{\begin{eqnarray*} \begin{aligned}}
\newcommand{\eenn}{\end{aligned} \end{eqnarray*}}

\newcommand*{\cB}{\mathcal{B}}
\newcommand*{\cC}{\mathcal{C}}

\newcommand*{\cF}{\mathcal{F}}

\newcommand*{\cH}{\mathcal{H}}

\newcommand*{\cL}{\mathcal{L}}

\newcommand*{\cD}{\mathcal{D}}

\newcommand*{\cR}{\mathcal{R}}
\newcommand*{\cO}{\mathcal{O}}
\newcommand*{\cP}{\mathcal{P}}

\newcommand*{\tr}{\mathop{\mathrm{tr}}\nolimits}

\newcommand{\bc}{\begin{center}}
\newcommand{\ec}{\end{center}}

\newtheorem{theorem}{Theorem}[section]
\newtheorem{lemma}[theorem]{Lemma}

\newtheorem{definition}[theorem]{Definition}

\newtheorem{observation}[theorem]{Observation}

\usepackage{amsfonts}

\def\Complex{\mathbb{C}}

\def\Integer{\mathbb{Z}}
\def\Anyons{\mathbb{A}}

\def\01{\{0,1\}}

\newcommand{\ket}[1]{|#1\rangle}
\newcommand{\bra}[1]{\langle#1|}

\newcommand{\proj}[1]{|#1\rangle\langle#1|}

\newcommand{\inp}[2]{\langle{#1}|{#2}\rangle} 

\newtheorem{conjecture}{Conjecture}

\newcommand*{\Heff}{H_{\textsf{eff}}} 
\newcommand*{\Htriv}{H_{\textsf{triv}}}
\newcommand*{\Htop}{H_{\textsf{top}}} 

\newcommand*{\Eeff}{E^{\textsf{eff}}} 

\newcommand*{\Peff}{\Pi^{\textsf{eff}}} 

\begin{document}

\newcommand*{\ad}[1]{\mathsf{Ad}_{#1}}
\newcommand{\Vod}{V_{\text{od}}}
\newcommand{\Vd}{V_{\text{d}}}

\title{Preparing topologically ordered states\\
 by Hamiltonian interpolation}
 
\author[1]{Xiaotong Ni} 
\author[2]{Fernando Pastawski}
\author[3]{Beni Yoshida}
\author[4]{Robert K\"onig}
\affil[1]{Max-Planck-Institute of Quantum Optics,  85748 Garching bei M\"unchen, Germany }
\affil[2]{Institute for Quantum Information \& Matter, California Institute of Technology,  Pasadena CA 91125, USA}
 \affil[3]{Perimeter Institute for Theoretical Physics, Waterloo, ON N2L 2Y5, Canada}  
 \affil[4]{Institute for Advanced Study  \& Zentrum Mathematik, Technische Universit\"at M\"unchen, 85748 Garching, Germany}

\maketitle 
 
\maketitle
\begin{abstract}
We study the preparation of topologically ordered states by interpolating between an initial Hamiltonian with a unique product ground state and a Hamiltonian with a topologically degenerate ground state space. By simulating the dynamics for small systems, we numerically observe a certain stability of the prepared state as a function of the initial Hamiltonian. For small systems or long interpolation times, we argue that the resulting state can be identified by computing suitable effective Hamiltonians. For effective anyon models, this analysis singles out the relevant physical processes and extends the study of the splitting of the topological degeneracy by Bonderson~\cite{bonderson2009splitting}.  We illustrate our findings using Kitaev's Majorana chain, effective anyon chains, the toric code and Levin-Wen string-net models.
\end{abstract}

\tableofcontents

\section{Introduction\label{sec:introduction}}
Topologically ordered phases of matter have attracted significant interest in the field of quantum information, following the seminal work of Kitaev~\cite{kitaev2003fault}. 
From the viewpoint of quantum computing, one of their most attractive features is their ground space degeneracy: it provides a natural quantum error-correcting code for encoding and manipulating  information. 
Remarkably, the ground space degeneracy  is approximately preserved in the presence of weak static Hamiltonian perturbations~\cite{bravyi2010topological, bravyi2011short, Michalakis2013}. 
This feature suppresses the uncontrolled accumulation of relative phases between code states, and thus helps to overcome decoherence. 
This is  a necessary requirement for the realization of many-body quantum memories~\cite{dennis2002topological}.

To use topologically ordered systems as quantum memories and for fault-tolerant quantum computation, concrete procedures for the preparation of specific ground states are required. 
Such mechanisms depend on the model Hamiltonian which is being realized as well as on the particular experimental realization.  Early work~\cite{dennis2002topological} discussed the use of explicit unitary encoding circuits for the toric code. This consideration is natural for systems where we have full access to unitary gates over the underlying degrees of freedom. We may call this the {\em bottom-up approach} to quantum computing: here one proceeds by building and characterizing individual components before assembling them into larger structures. An example are arrays of superconducting qubits~\cite{Barends2014, Chow2014, Corcoles2015}. 
Other proposed procedures for  state preparation in this approach involve engineered dissipation~\cite{dengis2014optimal, bardyn2013topology} or measurement-based preparation~\cite{odyga2015}.
However, achieving the control requirements for experimentally performing such procedures is quite challenging.
They require either 
a)~independently applying complex sequences of gates  on each of the elementary constituents 
b)~precisely engineering a dissipative evolution, or 
c)~performing an extensive set of local measurements and associated non-local classical data processing to determine and execute a suitable
unitary correction operation. Imperfections in the implementation of such protocols pose a severe problem, especially in cases where the preparation time is extensive~\cite{BravyiHastingsetal06,konig2014generating}.

In fact, these procedures achieve more than is strictly necessary for quantum computation: any ground state can be prepared in this fashion. That is, they constitute~\emph{encoders}, realizing an isometry from a number of unencoded logical qubits to the ground space of the target Hamiltonian. 
We may ask if the task of preparing topologically ordered state becomes easier if  the goal is to prepare specific states instead of encoding arbitrary states. In particular, we may ask this question in the {\em top-down approach} to quantum computing, where the
quantum information is encoded in the ground space of a given condensed matter Hamiltonian. An example are Majorana wires~\cite{Mourik2012, Nadj-Perge2014} or fractional quantum Hall substrates~\cite{Venkatachalam2011}.
Indeed, a fairly standard approach to preparing ground states of a Hamiltonian is to cool the system by weakly coupling it with a thermal bath at a temperature significantly lower than the Hamiltonian gap. 
Under appropriate ergodicity conditions, this leads to convergence to  a state mainly supported on the ground space.  Unfortunately, when using natural equilibration processes,  convergence may be slow, and the resulting prepared state is generally a  (logical) mixed state unsuitable for computation.

A natural  alternative method for preparing ground states of a given Hamiltonian is adiabatic evolution: here one initializes the system in an easy-to-prepare state (e.g., a product state), which is the unique ground state of a certain initial Hamiltonian (e.g., describing a uniform field). 
Subsequently, the Hamiltonian of the system is gradually changed (by tuning external control parameters in a time-dependent fashion) until the target Hamiltonian is reached.
If this time-dependent change of the Hamiltonian is ``slow enough'', i.e., satisfies a certain adiabaticity condition (see Section~\ref{sec:adiabaticity}), the state of the system will closely follow the trajectory of instantaneous ground states.
The resulting state then is  guaranteed to be mainly supported on the ground space of the target Hamiltonian, as desired.

Adiabatic preparation has some distinct advantages compared to e.g., encoding using a unitary circuit.  
For example, in contrast to the latter, adiabatic evolution guarantees that the final state is indeed a ground state of the {\em actual} Hamiltonian describing the system, independently of potential imperfections in the realization of the ideal Hamiltonians. 
In contrast, a unitary encoding circuit is designed to encode into the ground space of an ideal model Hamiltonian, and will therefore generally not prepare exact ground states of the actual physical system (which only approximate  the model Hamiltonian). 
Such an encoding into the ideal ground space may lead to a negligible quantum memory time in the presence of an unknown perturbation~\cite{Pastawski2010}; this is because ideal and non-ideal (perturbed) ground states may differ significantly (this phenomenon is referred to as Anderson's orthogonality catastrophe~\cite{Anderson1967}). 
Adiabatic evolution, on the other hand, elegantly sidesteps these issues.

The fact that adiabatic evolution can follow the actual ground state of a system Hamiltonian  makes it a natural candidate for achieving the task of topological code state preparation. 
An additional attractive feature is that its experimental requirements are rather modest: while some time-dependent control is required, this can be local, and additionally translation-invariant.
Namely, the number of external control parameters required does not scale with the system size or code distance.

\subsubsection*{Summary and outlook}
Motivated by these observations, we consider the general problem of preparing topologically ordered states by what we refer to as {\em Hamiltonian interpolation}. 
We will use this terminology instead of ``adiabatic evolution'' since in some cases, it makes sense to consider scenarios where adiabaticity guarantees cannot be given. For concreteness, we consider a time-dependent Hamiltonian $H(t)$ which monotonically sweeps over the path
\begin{align}
H(t)  = (1-t/T)\cdot \Htriv+t/T\cdot \Htop\qquad\ t\in [0,T]\ ,\label{eq:hlinearinterpol}
\end{align} 
i.e., we  assume  that the interpolation is linear in time and takes overall time\footnote{ We remark that in some cases, using a non-linear monotone `schedule' $\vartheta:[0,T]\rightarrow [0,1]$ with $\vartheta(0)=0$, $\vartheta(T)=1$ and smooth derivatives may be advantageous  (see Discussion in Section~\ref{sec:adiabaticity}).  
However, for most of our considerations, the simple linear interpolation~\eqref{eq:hlinearinterpol} is sufficient.}~$T$. Guided by experimental considerations, we focus on the translation-invariant case: here the Hamiltonians~$H(t)$ are translation-invariant throughout the evolution. More precisely, we consider the process of interpolating between a Hamiltonian~$\Htriv$ with unique ground state~$\Psi(0)=\varphi^{\otimes L}$ and a Hamiltonian~$\Htop$ with topologically degenerate ground space (which is separated from the remainder of the spectrum by a constant gap): the state~$\Psi(t)$ of the system  at time~$t\in [0,T]$ satisfies the equation of motion
\begin{align}
\frac{\partial\Psi(t)}{\partial t} &=-i H(t)\Psi(t)\ ,\qquad \Psi(0)=\varphi^{\otimes L}\ . \label{eq:adiabaticprep}
\end{align}
Generally, we consider families of Hamiltonians (or models) parametrized by a system size~$L$; throughout, we will assume that~$L$ is the number of single particles, e.g., the number of qubits (or sites) in a lattice with Hilbert space~$\cH=(\mathbb{C}^2)^{\otimes L}$. The dimension of the ground space of~$\Htop$ will be assumed to be independent of the system size.

Our goal is to characterize the set of states which are preparable by such Hamiltonian interpolations starting from various product states, i.e., by choosing different initial Hamiltonians~$\Htriv$. To each choice~$\Psi(0) = {\varphi}^{\otimes L}$ of product state we associate a normalized initial trivial Hamiltonian~$\Htriv :=-\sum_j P^{(j)}_{\varphi}$ which fully specifies the  interpolating path of Eq.~\eqref{eq:hlinearinterpol}, with $P^{(j)}_{\varphi}=\proj{\varphi}$ being the single particle projector onto the state~$\varphi$ at site $j$.

In the limit~$T\rightarrow\infty$, one may think of this procedure as  associating an encoded (logical) state~$\iota(\varphi)$ to any single-particle state~$\varphi$. However, some caveats are in order: first, the global phase of the state~$\iota(\varphi)$ cannot be defined in a consistent manner in the limit~$T\rightarrow\infty$, and is therefore not fixed. Second, the final state in the evolution~\eqref{eq:adiabaticprep} does not need to be supported entirely on the ground space of~$\Htop$  because of non-adiaticity errors, i.e., it is not a logical (encoded)  state itself. To obtain a logical state, we should think of $\iota(\varphi)$ as the final state projected onto the ground space of $\Htop$.  Up to these caveats,  our goal is essentially to characterize the image of 
the association~$\iota:\varphi \mapsto \iota(\varphi)$, as well as its continuity properties.
We will also define  an analogous map~$\iota_T$ associated to fixed evolution time $T$ and study it numerically by simulating the corresponding Schr\"odinger equation~\eqref{eq:adiabaticprep} on a classical computer.

While there is a priori no  obvious  relationship between the final states~$\iota_T(\varphi)$, $\iota_T(\varphi')$ resulting from different initial (product) states~$\varphi^{\otimes L}, \varphi'^{\otimes L}$, we numerically find that the image of~$\iota_T$ is concentrated around a particular discrete family of encoded states. 
In particular, we observe  for small system sizes  that the preparation enjoys a certain stability property: variations in the initial Hamiltonian do not significantly affect the final state.
We support this through analytic arguments, computing effective Hamiltonians associated to perturbations around $H_{\textsf{top}}$ which address the large $T$ limit.
This also allows us to provide a partial prediction of which states~$\iota(\varphi)$ may be obtained through such a preparation process. We find that under certain general conditions,~$\iota(\varphi)$ belongs  to a certain finite  family of preferred states which depend on the final Hamiltonian~$H_{\textsf{top}}$. As we will argue, there is  a natural relation between the corresponding states $\iota(\varphi)$ for different system sizes: they encode the same logical state if corresponding logical operators are chosen (amounting to a choice of basis of the ground space).

Characterizing the set~$\{\iota(\varphi)\}_\varphi$  of states preparable  using this kind of Hamiltonian interpolation is important for quantum computation because certain encoded states (referred to as ``magic states'') can be used as a resource for universal computation~\cite{bravyi2005universal}. 
Our work provides insight into this question for `small' systems, which we deem experimentally relevant.
Indeed, there is a promising degree of robustness for the Hamiltonian interpolation to prepare certain (stabilizer) states.
However, a similar preparation of magic states seems to require imposing additional symmetries which will in general not be robust.
We exemplify our considerations using various concrete models, including Kitaev's Majorana chain~\cite{kitaev2001unpaired} (for which we can provide an exact solution), effective anyon chains (related to the so-called golden chain~\cite{feiguin2007interacting} and the description used by Bonderson~\cite{bonderson2009splitting}), as well as the toric code~\cite{kitaev2003fault} and Levin-Wen string-net models~\cite{levin2005string} (for which we simulate the time-evolution for small systems, for both the doubled semion and the doubled Fibonacci model). 

\subsubsection*{Prior work}
The problem of preparing topologically ordered states by adiabatic interpolation has been considered prior to our work by Hamma and Lidar~\cite{hamma2008adiabatic}. 
Indeed, their contribution is one of the main motivations for our study.  
They study  an adiabatic evolution where a Hamiltonian having a trivial product ground state is interpolated into a toric code Hamiltonian having a four-fold degenerate ground state space.  
They found that while the gap for such an evolution must forcibly close, 
this may happen through   second order phase transitions. Correspondingly,  the closing of the gap is only polynomial in the system size.
This allows an efficient polynomial-time Hamiltonian interpolation to succeed at accurately preparing certain ground states. 
We revisit this case in Section~\ref{sec:symmetryprotected} and give further examples of this phenomenon.
The authors of~\cite{hamma2008entanglement} also observed the stability of the encoded states with respect to perturbations in the preparation process. 

Bonderson~\cite{bonderson2009splitting} considered the problem of characterizing the lowest order degeneracy splitting in topologically ordered models.
Degeneracy lifting can be associated to tunneling of anyonic charges, part of which may be predicted by the universal algebraic structure of the anyon model.
Our conclusions associated to Sections~\ref{sec:anyonchains} and~\ref{sec:twodimensionalsystems} can be seen as supporting this perspective.

\subsubsection*{Beyond small systems}
In general, the case of larger systems (i.e., the thermodynamic limit) requires a detailed understanding of the quantum phase transitions \cite{Sachdev2011} occurring when interpolating between~$H_{\textsf{triv}}$ and~$H_{\textsf{top}}$.  Taking the thermodynamic limit while making~$T$ scale as a polynomial of the system size raises a number of subtle points. A major technical difficulty is that existing adiabatic theorems do not apply, since at the phase transition gaps associated to either of the relevant phases close. This is alleviated by scaling  the interpolation time~$T$ with the system size and splitting the adiabatic evolution into two regimes, the second of which can be treated using degenerate adiabatic perturbation theory~\cite{Rigolin2010, Rigolin2012, Rigolin2014}. However, such a methodology still does not yield complete  information about the dynamical effects of crossing a phase boundary. 

More generally, it is natural to conjecture that interpolation between different phases yields only a discrete number of distinct states corresponding  to a discrete set of continuous phase transitions in the thermodynamic limit. Such a conjecture links the problem of Hamiltonian interpolation to that of classifying phase transitions between topological phases.
It can be motivated by the fact that only a discrete set of possible condensate-induced continuous phase transitions is predicted to exist in the thermodynamic limit~\cite{bais2009condensate, burnell2011condensation}.

\section{Adiabaticity and ground states\label{sec:adiabaticity}}
The first basic question arising in this context is whether the evolution~\eqref{eq:adiabaticprep} yields a state~$\Psi(T)$ close to the ground space of~$\Htop$. 
The adiabatic theorem in its multiple forms (see e.g.,~\cite{teufel2003adiabatic}) provides  \emph{sufficient} conditions for this to hold: These theorems guarantee that given a Hamiltonian path $\{H(t)\}_{0\leq t\leq T}$ satisfying certain smoothness and gap assumptions, initial eigenstates evolve into approximate instantaneous eigenstates under an evolution of the form~\eqref{eq:adiabaticprep}. 
The latter assumptions are usually of the following kind:
\begin{enumerate}[(i)]
\item{\bf Uniform gap:}\label{it:gapassumptionadia}
There is a uniform lower bound~$\Delta(t)\geq \Delta >0$ on the spectral gap of~$H(t)$ for all $t\in [0,T]$. The relevant spectral gap $\Delta(t)$ is the energy difference between the  ground space $P_0(t)\cH$ of the instantaneous Hamiltonian $H(t)$ and the rest of its spectrum. Here and below, we denote by~$P_0(t)$ the spectral projection onto the ground space
\footnote{More generally, $P_0(t)$ may be the sum of the spectral projections of $H(t)$ with eigenvalues in a given interval, which is separated by a gap~$\Delta(t)$ from the rest of the spectrum.} of~$H(t)$.
\item{\bf Smoothness:}\label{it:derivativesadia}
There are constants $c_1,\ldots,c_M$ such that the $M$ first derivatives of $H(t)$ are uniformly bounded in operator norm, i.e., for all $j=1,\ldots,M$, we have 
\begin{align}
\big\|\frac{d^j}{dt^j}H(t)\big\|\leq c_j\qquad\textrm{ for all }t\in [0,T]\ .\label{eq:operatornormderivative}
\end{align}
\end{enumerate}

The simplest version of such a theorem is: 
\begin{theorem}\label{thm:adiabaticthm}
Given a state $\Psi(0)$ such that $P_0(0)\Psi(0)=\Psi(0)$ and a uniformly gapped Hamiltonian path $H(t)$ for $t\in [0,T]$ given by Eq.~\eqref{eq:hlinearinterpol}, the state~$\Psi(T)$ resulting from the evolution~\eqref{eq:adiabaticprep} satisfies
\begin{align}
 \| \Psi(T) - P_0(T)\Psi(T) \| = O(1/T)\ .
\end{align}
In other words, in the adiabatic limit of large times $T$, the state $\Psi(T)$ belongs to the instantaneous eigenspace $P_0(T)\cH$ and its distance from the eigenspace is $O(1/T)$. 
\end{theorem}

This version is sufficient to support our analytical conclusions qualitatively.
For a quantitative analysis of non-adiabaticity errors, we perform numerical simulations.
Improved versions of the adiabatic theorem (see~\cite{ge2015rapid,lidar2009adiabatic}) provide tighter analytical error estimates for general interpolation schedules at the cost of involving higher order derivatives of the Hamiltonian path~$H(t)$ (see Eq.`\eqref{eq:operatornormderivative}), but do not change our main conclusions. 

Several facts prevent us from directly applying such an adiabatic theorem to our evolution~\eqref{eq:hlinearinterpol} under consideration. 

\paragraph{Topological ground space degeneracy.}  
Most notably, the gap assumption~\eqref{it:gapassumptionadia} is not satisfied if we study ground spaces: we generally consider the case where $H(0)=\Htriv$ has a unique ground state, whereas the final Hamiltonian $H(T)=\Htop$ is topologically ordered and  has a  degenerate ground space (in fact, this degeneracy is exact and independent of the system size for the models we consider). 
This means that if~$P_0(t)$ is the projection onto the ground space of $H(t)$, there is no uniform lower bound on the gap $\Delta(t)$.

 We will address this issue by restricting our attention to times $t\in [0,\kappa T]$, where $\kappa\approx 1$ is chosen such that $H(\kappa T)$ has a non-vanishing gap but still is ``inside the topological phase''. 
We will illustrate in specific examples how $\Psi(T)$ can  indeed be recovered by taking the limit~$\kappa\rightarrow 1$.

We emphasize that the expression ``inside the phase'' is physically not well-defined at this point since we are considering a Hamiltonian of a fixed size. 
Computationally, we take it to mean that the Hamiltonian can be analyzed by a convergent perturbation theory expansion starting from the unperturbed Hamiltonian~$\Htop$. 
The resulting lifting of the ground space degeneracy of~$\Htop$ will be discussed in more detail in Section~\ref{sec:effectivehamiltonians}.  

\paragraph{Dependence on the system size.}
A second potential obstacle for the use of the adiabatic theorem
is the dependence on the system size~$L$ (where e.g., $L$ is the number of qubits). 
This dependence enters in the operator norms~\eqref{eq:operatornormderivative}, which are extensive in~$L$ -- this would lead to polynomial dependence of $T$ on $L$ even if e.g., the gap were constant (uniformly bounded). 

More importantly, the system size enters in the gap~$\Delta(t)$: in the topological phase, 
the gap  (i.e., the splitting of the topological degeneracy of~$\Htop$) is exponentially small in~$L$ for constant-strength local perturbations to~$\Htop$, as shown for the models considered here by Bravyi, Hastings and Michalakis~\cite{bravyi2010topological}.
 Thus a na\"ive application of the adiabatic theorem only yields a guarantee on the ground space overlap of the final state if the evolution time is exponentially large in~$L$. 
 This is clearly undesirable for large systems; one may try to prepare systems faster  (i.e., more efficiently) but would need alternate arguments to ensure that the final state indeed belongs to the ground space of $\Htop$.
 
For these reasons, we restrict our attention to the following two special cases of the Hamiltonian interpolation~\eqref{eq:hlinearinterpol}:
\begin{itemize}
\item
{\em Symmetry-protected preparation}: if there is a set of observables commuting with both~$\Htriv$ and~$\Htop$, these will represent conserved quantities throughout the Hamiltonian interpolation.
If the initial state is an eigenstate of such observables, one may restrict the Hilbert space to the relevant eigenvalue, possibly resolving the  topological degeneracy and guaranteeing a uniform gap.
This observation was first used in~\cite{hamma2008adiabatic} in the context of the toric code: for this model, such a restriction allows mapping the problem to a transverse field Ising model, where the gap closes polynomialy with the system size.
We identify important cases  satisfying this condition. 
While this provides the most robust preparation scheme, the resulting encoded states are somewhat restricted (see Section~\ref{sec:symmetryprotected}). 
\item
{\em Small systems: } For systems of relatively small (constant) size $L$ , the adiabatic theorem can be applied as all involved quantities are essentially constant. 
In other words, although `long' interpolation times are needed to reach ground states of~$\Htop$ (indeed, these may depend exponentially on $L$), these may still be reasonable experimentally. 
The consideration of small system is motivated by  current experimental efforts to realize surface codes~\cite{kelly2015state}: they are usually restricted to a small number of qubits, and this is the scenario we are considering here (see Section~\ref{sec:smallsystemcase}). 
\end{itemize}
Obtaining a detailed understanding of the general large~$L$ limiting behaviour (i.e., the thermodynamic limit) of the interpolation process~\eqref{eq:hlinearinterpol} is beyond the scope of this work.

\subsection{Symmetry-protected preparation\label{sec:symmetryprotected}}
Under particular circumstances, the existence of conserved quantities permits applying the adiabatic theorem while evading the technical obstacle posed by a vanishing gap in the context of topological order. 
Such a case was considered by Hamma and Lidar~\cite{hamma2008adiabatic}, who showed that certain ground states of the toric code can be prepared efficiently. 
We can formalize sufficient conditions in the following general way (which then is applicable to a variety of models, as we discuss below).

\begin{observation}\label{obs:observationsymmetryprotec}
Consider the interpolation process~\eqref{eq:hlinearinterpol} in a Hilbert space~$\cH$. 
Let $P_0(T)$ be the projection onto the ground space $P_0(T)\cH$ of $H(T)=\Htop$. 
Suppose that $Q=Q^2$ is a projection such that 
\begin{enumerate}[(i)]
\item $Q$ is a conserved quantity: $[Q,\Htop]=[Q,\Htriv]=0$.\label{it:firstpropertyconserved}
\item The initial state $\Psi(0)$ is the ground state of $\Htriv$, i.e., $P_0(0)\Psi(0)=\Psi(0)$ and  satisfies $Q\Psi(0)=\Psi(0)$. \label{it:initialstateconserved}
\item The final ground space has support on $QP_0(T)\cH \neq 0$
\item The restriction~$QH(t)$ of $H(t)$ to $Q\cH$ has gap $\Delta(t)$ which is bounded by a constant~$\Delta$ uniformly in~$t$, i.e., $\Delta(t)\geq \Delta$ for all $t\in [0,T]$. \label{it:lastidentityconserved}
\end{enumerate}
Then $Q\Psi(t)=\Psi(t)$, and the adiabatic theorem  can be applied with lower bound~$\Delta$ on the gap, yielding $\| \Psi(T) - P_0(T)\Psi(T)\| \leq O(1/T)$.
\end{observation}
The proof of this statement is a straightforward application of the adiabatic theorem (Theorem~\ref{thm:adiabaticthm}) to the Hamiltonians $Q\Htriv$ and $Q\Htop$ in the restricted subspace $Q\cH$. In the following sections, we will apply Observation~\ref{obs:observationsymmetryprotec} to various systems.
It not only guarantees that the  ground space is reached, but also gives us information about the specific state prepared in a degenerate ground space.

As an example of the situation discussed in Observation~\ref{obs:observationsymmetryprotec}, we discuss the case of fermionic parity conservation in Section~\ref{sec:majoranachain}. This symmetry is naturally present in fermionic systems. We expect our arguments to extend to more general topologically ordered Hamiltonians with additional symmetries. It is well-known that imposing global symmetries on top of topological Hamiltonians provides interesting classes of systems. Such symmetries can exchange anyonic excitations, and their classification as well as the construction of  associated defect lines in topological Hamiltonians is a topic of ongoing research~\cite{Beigi2011, Kitaev2012, Barkeshli13}. The latter problem is intimately related to the realization (see e.g.,~\cite{bombindelgado09,bombin15}) of transversal logical gates, which leads to similar classification problems~\cite{BravyiKoenig13, Beverland2014, Yoshida2015, Yoshida2015c}.  Thus we expect that there is a close connection between adiabatically preparable states and  transversally implementable logical gates. Indeed, a starting point for establishing such a connection could be the consideration of interpolation processes respecting symmetries realized by transversal logical gates.

For later reference, we also briefly discuss a situation involving conserved quantities which -- in contrast to Observation~\ref{obs:observationsymmetryprotec} -- project onto excited states of the final Hamiltonian. In this case, starting with certain eigenstates of the corresponding  symmetry operator~$Q$, the ground space cannot be reached:

\begin{observation}\label{obs:groundstatenotreached}
Assume that $Q, \Htriv, \Htop,\Psi(0)$ 
obey properties~\eqref{it:firstpropertyconserved},\eqref{it:initialstateconserved} and~\eqref{it:lastidentityconserved} of Observation~\ref{obs:observationsymmetryprotec}. If the ground space~$P_0(T)\cH$ of $\Htop$ satisfies $QP_0(T)\cH=0$ (i.e., is orthogonal to the image of~$Q$),  then the Hamiltonian interpolation cannot reach the ground space of $\Htop$, i.e., 
$\langle \Psi(T),P_0(T)\Psi(T)\rangle=\Omega(1)$. 
\end{observation}
The proof of this observation is trivial since $Q$ is a conserved quantity of the Schr{\"o}dinger evolution. 
Physically, the assumptions imply  the occurrence of a level-crossing where the energy gap exactly vanishes and  eigenvalue of~$Q$ restricted to the ground space changes.
We will encounter this scenario in the case of the toric code on a honeycomb lattice, see Section~\ref{sec:toric_numerics}.

\subsection{Small-system case\label{sec:smallsystemcase}}
In a more general scenario, there may not be a conserved quantity as in Observation~\ref{obs:observationsymmetryprotec}. Even assuming that the ground space is reached by the interpolation process~\eqref{eq:hlinearinterpol}, it
is a priori unclear which of the ground states is prepared. Here we address this question.

As remarked earlier,  we focus on systems of a constant size~$L$, and assume that the preparation time $T$ is large compared to~$L$. Generically, the Hamiltonians $H(t)$ are then non-degenerate (except  at the endpoint, $t\approx T$, where $H(t)$ approaches~$\Htop$). Without fine tuning, we may expect that there are no exact level crossings in the spectrum of $H(t)$ along the path $t\mapsto H(t)$ (say for some times~$t\in [0,\kappa T]$, $\kappa\approx 1$). For 
sufficiently large overall evolution times~$T$, we may apply the adiabatic theorem to conclude that the state of the system follows the (unique) instantaneous ground state (up to a constant error). Since our focus is on small systems, we will henceforth assume that this is indeed the case,  and summarily refer to this as the {\em adiabaticity assumption}. Again, we emphasize that this is a priori only reasonable  for small systems.

Under the adiabaticity assumption, we can  conclude that the prepared state~$\Psi(T)$ roughly coincides with the state obtained by computing the (unique) ground state $\psi_\kappa$ of $H(\kappa T)$, and taking the limit $\kappa\rightarrow 1$.  
In what follows, we adopt this computational prescription for identifying prepared states. 
Indeed, this approach yields  states that match our numerical simulation, and provides the correct answer for certain exactly solvable cases.  
Furthermore, the computation of the states $\psi_\kappa$ (in the limit $\kappa\rightarrow 1$) also clarifies the physical mechanisms responsible for the observed stability property of preparation: we can relate the prepared states to certain linear combination of string-operators (Wilson-loops), whose coefficients depend on the geometry (length) of these loops, as well as the amplitudes of certain local particle creation/annihilation and tunneling processes.

Since $H(\kappa T)$ for $\kappa\approx 1$ is close to the topologically ordered Hamiltonian~$\Htop$, it is natural to use ground states (or logical operators) of the latter as a reference to express the instantaneous states~$\psi_\kappa$. 
Indeed, the problem essentially reduces to a system described by~$\Htop$, with an additional perturbation given by  a scalar multiple of $\Htriv$. 
Such a local perturbation generically splits the topological degeneracy of the ground space. 
The basic mechanism responsible for this splitting for topologically ordered systems has been investigated by  Bonderson~\cite{bonderson2009splitting}, who quantified the degeneracy splitting in terms of local anyon-processes. 
We seek to 
identify low-energy ground states: this amounts to considering the effective low-energy dynamics (see Section~\ref{sec:effectivehamiltonians}).
This will provide valuable information concerning the set $\{\iota(\varphi)\}$.

\section{Effective Hamiltonians\label{sec:effectivehamiltonians}}

As discussed in Section~\ref{sec:smallsystemcase}, 
for small systems (and sufficiently large times~$T$), the  state~$\Psi(\kappa T)$  in the interpolation process~\eqref{eq:hlinearinterpol} should coincide with the
ground state of the instantaneous Hamiltonian $H(\kappa T)$. For $\kappa \approx 1$, the latter is  a perturbed version of the Hamiltonian~$\Htop$, where the perturbation is a  scalar multiple of~$\Htriv$. That is, up to rescaling by an overall constant, we are concerned with  
a Hamiltonian of the form
\begin{align}
H_0+\epsilon V\qquad\label{eq:perturbedHamiltonian}
\end{align}
where $H_0=\Htop$ is the target Hamitonian and  $V=\Htriv$ is the perturbation.  To compute the ground state of a Hamiltonian of the form~\eqref{eq:perturbedHamiltonian}, we use  {\em effective Hamiltonians}. These provide a description of the system in terms of effective low-energy degrees of freedom.

\subsection{Low-energy degrees of freedom}
 Let us denote by $P_0$ the projection onto the degenerate ground space of~$H_0$.  Since~$H_0$ is assumed to have a constant gap, a perturbation of the form~\eqref{eq:perturbedHamiltonian} effectively preserves the low-energy subspace~$P_0\cH$ for small~$\epsilon>0$, and generates a dynamics on this subspace according to an {\em  effective Hamiltonian} $\Heff(\epsilon)$. 
We will discuss natural definitions of this effective Hamiltonian in Section~\ref{sec:perturbativeeffectiveHamiltonians}.  For the purpose of this section, it suffices to mention that it is entirely supported on the ground space of~$H_0$, i.e., $\Heff(\epsilon)=P_0\Heff(\epsilon) P_0$. As such, it has spectral decomposition
\begin{align}
\Heff(\epsilon)&=\sum_{k=0}^{K-1} \Eeff_k(\epsilon) \Peff_k(\epsilon)\ , \label{eq:heffspectraldecomp}
\end{align}
where $\Eeff_0<\Eeff_1<\ldots $,  and where  $\Peff_k(\epsilon)=\Peff_k(\epsilon)P_0$ are commuting projections onto subspaces of the  ground space~$P_0\cH$ of $H_0$. (Generally, we expect $\Heff(\epsilon)$ to be non-degenerate such that $K=\dim P_0\cH$.) 
In particular, the effective Hamiltonian~\eqref{eq:heffspectraldecomp} gives rise to an orthogonal decomposition of the ground space~$P_0\cH$ by projections~$\{\Peff_k(\epsilon)\}^{K-1}_{k=0}$. 
States in~$\Peff_0(\epsilon)\cH$ are distinguished by having minimal energy. We can take the limiting projections as the perturbation strength goes to~$0$, setting
\begin{align*}
\Peff_k(0)&=\lim_{\epsilon\rightarrow 0} \Peff_k(\epsilon)\qquad\textrm{ for }k=0,\ldots,K-1\ .
\end{align*}
In particular, the effective Hamiltonian $\Heff(\epsilon)$ has ground space~$\Peff_0(0)\cH$ in the limit~$\epsilon\rightarrow 0$.   Studying $\Heff(\epsilon)$, and, in particular, the space~$\Peff_0(0)\cH$ appears to be of independent interest, as it determines how perturbations affect the topologically ordered ground space beyond spectral considerations as in~\cite{bonderson2009splitting}.

\subsection{Hamiltonian interpolation  and effective Hamiltonians\label{sec:adiabaticprepgroundstates}}

The connection to the interpolation process~\eqref{eq:hlinearinterpol}
is then given by the following conjecture. It is motivated by the discussion in Section~\ref{sec:smallsystemcase} and deals with the case where 
there are no conserved quantities (unlike, e.g., in the case of the Majorana chain, as discussed in Section~\ref{sec:majoranachain}). 
\begin{conjecture}\label{claim:targetstates}
Under suitable adiabaticity assumptions (see Section~\ref{sec:smallsystemcase}) the projection of the final state~$\Psi(T)$ onto the ground space of~$\Htop$ belongs to $\Peff_0(0)\cH$ (up to negligible errors\footnote{By negligible, we mean that the errors can be made to approach zero as $T$ is increased.}), i.e., it is a ground state of the effective  Hamiltonian~$\Heff(\epsilon)$ in the limit $\epsilon\rightarrow 0$.
 \end{conjecture}

In addition to the arguments in Section~\ref{sec:smallsystemcase}, we provide evidence for this conjecture by explicit examples, where we illustrate how~$\Peff_0(0)\cH$ can be computed analytically. We also verify that 
Conjecture~\ref{claim:targetstates} correctly determines the final states by  numerically studying the evolution~\eqref{eq:hlinearinterpol}.

We remark that the statement of Conjecture~\ref{claim:targetstates} severly constrains the states that can be prepared by Hamiltonian interpolation in the large $T$ limit: we will argue that the space~$\Peff_0(0)\cH$ has a certain robustness with respect to the choice of the initial Hamiltonian~$\Htriv$. 
In fact, the space~$\Peff_0(0)\cH$ is typically $1$-dimensional and spanned by a single vector~$\varphi_0$. 
Furthermore, this vector~$\varphi_0$ typically belongs to a finite family~$\mathcal{A}\subset P_0\cH$ of states  defined solely by~$\Htop$.  In particular,
 under Conjecture~\ref{claim:targetstates}, the dependence of the final state~$\Psi(T)$ on the Hamiltonian~$\Htriv$ is very limited: the choice of~$\Htriv$ only determines which of the states in~$\mathcal{A}$ is prepared.  We numerically verify that the resulting target states~$\Psi(T)$ indeed belong to the finite family~$\mathcal{A}$ of states obtained analytically. 
 
\subsection{Perturbative effective Hamiltonians\label{sec:perturbativeeffectiveHamiltonians}}
As discussed in Section~\ref{sec:adiabaticprepgroundstates}, 
we obtain distinguished final ground states by computation of suitable effective Hamiltonians~$\Heff(\epsilon)$, 
approximating the action of $H_0+\epsilon V$ on the ground space $P_0\cH$ of $H_0$. In many cases of interest, computing this effective Hamiltonian (whose definition for the Schrieffer-Wolff-case we present in Appendix~\ref{app:exactschriefferwolfdef}) exactly is infeasible (The effective Hamiltonian for the Majorana chain (see Section~\ref{sec:majoranachain}) is an exception.). 

Instead, we seek a perturbative expansion
\begin{align}
\Heff^{(n)}&= \sum_{n=0}^\infty \epsilon^n X_n\ \label{eq:heffnvexp}
\end{align}
in terms of powers of the perturbation strength~$\epsilon$. This is particularly natural as we are interested in the limit~$\epsilon\rightarrow 0$ anyway (see Conjecture~\ref{claim:targetstates}). Furthermore,  it turns out that such perturbative expansions provide insight into the physical mechanisms underlying the `selection' of particular ground states. 

We remark that there are several different methods for obtaining low-energy effective Hamiltonians. The {\em Schrieffer-Wolff  method}~\cite{schriefferwolff,bravyietal}  provides a unitary $U$ such that $\Heff=U(H_0+\epsilon V)U^\dagger$ preserves~$P_0\cH$ and can be regarded as an effective Hamiltonian. One systematically obtains a series expansion
\begin{align*}
S&=\sum_{n=1}^\infty \epsilon^n S_n\qquad\textrm{ where } S_n^\dagger =-S_n
\end{align*}
for the anti-Hermitian generator $S$ of $U=e^{S}$; this then naturally gives rise to an order-by-order expansion 
\begin{align}
\Heff^{(n) }&=H_0 P_0 +\epsilon P_0 V P_0+\sum_{q=2}^n \epsilon^q {\Heff}_{,q}\ .\label{eq:heffexpansion}
\end{align}
of the effective Hamiltonian, where $P_0$ is the projection onto the ground space $P_0\cH$ of $H_0$ (explicit expressions are given in Appendix~\ref{sec:seriesexpansion}).

Using the Schrieffer-Wolff method  has several distinct advantages, including the fact that
\begin{enumerate}[(i)]
\item
the resulting effective Hamiltonian $\Heff$, as well as the terms $\Heff^{(n)}$ are Hermitian, and hence have a clear physical interpretation. This is not the case e.g., for the Bloch expansion~\cite{bloch1958theorie}.
\item
There is no need to address certain self-consistency conditions arising e.g., when using the Dyson equation and corresponding self-energy methods~\cite{abrikosov1965quantum,fetter2003quantum}
\end{enumerate}
We point out that the
 series resulting by taking the limit $n\rightarrow\infty$ in~\eqref{eq:heffexpansion} has the usual convergence issues encountered in many-body physics:
convergence is guaranteed only if $\|\epsilon V\|\leq \Delta$, where~$\Delta$ is the gap of~$H_0$.
For a many-body system with extensive Hilbert space (e.g., $L$ spins), the norm $\|V\|=\Omega(L)$ is extensive  while the gap $\Delta=O(1)$ is constant, leading to convergence only in a regime where $\epsilon=O(1/L)$. In this respect, the Schrieffer-Wolff method does not provide direct advantages compared to other methods. As we are considering the limit~$\epsilon\rightarrow 0$, this is not an issue (also, for small systems as those considered in our numerics, we do not have such issues either).

We point out, however, that the results obtained by Bravyi et al.~\cite{bravyietal} suggest that 
considering partial sums of the form~\eqref{eq:heffexpansion} is meaningful even in cases in which the usual convergence guarantees are not given: indeed,~\cite[Theorem 3]{bravyietal} shows that 
the ground state energies of
$\Heff^{(n)}$ and $H_0+\epsilon V$ are approximately equal for suitable choices of~$\epsilon$ and~$n$.  Another key feature of the Schrieffer-Wolff method is  the fact that  the effective Hamiltonians~$\Heff^{(n)}$
are essentially local (for low orders~$n$) when the method is applied to certain many-body systems, see~\cite{bravyietal}. We will not need the corresponding results here, however.  

Unfortunately,  computing the Schrieffer-Wolff Hamiltonian~$\Heff^{(n)}$ generally involves a large amount of combinatorics (see~\cite{bravyietal} for a diagrammatic formalism for this purpose). In this respect, other methods may appear to be somewhat more accessible. Let us mention in particular the method involving the Dyson equation (and the so-called `self-energy' operator), which was used e.g., in~\cite[Section~5.1]{kitaev2006anyons} to compute $4$-th order effective Hamiltonians. This leads to remarkably simple expressions of the form
\begin{align}
P_0(VG)^{n-1}VP_0\label{eq:productGresolventseq}
\end{align}
for the $n$-th order term  effective Hamiltonian, where $G=G(E_0)$
is the resolvent operator
\begin{align}
G(z)&=(I-P_0)(zI-H_0)^{-1}(I-P_0)\label{eq:resolventoperatordef}
\end{align}
evaluated at the ground state energy~$E_0$ of $H_0$.  In general, though, the expression~\eqref{eq:productGresolventseq} only coincides with the Schrieffer-Wolff-method (that is,~\eqref{eq:heffexpansion}) up to the lowest non-trivial order.

\subsection{Perturbative effective Hamiltonians for topological order\label{sec:perturbativeeffectivetop}}
Here we identify simple conditions under which the  Schrieffer-Wolff Hamiltonian of lowest non-trivial order has the simple structure~\eqref{eq:productGresolventseq}. We will see that these conditions are satisfied for the systems we are interested in. In other words, for our purposes, the self-energy methods and the Schrieffer-Wolff method are equivalent.  While establishing this statement (see Theorem~\ref{thm:effectivehamiltoniaschrieffer} below) requires some work, this result vastly simplifies the subsequent analysis of concrete systems.

The condition we need is closely related to quantum error correction~\cite{Knill1997}.  
In fact, this condition has been identified as one of the requirements for topological quantum order (TQO-1) in Ref. \cite{bravyi2010topological}.
To motivate it, consider the case where $P_0\cH$ is an error-correcting code of distance~$L$. 
Then all operators~$T$ acting on less than $L$ particles\footnote{By particle we mean  a physical constituent qubit or qudit degree of freedom.} have trivial action on the code space, i.e., for such~$T$, the operator~$P_0TP_0$ is proportional to~$P_0$ (which we will write as $P_0TP_0\in\mathbb{C}P_0$). 
In particular, this means that if~$V$ is a Hermitian linear combination of single-particle operators, then $P_0V^nP_0\in\mathbb{C}P_0$ for all~$n<L$. 
The condition we need is a refinement of this error-correction criterion that incorporates energies (using the resolvent):

\begin{definition}\label{def:topologicalorderconditionL}
We say that the pair $(H_0,V)$ {\em satisfies the topological order condition with parameter~$L$} if 
$L$ is the smallest interger such that for all $n<L$, we have
\begin{align}
P_0VZ_1VZ_2\cdots Z_{n-1}VP_0\in \mathbb{C}P_0\label{eq:vpzerozjdef}
\end{align}
for all $Z_j\in \{P_0,Q_0\}\cup \{G^m\ |\ m\in\mathbb{N}\}$. Here $P_0$ is the ground space projection of~$H_0$, $Q_0=I-P_0$ is the projection onto the orthogonal complement, and $G=G(E_0)$ is the resolvent~\eqref{eq:resolventoperatordef} (supported on $Q_0\cH$). 
\end{definition}

We remark that this definition is easily verified in the systems we consider: if excitations in the system are local, the resolvent operators 
and projection in a product of the form~\eqref{eq:vpzerozjdef} can be replaced by local operators, and condition~\eqref{eq:vpzerozjdef} essentially reduces to a standard error correction condition for operators with local support. 

Assuming this definition, we then have the following result:
\begin{theorem}\label{thm:effectivehamiltoniaschrieffer}
Suppose that $(H_0,V)$ satisfies the topological order condition with parameter~$L$. 
Then the $n$-th order Schrieffer-Wolff effective Hamiltonian satisfies
\begin{align}
\Heff^{(n)}\in \mathbb{C}P_0\qquad\textrm{for all }n<L\ ,\label{eq:hefftrivialloworderx}
\end{align}
i.e., the effective Hamiltonian is trivial for these orders, and
\begin{align}
\Heff^{(L)}= P_0(VG)^{L-1}VP_0+\mathbb{C}P_0\ .\label{eq:hleffcomputedx}
\end{align}
\end{theorem}
We give the proof of this statement in Appendix~\ref{sec:equivalenceselfenergyschrieff}.

\newcommand*{\mV}{{\bf V}} 
\newcommand*{\mH}{{\bf H}} 
\newcommand*{\mA}{{\bf A}} 

\section{The Majorana chain\label{sec:majoranachain}}
In this section, we apply our general results to Kitaev's Majorana chain. We describe the model in Section~\ref{sec:modelmajorana}.
In Section~\ref{sec:stateprepinterpolationmajor}, we argue that
the interpolation process~\eqref{eq:adiabaticprep} is an instance of  symmetry-protected preparation; this allows us to identify the resulting final state.
We also observe that the effective Hamiltonian is essentially given by a `string'-operator~$F$, which happens to be the fermionic parity operator in this case. That is, up to a global energy shift, we have 
\begin{align}
\Heff\approx f\cdot  F
\end{align}
for a certain constant~$f$ depending on the choice of perturbation.  

\subsection{The model\label{sec:modelmajorana}}
Here we consider the case where~$\Htop$ is 
 Kitaev's Majorana chain~\cite{kitaev2001unpaired}, a system of spinless electrons confined to a line of $L$~sites.  In terms of  $2L$~Majorana operators~$\{c_p\}_{p=1}^{2L}$ satisfying the anticommutation relations
\begin{align*}
\{c_p,c_q\}&=2\delta_{p,q}\cdot I
\end{align*}
as well as $c_p^2=I$, $c_p^\dagger=c_p$, the Hamiltonian has the form
\begin{align}
\Htop &=\frac{i}{2}\sum_{j=1}^{L-1}c_{2j}c_{2j+1}\  .\label{eq:majoranachain}
\end{align}
Without loss of generality, we have chosen the normalization such that elementary excitations have unit energy. The Hamiltonian  has a two-fold degenerate ground space. The Majorana operators $c_1$ and $c_{2L}$ correspond to a complex boundary mode, and combine to form a Dirac fermion
\begin{align}
a&=\frac{1}{2}(c_1+ic_{2L})\ \label{eq:edgemodemajorana}
\end{align}
which commutes with the Hamiltonian. The operator $a^\dagger a$ hence provides a natural occupation number basis~$\{\ket{g_\sigma}\}_{\sigma\in \{0,1\}}$  for the ground space $P_0\cH$  defined  (up to arbitrary phases) by
\begin{align}
a^\dagger a\ket{g_\sigma}&=\sigma \ket{g_\sigma}\qquad\textrm{ for }\sigma\in \{0,1\}\ .\label{eq:groundstatesmajoranaunperturbed}
\end{align}
As a side remark, note that the states~$\ket{g_0}$ and $\ket{g_1}$ cannot be used directly to encode a qubit. 
This is because they have even and odd fermionic parity, respectively, and thus belong to different superselection sectors.
In other words, coherent superposition between different parity sectors are nonphysical.
This issue can be circumvented by using another fermion or a second chain, see~\cite{bravyi2012disorder}. 
Since the conclusions of the following discussion will be unchanged, we will neglect this detail for simplicity.

We remark that the Hamiltonian $\Htop$ of Eq.~\eqref{eq:majoranachain} 
belongs to a one-parameter family of extensively studied and well-understood quantum spin Hamiltonians. Indeed, the
Jordan-Wigner transform of 
the Hamiltonian  (with $g\in\mathbb{R}$ an arbitrary parameter)
\begin{align}\label{eq:transverseFieldMajorana}
H_{I,g} =  \frac{i}{2} \sum_{j=1}^{L-1} c_{2j}c_{2j{+}1}-\frac{g i}{2} \sum_{j=1}^{L}  c_{2j{-}1} c_{2j}.
\end{align}
is the transverse field Ising model
\begin{align}\label{eq:transverseFieldSpinIsing}
H'_{I,g} =- \frac{1}{2}\sum_{j=1}^{L-1} X_jX_{j+1}+\frac{g}{2}\sum_{j=1}^L Z_j 
\end{align} 
where $X_j$ and $Z_j$ are the spin $1/2$ Pauli matrices acting on qubit~$j$, $j=1,\ldots,L$.  This transformation allows analytically calculating the complete spectrum of the translation invariant chain for both periodic and open boundary conditions~\cite{Pfeuty1970}.

The Hamiltonian $H'_{I,g}$ has a quantum phase transition at $g=1$, for which the lowest energy modes in the periodic chain have an energy scaling as~$1/L$.   The open boundary case has been popularized by Kitaev as the Majorana chain and has a unique low energy mode $a$ (see Eq.~\eqref{eq:edgemodemajorana}) which has zero energy for $g = 0$ and for finite $0<g<1$, becomes a dressed mode with exponentially small energy (in $L$) and which is exponentially localized at the boundaries.

\subsection{State preparation by interpolation\label{sec:stateprepinterpolationmajor}}
The second term in~\eqref{eq:transverseFieldMajorana} may be taken to be the initial Hamiltonian~$\Htriv$ for the interpolation process. More generally, to prepare ground states of~$\Htop$, we may assume that 
our initial Hamiltonian is a quadratic Hamiltonian with a unique ground state. That is,~$\Htriv$ is of the form 
\begin{align}
\Htriv&=\frac{i}{4} \sum_{p,q=1}^{2L}\mV_{p,q}c_pc_q\ ,\label{eq:vdefinitionmajorana}
\end{align}
where $\mV$ is a real antisymmetric $2L\times 2L$~matrix. 
We will assume that it is bounded and local (with {\em range} $r$) in the sense that
\begin{align}
\|\mV\|\leq 1\qquad\textrm{ and }\qquad \mV_{p,q}=0 \textrm{ if } |p-q|>r\ ,\label{eq:normconditantisymmetricV}
\end{align}
where $\|\cdot\|$ denotes the operator norm. As shown in~\cite[Theorem~1]{bravyi2012disorder}, the Hamiltonian~$\Htop+\epsilon \Htriv$ has two lowest energy states with exponentially small energy difference, and this lowest-energy space remains separated from the rest of the spectrum by a constant gap for a fixed (constant) perturbation strength~$\epsilon>0$. Estimates on the gap along the complete path~$H(t)$ are, to the best of our knowledge, not known in this more general situation.

Let us assume that $\Psi(0)$ is the unique ground state of $\Htriv$ and consider the linear interpolation~\eqref{eq:adiabaticprep}. 
The corresponding process is an instance of the symmetry-protected preparation, i.e., Observation~\ref{obs:observationsymmetryprotec}  applies in this case. Indeed, the {\em fermionic parity operator}
\begin{align}
 F&=\prod_{j=1}^L (-i)c_{2j-1}c_{2j}\ ,\label{eq:stringoperatormajorana}
 \end{align}
 commutes with both $\Htriv$ and $\Htop$. Therefore, the  initial ground state~$\Psi(0)$ lies either in the even-parity sector, i.e., $F\Psi(0)=\Psi(0)$, or in the odd-parity sector ($F\Psi(0)=-\Psi(0)$).  (Even parity is usually assumed by convention, since the fermionic normal modes used to describe the system are chosen to have positive energy.) 
In any case, the $\pm 1$ eigenvalue of the initial ground state with respect to $F$ will persist throughout the full interpolation. This fixes the final state:
\begin{lemma}\label{lem:majoranainterpol}
Under suitable  adiabaticity assumptions (see Observation~\ref{obs:observationsymmetryprotec}), the resulting
 state in the evolution~\eqref{eq:adiabaticprep} is (up to a phase) given by the ground state~$\ket{g_0}$ or $\ket{g_1}$, depending on 
whether the initial ground state~$\Psi(0)$ lies in the even- or odd-parity  sector. 
 \end{lemma}
In particular, if $\Htriv=-\frac{gi}{2}\sum_{j=1}^Lc_{2j-1}c_{2j}$ is given by the second term in~\eqref{eq:transverseFieldMajorana}, we can apply the results of~\cite{Pfeuty1970}: 
the gap at the phase transition is associated with the lowest energy mode (which is not protected by symmetry) and is given by $ \lambda_{2}(H'_{I,g=1}) = 2\sin\left[\pi/(2L+1)\right]$. In other words, it is linearly decreasing in the system size~$L$. Therefore, the total evolution time~$T$ only needs to grow polynomially in the system size~$L$ for Hamiltonian interpolation to accurately follow the ground state space at the phase transition.  We conclude that translation-invariant Hamiltonian interpolation allows preparing the state~$\ket{g_0}$ in a time $T$ polynomial in the system size~$L$ and the desired approximation accuracy.
 
To achieve efficient preparation through Hamiltonian interpolation, one issue that must be taken into account is the effect of  disorder (possibly in the form of a random site-dependent chemical potential).
In the case where the system is already in the topologically ordered phase, a small amount of Hamiltonian disorder can enhance the zero temperature memory time of the Majorana chain Hamiltonian~\cite{bravyi2012disorder}.
This 1D~Anderson localization effect~\cite{Anderson1958}, while boosting memory times, was also found to hinder the convergence to the topological ground space through Hamiltonian interpolation. Indeed, in~\cite{Caneva2007} it was found that the residual energy density $[E_{\textrm{res}}(T)/L]_{\textrm{av}} \propto 1/\ln^{3.4}(T)$ averaged over disorder realizations decreases only polylogarithmically with the Hamiltonian interpolation time.
Such a slow convergence of the energy density indicates that in the presence of disorder, the time~$T$ required to accurately reach the ground space  scales exponentially with the system size~$L$. For this reason, translation-invariance (i.e., no disorder) is required for an efficient preparation, and this may be challenging in practice.

 We emphasize that according to Lemma~\ref{lem:majoranainterpol}, the prepared state is largely independent of the choice of the initial Hamiltonian~$\Htriv$ (amounting to a different choice of $\mV$): we do not obtain  a continuum of final states.
 As we will see below, this stability property appears in a similar form in other models. 
 The parity operator~\eqref{eq:stringoperatormajorana}, which should be thought of as a string-operator connecting the two ends of the wire, plays a particular role -- it is essentially the effective Hamiltonian which  determines the prepared ground state.

Indeed, the Schrieffer-Wolff-effective Hamiltonian can be computed {\em exactly} in this case, yielding
\begin{align}
\Heff(\epsilon)=\frac{E_0(\epsilon)}{2}I -\frac{\Delta(\epsilon)}{2}F\ ,\label{eq:exactmajoranahamiltonian} 
\end{align}
where $E_0(\epsilon)$ is the ground state energy of $\Htop+\epsilon\Htriv$, and 
$\Delta(\epsilon)=E_1(\epsilon)-E_0(\epsilon)$ is the gap.
Expression~\eqref{eq:exactmajoranahamiltonian} can be computed 
based on the variational expression~\eqref{eq:expressionvariationalSW} for the Schrieffer-Wolff transformation, using the fact that the ground space is two-dimensional and spanned by two states belonging to the even- and odd-parity sector, respectively. Note that the form~\eqref{eq:exactmajoranahamiltonian} can also be deduced (without the exact constants) from the easily verified fact (see e.g., Eq.~\eqref{eq:nonvariationalcharacterization}) that the Schrieffer-Wolff unitary~$U$ commutes with the fermionic parity operator~$F$, and thus the same is true for $\Heff(\epsilon)$. This expression illustrates that Conjecture~\ref{claim:targetstates} does not directly apply in the context of preserved quantities, as explained in Section~\ref{sec:adiabaticprepgroundstates}: rather, it is necessary to know the parity of the initial state~$\Psi(0)$ to identify the resulting final state~$\Psi(T)$ in the interpolation process.

\newcommand*{\figannihilate}{\includegraphics[scale=0.33]{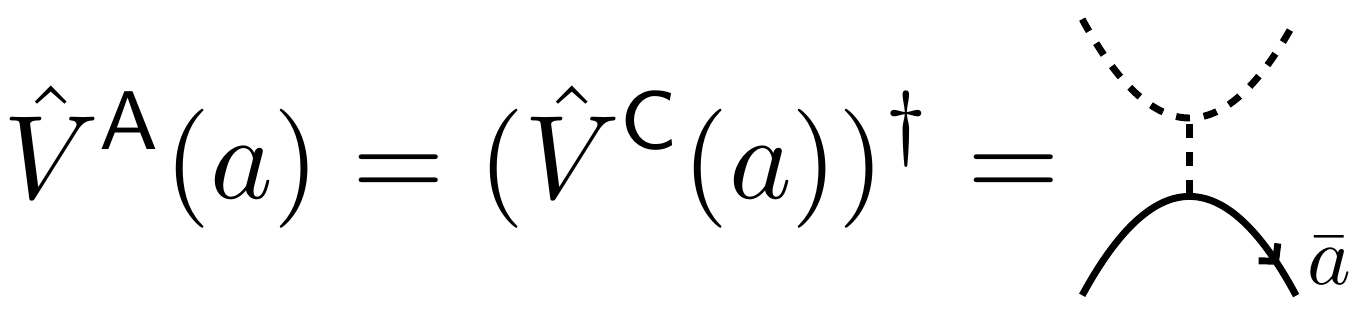}}
\newcommand*{\figantiparticle}{\includegraphics[scale=0.45]{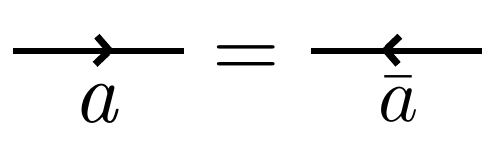}}
\newcommand*{\figbra}{\includegraphics[scale=0.4]{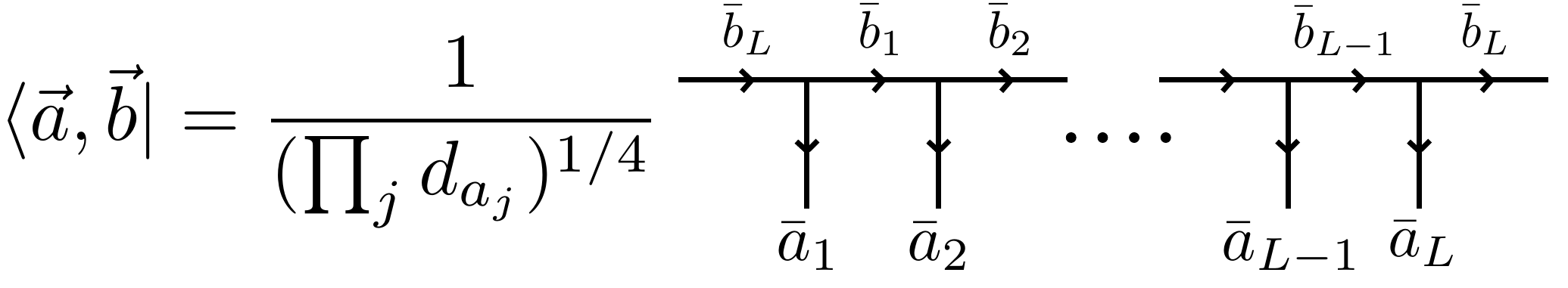}}
\newcommand*{\figket}{\includegraphics[scale=0.4]{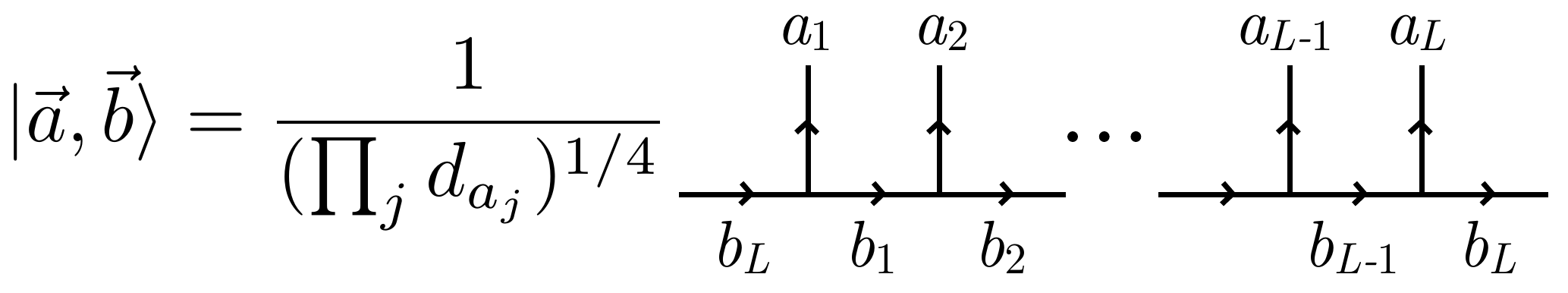}}
\newcommand*{\figbraidone}{\includegraphics[width=1cm]{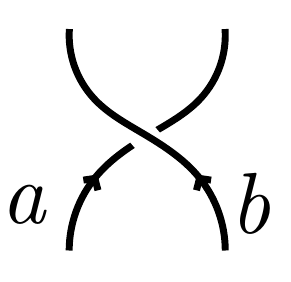}}
\newcommand*{\figbraidtwo}{\includegraphics[width=1cm]{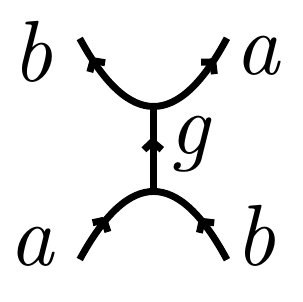}}
\newcommand*{\figbubble}{\includegraphics[scale=0.4]{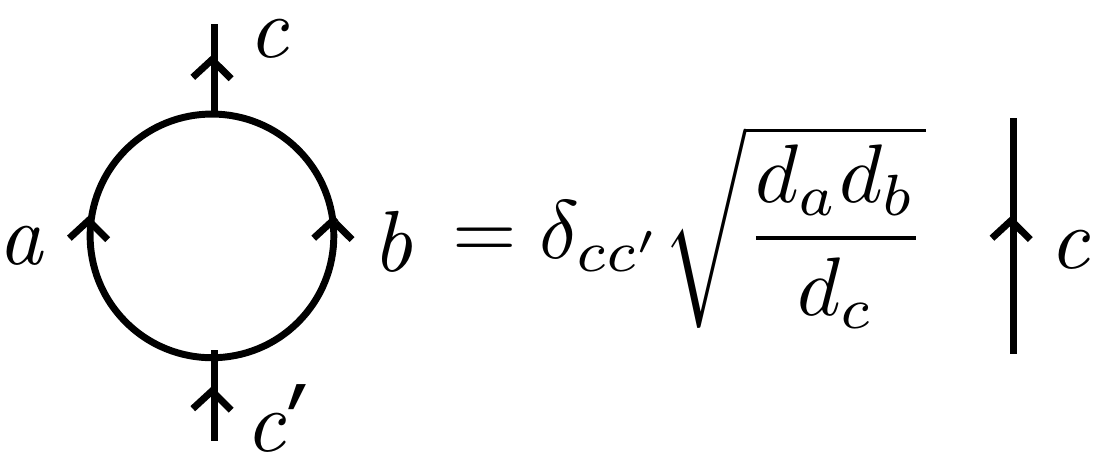}}
\newcommand*{\figcompositeone}{\includegraphics[width=5cm]{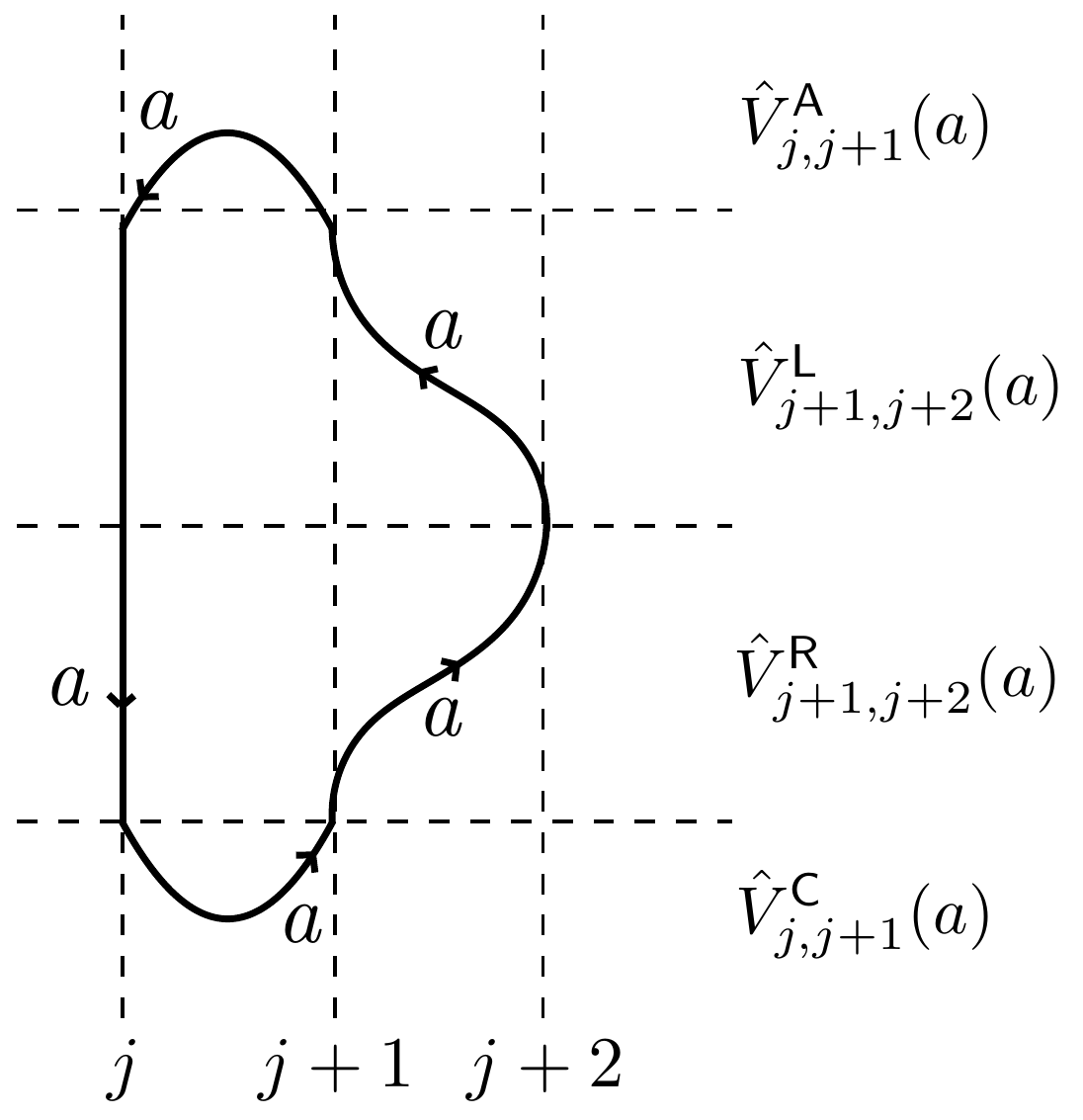}}
\newcommand*{\figcreate}{\includegraphics[scale=0.33]{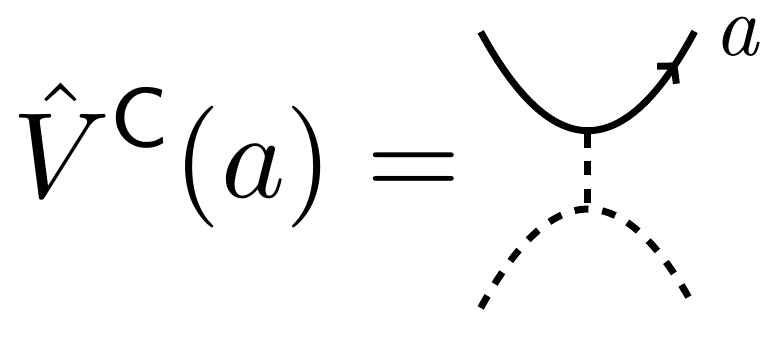}}
\newcommand*{\figfmove}{\includegraphics[scale=0.4]{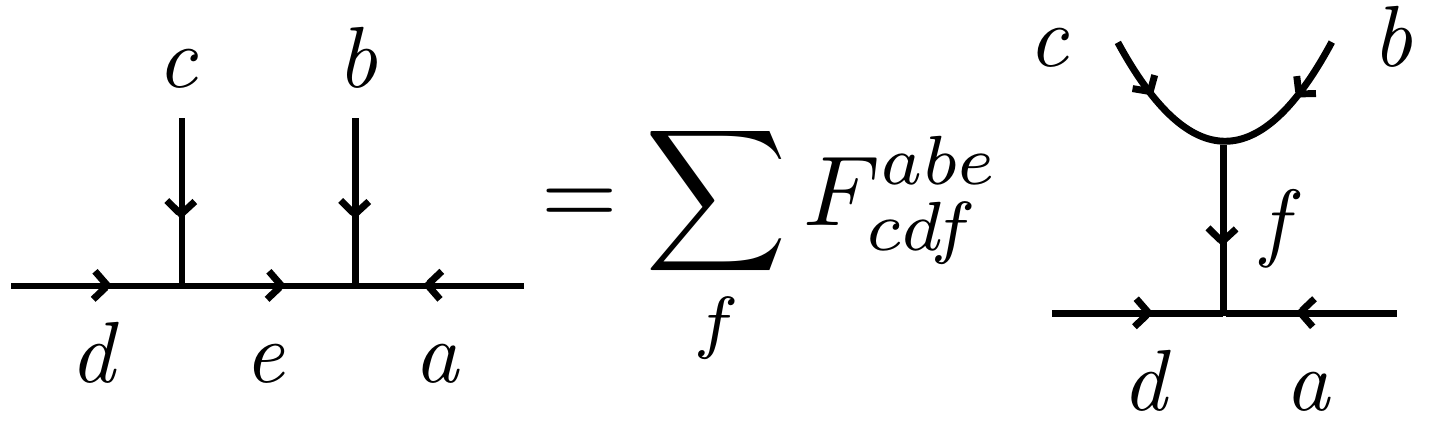}}
\newcommand*{\figfuse}{\includegraphics[scale=0.4]{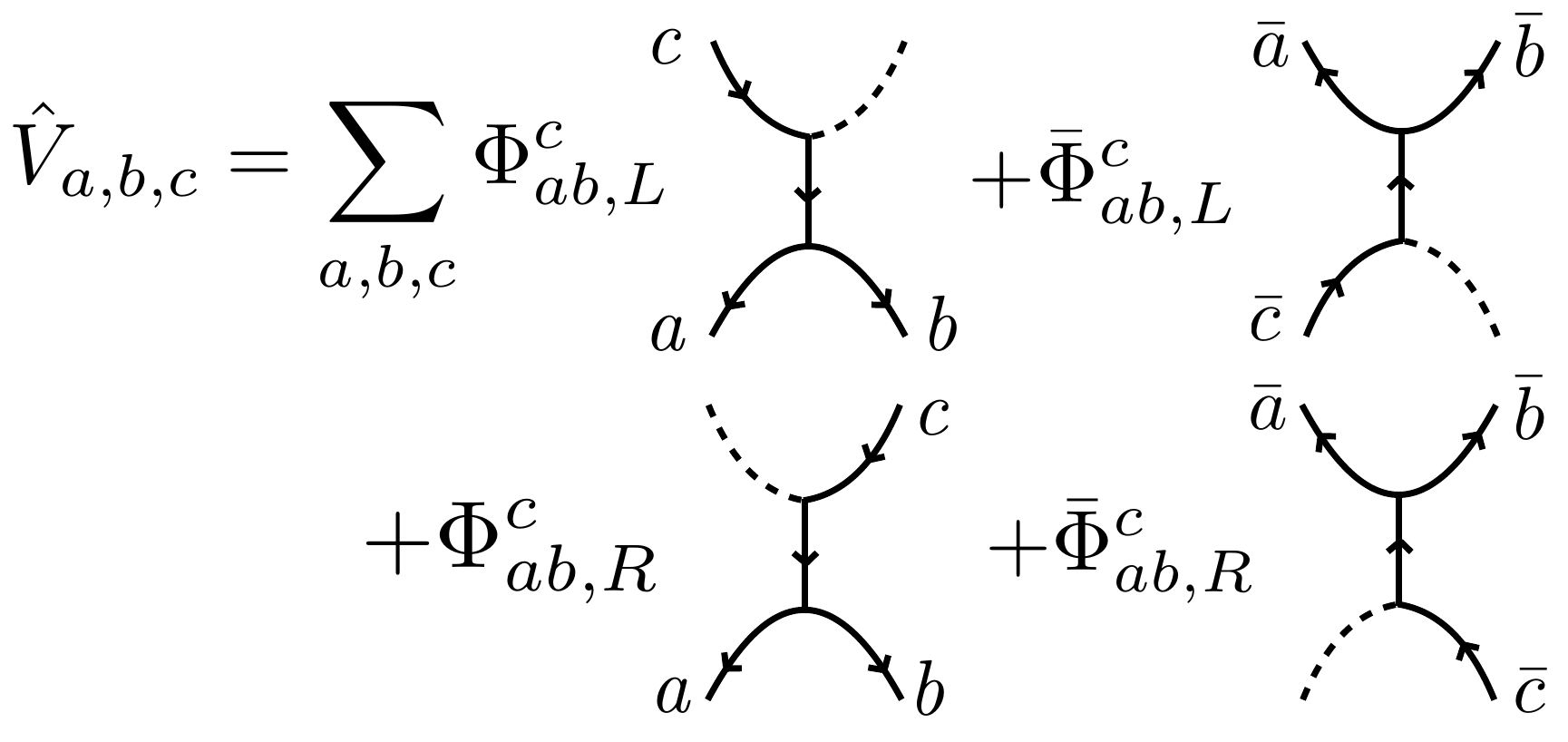}}
\newcommand*{\figmovel}{\includegraphics[scale=0.33]{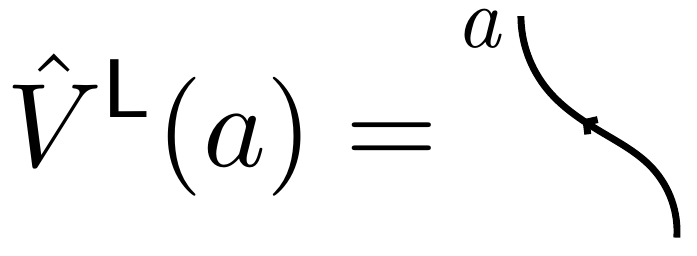}}
\newcommand*{\figmover}{\includegraphics[scale=0.33]{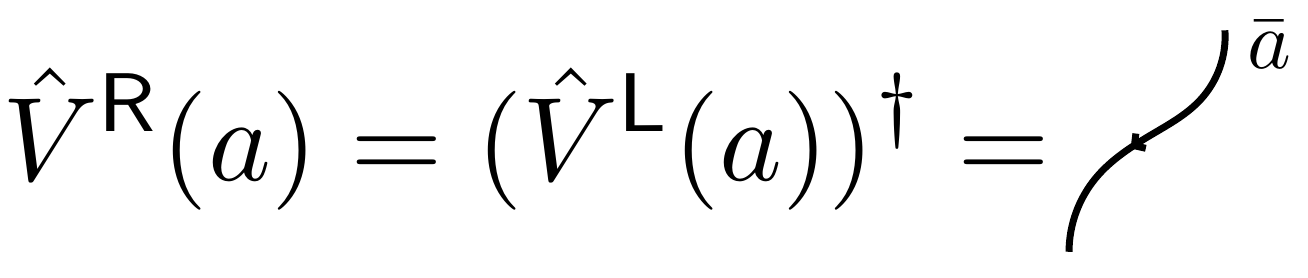}}
\newcommand*{\figpair}{\includegraphics[width=2cm]{./figures/figpair.pdf}}
\newcommand*{\figsinglesitehamiltonian}{\includegraphics[width=2.6cm]{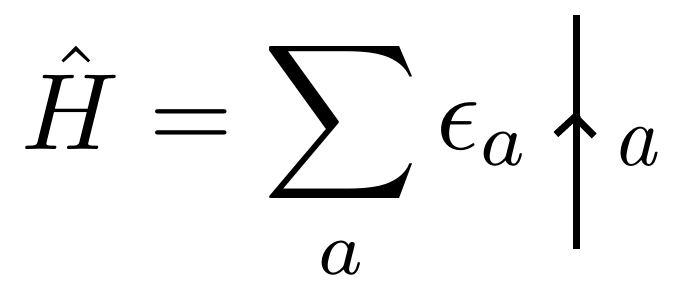}}
\newcommand*{\figspr}{\includegraphics[scale=0.45]{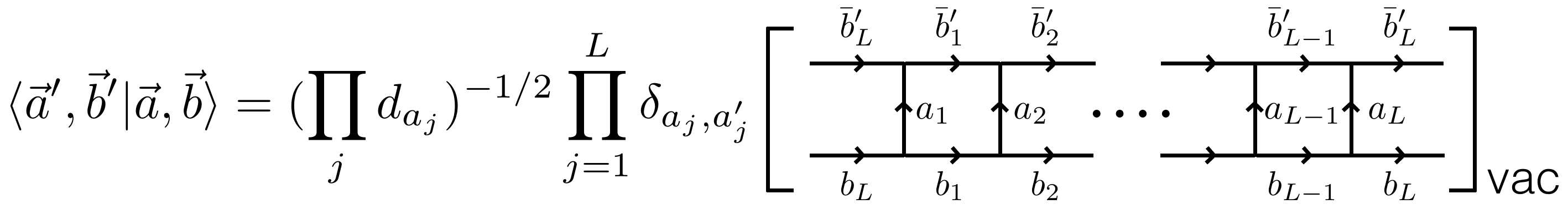}}
\newcommand*{\figstringop}{\includegraphics[width=5cm]{./figures/figstringop.pdf}}
\newcommand*{\figtunnel}{\includegraphics[width=5cm]{./figures/figtunnel.pdf}}
\newcommand*{\figtwolocalaction}{\includegraphics[height=2cm]{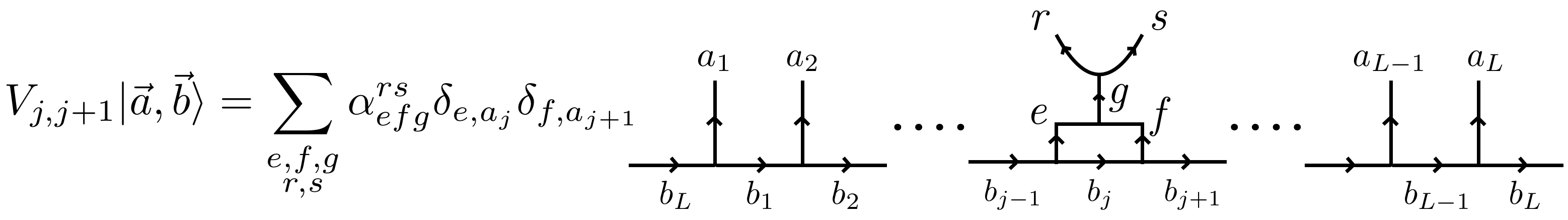}}
\newcommand*{\figtwolocaloperator}{\includegraphics[width=4cm]{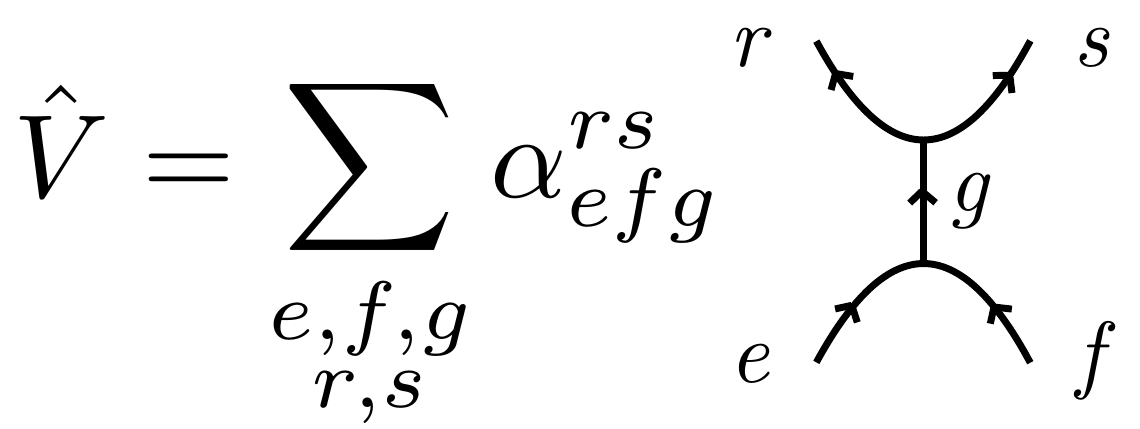}}
\newcommand*{\figvacuum}{\includegraphics[scale=0.45]{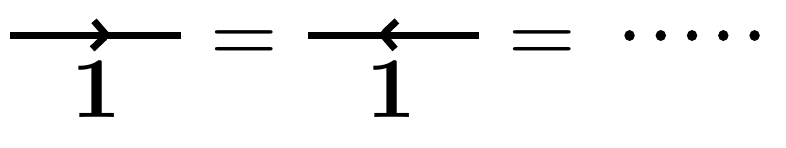}}


 \section{General anyon chains}\label{sec:anyonchains}
 In this section, we generalize the considerations related to the Majorana chain to more general anyonic systems. Specifically, we consider a $1$-dimensional lattice of anyons with periodic boundary conditions. 
 This choice retains many features from the Majorana chain such as locally conserved charges and topological degeneracy yet further elucidates some of the general properties involved in the perturbative lifting of the topological degeneracy. 
 
 In Section~\ref{sec:backgroundanyon}, we review the description of effective models for topologically ordered systems. 
 A key feature of these models is the existence of a family~$\{F_a\}_a$ of string-operators indexed by particle labels. 
 Physically, the operators~$F_a$ correspond to the process of creating a particle-antiparticle pair~$(a,\bar{a})$, tunneling along the $1$-dimensional (periodic) lattice, and subsequent fusion of the pair to the vacuum (see Section~\ref{sec:stringoperators}). These operators play a fundamental role in distinguishing different ground states. 
 
 In Section~\ref{sec:perturbationtheoryanyon}, we derive our main result concerning these models. We consider local translation-invariant perturbations to the Hamiltonian of such a model, and show that the effective Hamiltonian is a linear combination of string-operators, i.e., 
 \begin{align}
 \Heff \approx \sum_{a}f_a F_a\label{eq:heffanyonbasis}
 \end{align}
  up to an irrelevant global energy shift. 
  The coefficients $\{f_a\}_a$ are determined by the perturbation. 
  They can be expressed in terms of a certain sum of diagrams, as we explain below.
While not essential for our argument, translation-invariance allows us to simplify the parameter dependence when expressing the coefficients $f_a$ and may also be important for avoiding the proliferation of small gaps.
  
We emphasize that the effective Hamiltonian has the form~\eqref{eq:heffanyonbasis} independently of the choice of perturbation. The operators~$\{F_a\}_a$ are mutually commuting, and thus  have a distinguished simultaneous eigenbasis (we give explicit expressions for the latter in Section~\ref{sec:stringoperators}). The effective Hamiltonian~\eqref{eq:heffanyonbasis} is therefore diagonal in a fixed basis irrespective of the considered perturbation. Together with
 the general reasoning for Conjecture~\ref{claim:targetstates}, this suggests that Hamiltonian interpolation can only prepare a discrete family of different ground states in these anyonic systems.
 
 In Section~\ref{sec:twodimensionalsystems}, we consider two-dimensional topologically ordered systems and find effective Hamiltonians analogous to~\eqref{eq:heffanyonbasis}. We will also show numerically that Hamiltonian interpolation indeed prepares corresponding ground states.

\subsection{Background on anyon chains\label{sec:backgroundanyon}}
The models we consider here describe effective degrees of freedom of a topologically ordered system. Concretely, we consider  one-dimensional chains with periodic boundary conditions,  where anyonic excitations may be created/destroyed on~$L$ sites, and may hop between neighboring sites. 
Topologically (that is, the language of topological quantum field theory), the system can be thought of as a torus with $L$ punctures aligned along one fundamental cycle. 
Physically, this means that excitations are confined to move exclusively along this cycle (we will consider more general models in section \ref{sec:twodimensionalsystems}). 
A well-known example of such a model is the Fibonacci golden chain~\cite{feiguin2007interacting}.  Variational methods for their study were developed in~\cite{Pfeifer2010,konig2010anyonic}, which also provide a detailed introduction to the necessary formalism. In this section, we establish notation for anyon models and review minimal background to make the rest of the paper self-contained.

\subsubsection{Algebraic data of anyon models: modular tensor categories}
\label{sec_modular_tensor_cat}
Let us briefly describe  the algebraic data defining an  anyon model. 
The underlying mathematical object is a tensor category. 
This specifies among other things:
\begin{enumerate}[(i)]
\item
A finite set of particle labels $\Anyons = \{1,a,\ldots\}$ together
with an involution $a\mapsto \bar{a}$ (called particle-anti-particle exchange/charge conjugation). 
There is a distinguished particle $1=\bar{1}$ called the trivial or vacuum particle.
\item
A collection of integers  $N_{ab}^{c}$ indexed by particle labels, specifying the so-called {\em fusion multiplicities} (as well as the fusion rules). 
For simplicity, we will only consider the multiplicity-free case, where~$N^c_{ab}\in \{0,1\}$ (this captures many models of interest).  In this case, we will write~$N^c_{ab}=\delta_{ab\bar{c}}$. 
\item
A $6$-index tensor $F:\Anyons^6 \rightarrow \Complex$ (indexed by particle labels)  $F^{abe}_{cdf}$ which is unitary with respect to the rightmost two indices ($e,f$) and can be interpreted as a change of basis for fusion trees. 
\item
A positive scalar $d_a$ for every particle label~$a$, called the {\em quantum dimension}.
\item
A unitary, symmetric matrix $S_{ij}$ indexed by particle labels such that $S_{\bar{i}j}=\overline{S_{ij}}$. 
\item
A {\em topological phase} $e^{i\theta_j}$, $\theta_j\in\mathbb{R}$,  associated with each particle~$j$. We usually collect these into a diagonal matrix~$T=\mathsf{diag}(\{e^{i\theta_j}\}_j)$; the latter describes the action of 
a twist in the mapping class group representation associated with the torus (see Section~\ref{sec:stringoperatorstqft}). 
\end{enumerate}

A list of the algebraic equations satisfied by these objects can be found e.g., in ~\cite{levin2005string} (also see~\cite{nayak2008non,levin2005string,kitaev2006anyons,wang} for more details). Explicit examples of such tensor categories can also be found in~\cite{levin2005string}, some of which we discuss in Section~\ref{sec:Short_Introduction_LW}.

Here we mention just a few which will be important in what follows: the fusion rules~$\delta_{ijk}$ are symmetric under permutations of~$(i,j,k)$. They satisfy
\begin{align}
\sum_{m}\delta_{ij\bar{m}}\delta_{mk\bar{\ell}}&=\sum_{m}\delta_{jk\bar{m}}\delta_{im\bar{\ell}}\label{eq:associativityfusion}
\end{align}
which expresses the fact that fusion (as explained below) is associative,
 as well as
\begin{align}
\delta_{i\bar{j}1}&=\delta_{ij}=\begin{cases}
1\qquad&\textrm{if }i=j\\
0  &\textrm{otherwise}\ .
\end{cases}\label{eq:deltapropertyfusionrule}
\end{align}
Some of the entries of the tensor $F$ are determined by the fusion rules and the quantum dimensions, that is, 
\begin{align}
F^{i\bar{i}1}_{\bar{j}jk}&=\sqrt{\frac{d_k}{d_id_j}}\delta_{ijk}\ .  \label{eq:ftensornormalisation}
\end{align}
Another important property is the Verlinde formula
\begin{align}
\delta_{bc\bar{d}}=N^{d}_{bc}&=\sum_a \frac{S_{ba}S_{ca}S_{\bar{d}a}}{S_{1a}}\ ,\label{eq:verlindeformulasmatrix}
\end{align}
which is often summarized by stating that $S$ ``diagonalizes the fusion rules''.

\subsubsection{The Hilbert space}
The Hilbert space of a one-dimensional periodic chain of $L$ anyons is the space associated by a TQFT to a torus with punctures.
It has the form
\begin{align}
\cH\cong\bigoplus_{\substack{a_1,\ldots,a_L\\
b_0,\ldots,b_L}} V_{b_0}^{a_1 b_1}\otimes V_{b_1}^{a_2b_2}\otimes\cdots\otimes V^{a_Lb_L}_{b_{L-1}}\ ,
\end{align}
where the indices $a_j,b_k$ are particle labels, $V_c^{ab}$ are the associated finite-dimensional fusion spaces and we identify $b_0 =b_L$. 
The latter have dimension~$\dim V_c^{ab}=N^c_{ab}$. 
Again, we will  focus on the multiplicity-free case where  
$N^c_{ab}=\delta_{ab\bar{c}}\in \{0,1\}$.  
In this case, we can give an orthonormal basis~$\{\ket{\vec{a},\vec{b}}\}_{(\vec{a},\vec{b})}$ of $\cH$ in terms of `fusion-tree' diagrams, i.e, 
\begin{align}
\figket \label{eq:abvecbasisc}
\end{align}
where $\vec{a}=(a_1,\ldots,a_L)$ and $\vec{b}=(b_1,\ldots,b_L)$ 
have to satisfy the fusion rules at each vertex, i.e., 
$\dim V_{b_{j-1}}^{a_jb_j}=\delta_{a_jb_j\bar{b}_{j-1}}=1$ for all $j=1,\ldots,L$.

 The prefactor in the definition of the state~\eqref{eq:abvecbasisc} involves the quantum dimensions of the particles, and is chosen in such a way that~$\{\ket{\vec{a},\vec{b}}\}$ is an orthonormal basis with respect to the  inner product
defined in terms of the isotopy-invariant calculus of diagrams: 
the adjoint of $\ket{\vec{a},\vec{b}}$ is represented as
\begin{align}
\figbra\ .
\end{align}

\subsubsection{Inner products and diagramatic reduction rules}
Inner products are evaluated by composing diagrams and then reducing, i.e., 
\begin{align}
\figspr\  \label{eq:innerproductanyondefinition}
\end{align}
where $[\cdot ]_{\textsf{vac}}$ is the coefficient of the empty diagram when reducing. Reduction is defined in terms of certain local moves. These include 
\begin{enumerate}[(i)]
\item
reversal of arrows (together particle-antiparticle involution $a\mapsto \bar{a}$)
\begin{align}
\figantiparticle\ .
\end{align}
\item
(arbitrary) insertions/removals  of lines labeled by the trivial particle~$1$. 
Since $\bar{1}=1$, such lines are not directed, and will often be represented by dotted lines or omitted altogether,
\begin{align}
\figvacuum\ .
\end{align}
\item
application of the  $F$-matrix in the form
\begin{align}\label{eq:Fmove}
\figfmove
\end{align}
which leads to a formal  linear combination  of diagrams
where subgraphs are replaced locally by the figure on the rhs.
\item
removal of ``bubbles'' by  the substitution rule
\begin{align}
\figbubble\ .\label{eq:bubbleremoval}
\end{align}
\end{enumerate}
These reduction moves  can be applied iteratively in arbitrary order to yield superpositions of diagrams. An important example of this computation is the following:
\begin{align}
\includegraphics[scale=0.4]{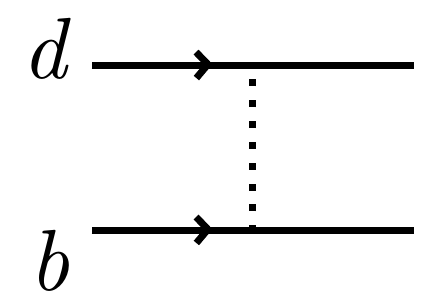} & \includegraphics[scale=0.4]{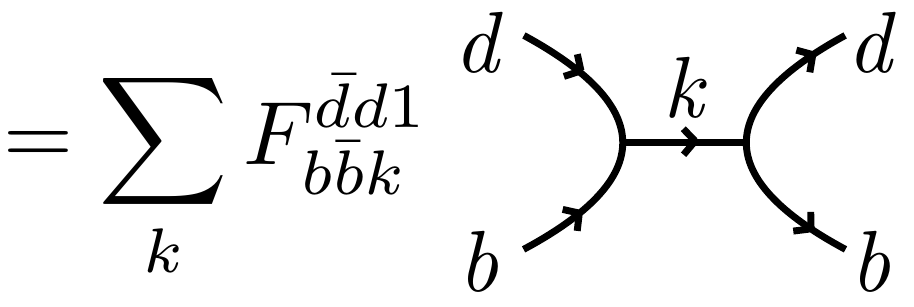}\\
& \includegraphics[scale=0.4]{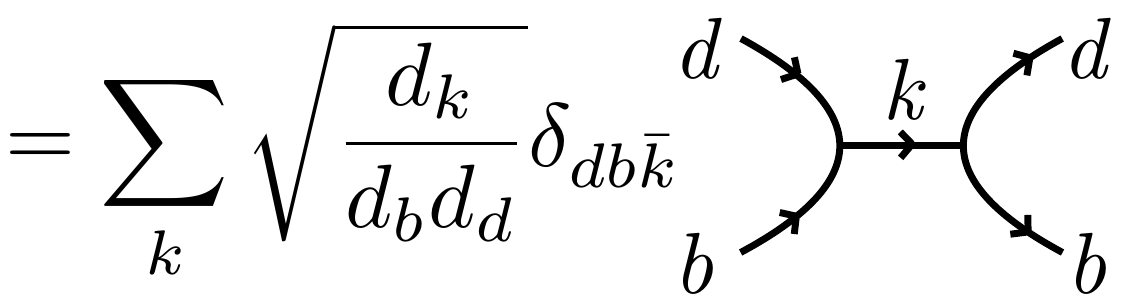}\\
& \includegraphics[scale=0.4]{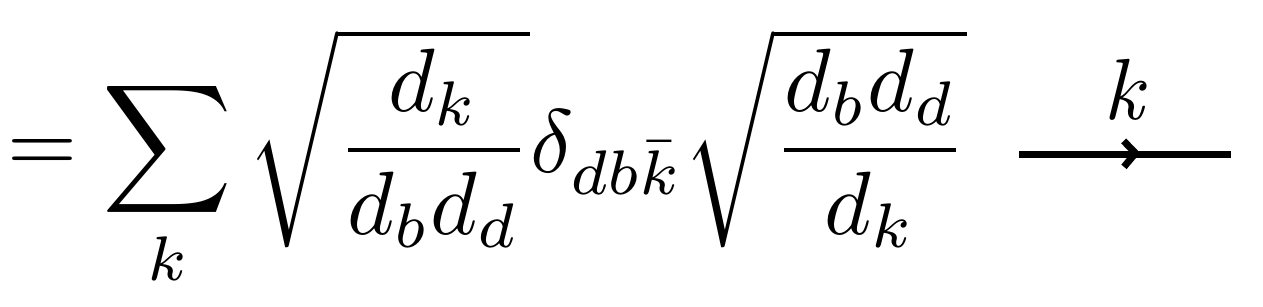}\\
& \includegraphics[scale=0.4]{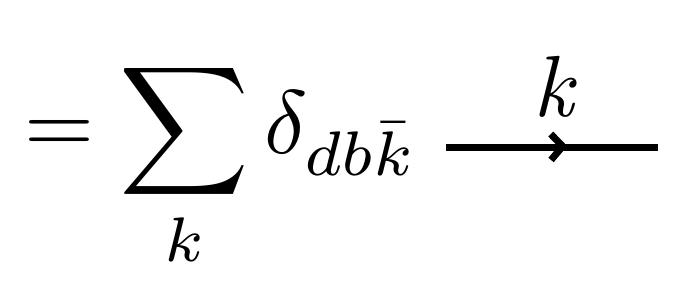} \ .
\label{eq_Fstring_multi}
\end{align}
The series of steps first makes use of an $F$-move \eqref{eq:Fmove}, followed by Eq.~\eqref{eq:ftensornormalisation} as well as~\eqref{eq:bubbleremoval}. 
Together with property~\eqref{eq:deltapropertyfusionrule} and evaluation of the inner product~\eqref{eq:innerproductanyondefinition}, this particular calculation shows that the flux-eigenstates~\eqref{eq:basisstandardanyon} are mutually orthogonal. We refer to~\cite{levin2005string} for more details. 

\subsubsection{Local operators\label{sec:anyonlocalops}}
Operators are also defined by diagrams, and are applied to vectors/multiplied by
stacking (attaching) diagrams on top of the latter.  
Expressions vanish unless all attachment points have identical direction and labels. 
Here we concentrate on $1$- and $2$-local operators, although the generalization is straightforward (see~\cite{konig2010anyonic, bonderson2009splitting}). 

A single-site operator  $\hat{H}$ is determined by coefficients $\{\epsilon_a\}_a$ and represented at
\begin{align}
\figsinglesitehamiltonian\ .\label{eq:singlesitehamiltoniananyon}
\end{align}
It acts diagonally in the fusion tree basis, i.e., writing $H_j$ for the operator~$\hat{H}$ applied to site~$j$, we have
\begin{align}
H_j\ket{\vec{a},\vec{b}}&=\epsilon_{a_j}\ket{\vec{a},\vec{b}}\ .
\end{align}
A two-site operator~$\hat{V}$ acting on two neighboring sites  is determined by 
a tensor $\{\alpha^{rs}_{efg}\}_{r,s,e,f,g}$ (where the labels have to satisfy appropriate fusion rules) via the linear combinations of diagrams
\begin{align}
\figtwolocaloperator\ .\label{eq:vtwolocaloperator}
\end{align}
When applied to sites $j$ and $j+1$ it acts as
\begin{align}
\figtwolocalaction\ ,
\end{align}
where the rhs. specifies a vector in~$\cH$ in terms of  the reduction rules. It will be convenient in the following to distinguish between linear combinations of the form~\eqref{eq:vtwolocaloperator} and
operators which are scalar multiplies of a single diagram (i.e., with only one non-zero coefficient~$\alpha^{rs}_{efg}$). We call the latter kind of two-site operator {\em elementary}. 
 
We can classify the terms appearing in~\eqref{eq:vtwolocaloperator} according to the different physical processes they represent: in particular, we
have pair creation- and annihilation operators
\begin{align}
\raisebox{-3mm}{\figcreate}\qquad\textrm{ and }\qquad \raisebox{-3mm}{\figannihilate}\ ,
\end{align}
simultaneous annihilation- and creation operators
\begin{align}
\hat{V}^{\textsf{CA}}(a,b)=\hat{V}^{\textsf{C}}(a)\hat{V}^{\textsf{A}}(b)
\end{align}
left- and right-moving `propagation' terms
\begin{align}
\raisebox{-3mm}{\figmovel}\qquad\textrm{ and }\qquad \raisebox{-3mm}{\figmover}
\end{align}
    as well as more general fusion operators such as e.g.,
\begin{align}
\figfuse\ ,
\end{align}
(We are intentionally writing down a linear combination here.)
Note that a general operator of the form
\eqref{eq:vtwolocaloperator} also involves braiding processes since
\raisebox{-3mm}{\figbraidone}
can be resolved to diagrams of the form \raisebox{-4mm}{\figbraidtwo} using the $R$-matrix (another object specified by the tensor category). We will consider composite processes composed of such two-local operators in Section~\ref{sec:productlocaloperatorslogical}.

\subsubsection{Ground states of anyonic chains\label{sec:groundstatesanyonicchains}}
We will consider  translation-invariant Hamiltonians
$H_0=\sum_j \hat{H}_j$ with local terms of the form
\begin{align}
\raisebox{-5mm}{\figsinglesitehamiltonian}\qquad\textrm{ with }\epsilon_a>0
\textrm{ for }a\neq 1 \textrm{ and }\epsilon_1=0\ .\label{eq:unperturbedhamiltoniananyon}
\end{align}
Such a Hamiltonian~$H_0$ corresponds to an on-site potential for anyonic excitations, where a particle of type~$a$ has associated energy~$\epsilon_a$ independently of the site~$j$. We denote the projection onto the ground space of this Hamiltonian by~$P_0$. This is the space
\begin{align}
P_0\cH=\mathsf{span}\{\ket{\vec{1},b\cdot \vec{1}}\ |\ b\textrm{ particle label}\}\label{eq:groundspaceanyonicchain}
\end{align}
where $\vec{1}=(1,\ldots,1)$ and $b\cdot\vec{1}=(b,\ldots,b)$. 
In other words, the ground space of $H_0$ is degenerate, with degeneracy equal to the number of particle labels.

It will be convenient to use the basis~$\{\ket{b}\}_b$ of the ground space consisting of the `flux' eigenstates
\begin{align}
\ket{b}=\ket{\vec{1},b\cdot\vec{1}}\ .\label{eq:basisstandardanyon}
\end{align}
In addition, we can define a dual basis~$\{\ket{b'}\}_b$ of the ground space using the $S$-matrix. 
The two bases are related by 
\begin{align}
\ket{a'}=\sum_{b}\overline{S_{ba}}\ket{b}\ \label{eq:dualbasissmatrixv}
\end{align}
for all particle labels $a,b$. 

As we discuss in Section~\ref{sec:Short_Introduction_LW}, in the case of two-dimensional systems, the dual basis~\eqref{eq:dualbasissmatrixv} is simply the basis of flux eigenstates with respect to a `conjugate' cycle. While  this interpretation does not directly apply in this $1$-dimensional context, the basis~$\{\ket{a'}\}_{a}$ is nevertheless well-defined and important (see Eq.~\eqref{eq:stringoperatorsdiagonalanyon}).

\subsubsection{Non-local string-operators\label{sec:stringoperators}}
In the following, certain non-local operators, so-called {\em string-operators},
will play a special role. Strictly speaking, these are only defined  on the
subspace~\eqref{eq:groundspaceanyonicchain}. However, we will see in Section~\ref{sec:perturbationtheoryanyon} that they arise naturally from certain non-local operators.

The string-operators $\{F_a\}_a$ are indexed by particle labels~$a$. In terms of the basis~\eqref{eq:basisstandardanyon} of the ground space~$P_0\cH$ of~$H_0$, the action of $F_a$ is given in terms of the fusion rules as
\begin{align}
F_a\ket{b}&=\sum_c N_{a b}^{c}\ket{c}=\sum_c \delta_{ab\bar{c}}\ket{c}\ .\label{eq:anyonfusionrulesfab}
\end{align}
\footnote{In fact, the operators $\{F_a\}_a$ form a representation of the Verlinde algebra, although we will not use this fact here.}
The operator $F_a$ has the interpretation of creating a particle-antiparticle pair~$(a,\bar{a})$, moving one around the torus, and then fusing to vacuum. 
For later reference, we 
show that every string-operator~$F_a$ is diagonal in the dual basis~$\{\ket{a'}\}$. Explicitly, we have 
\begin{align}
F_b P_0 =\sum_a \frac{S_{ba}}{S_{1a}}\proj{a'}\ .\label{eq:stringoperatorsdiagonalanyon}
\end{align}
\begin{proof}
We first expand $P_0$ into its span and $F_b$ according to eq. \eqref{eq:anyonfusionrulesfab}, followed by an expansion of $N_{bc}^{d}$ through the Verlinde formula \eqref{eq:verlindeformulasmatrix}.
Finally, we use the unitarity and symmetry of $S$ to transform bra and ket factors into the dual basis given by Eq.~\eqref{eq:dualbasissmatrixv}
\begin{align}
F_b P_0 =\sum_{c,d}N^{d}_{bc}\ket{d}\bra{c}=\sum_{a}\frac{S_{ba}}{S_{1a}}\sum_{c,d}S_{ca}{S_{\bar{d}a}}\ket{d}\bra{c}=\sum_a \frac{S_{ba}}{S_{1a}}\proj{a'}\ .
\end{align}
\end{proof}

\subsubsection{Products of local operators and their logical action\label{sec:productlocaloperatorslogical}}
Operators preserving the ground space~$P_0\cH$ (cf.~\eqref{eq:basisstandardanyon}) are called {\em logical operators}. As discussed in Section~\ref{sec:stringoperators},  string-operators~$\{F_a\}$ are an example of such logical operators.  Clearly, because they can  simultaneously be diagonalized (cf.~\eqref{eq:stringoperatorsdiagonalanyon}), they do not generate the full algebra of logical operators. 
Nevertheless, they span the set of logical operators that are 
generated by geometrically local physical processes preserving the space~$P_0\cH$.

That is, if $O=\sum_{j}\prod_{k}V_{j,k}$ is a linear combinations of products of local operators $V_{j,k}$, then its restriction to the ground space is of the form
\begin{align}
P_0OP_0&=\sum_{a}o_a F_a\  ,\label{eq:localoperatorsuperselectionruleanyon}
\end{align}
i.e., it is a linear combination of string operators (with some coefficients $o_a$). Eq.~\eqref{eq:localoperatorsuperselectionruleanyon} can be interpreted as an emergent superselection rule for topological charge,
which can be seen as the generalization of the parity superselection observed for the Majorana chain.
It follows directly from the diagrammatic formalism for local operators.

To illustrate this point (and motivate the following computation), let us consider three examples of such operators, shown in Figures~\ref{fig:exampleprocessa},~\ref{fig:exampleprocessb}  and~\ref{fig:exampleprocessc}.  
\begin{description}
\item[$O_1=\hat{V}^\textsf{A}_{j-1,j}(a)\hat{V}^\textsf{L}_{j+1,j+2}(a)\hat{V}^\textsf{R}_{j+1,j+2}(a)\hat{V}^\textsf{C}(a)_{j,j+1}$:]
This processes has trivial action on the ground space: it is entirely local. It has action $P_0O_1P_0=d_a P_0$, where the proportionality constant~$d_a$ results from Eq.~\eqref{eq:bubbleremoval}. 

\item[$O_2=\hat{V}^\textsf{L}_{j-1,j}(\bar{a})\hat{V}^\textsf{R}_{j,j+1}(a)\hat{V}^\textsf{C}_{j,j+1}(a)$: ]  This process creates a particles anti-particle pair $(a,\bar{a})$ and further separates these particles. 
Since the operator maps ground states to excited states, we have $P_0O_3P_0=0$.

\item[$O_3=\hat{V}^\textsf{A}(\bar{a})_{N,1} \hat{V}^\textsf{R}(a)_{N-1,N} \ldots \hat{V}^\textsf{R}(a)_{3,4}\hat{V}^\textsf{R}(a)_{2,3}\hat{V}^\textsf{C}(a)_{1,2}$:]
This process involves the creation of a pair of particles $(a,\bar{a})$, with subsequent propagation and annihilation. Its logical action is $P_0O_2P_0=F_a$ is given by the string-operator~$F_a$, by a computation similar to that of~\eqref{eq_Fstring_multi}.

\end{description}
\begin{figure}
\begin{subfigure}{0.45\textwidth}
\begin{center}
\includegraphics[width=0.8\textwidth]{./figures/figcompositeone}
\end{center}
\caption{The operator $O_1=\hat{V}^\textsf{A}_{j,j+1}(a)\hat{V}^\textsf{L}_{j+1,j+2}(a)\hat{V}^\textsf{R}_{j+1,j+2}(a)\hat{V}^\textsf{C}(a)_{j,j+1}$ corresponds to a process where a particle pair $(a,\bar{a})$ is created, there is some propagation, and the particles fuse subsequently. This has trivial action on the ground space, i.e., $P_0OP_0=d_a P_0$ is proportional to the identity.\label{fig:exampleprocessa} } 
\end{subfigure}
\hspace*{\fill} 
\begin{subfigure}{0.45\textwidth}
\begin{center}
\includegraphics[width=0.8\textwidth]{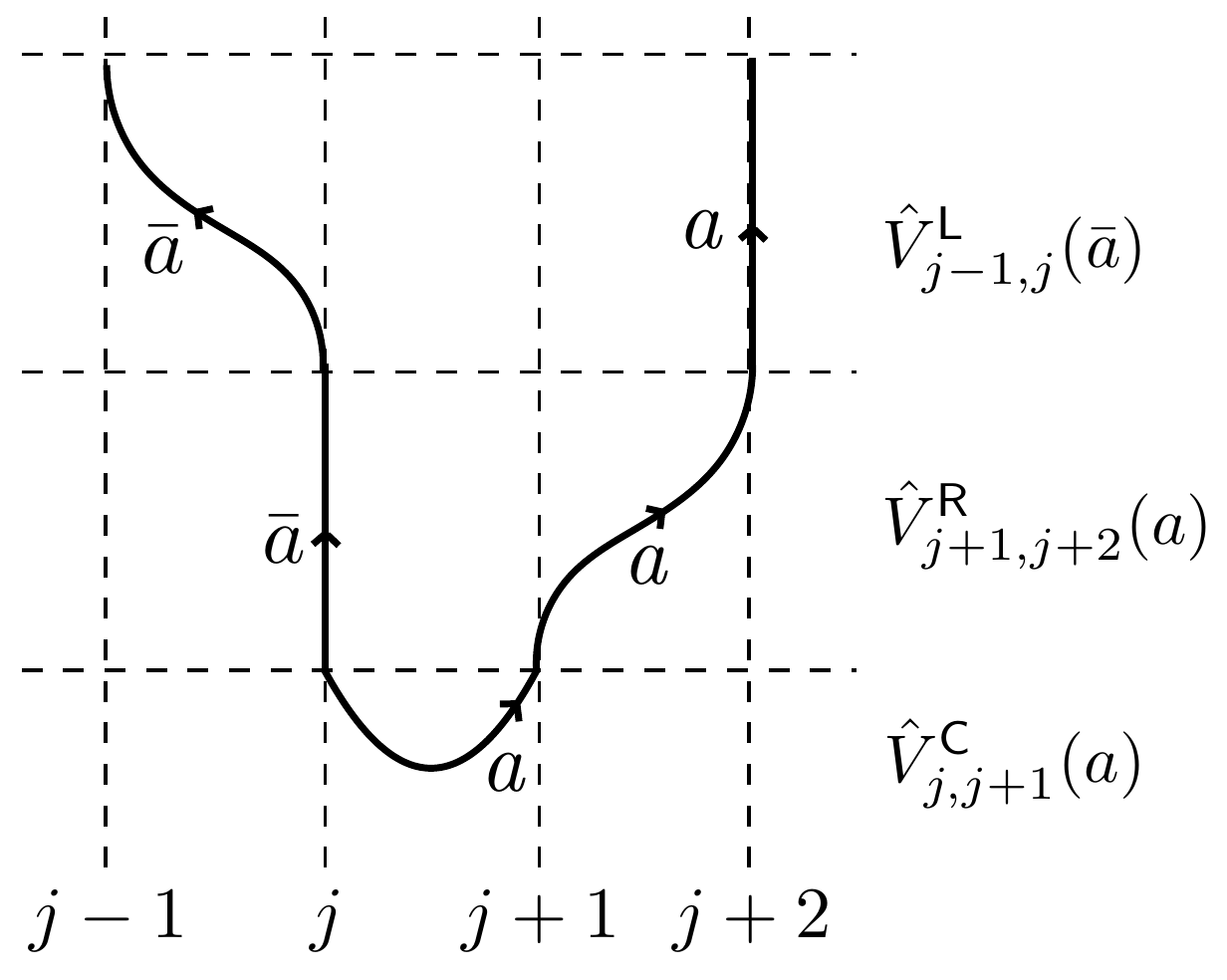}
\end{center}
\caption{The process described by the operator $O_2=\hat{V}^\textsf{L}_{j-1,j}(\bar{a})\hat{V}^\textsf{R}_{j,j+1}(a)\hat{V}^\textsf{C}_{j,j+1}(a)$
  leaves behind excitations, hence  $P_0O_3P_0=0$.\label{fig:exampleprocessc} }
\end{subfigure}\\
\begin{center}
\begin{subfigure}{0.8\textwidth}
\begin{center}
\includegraphics[width=0.8\textwidth]{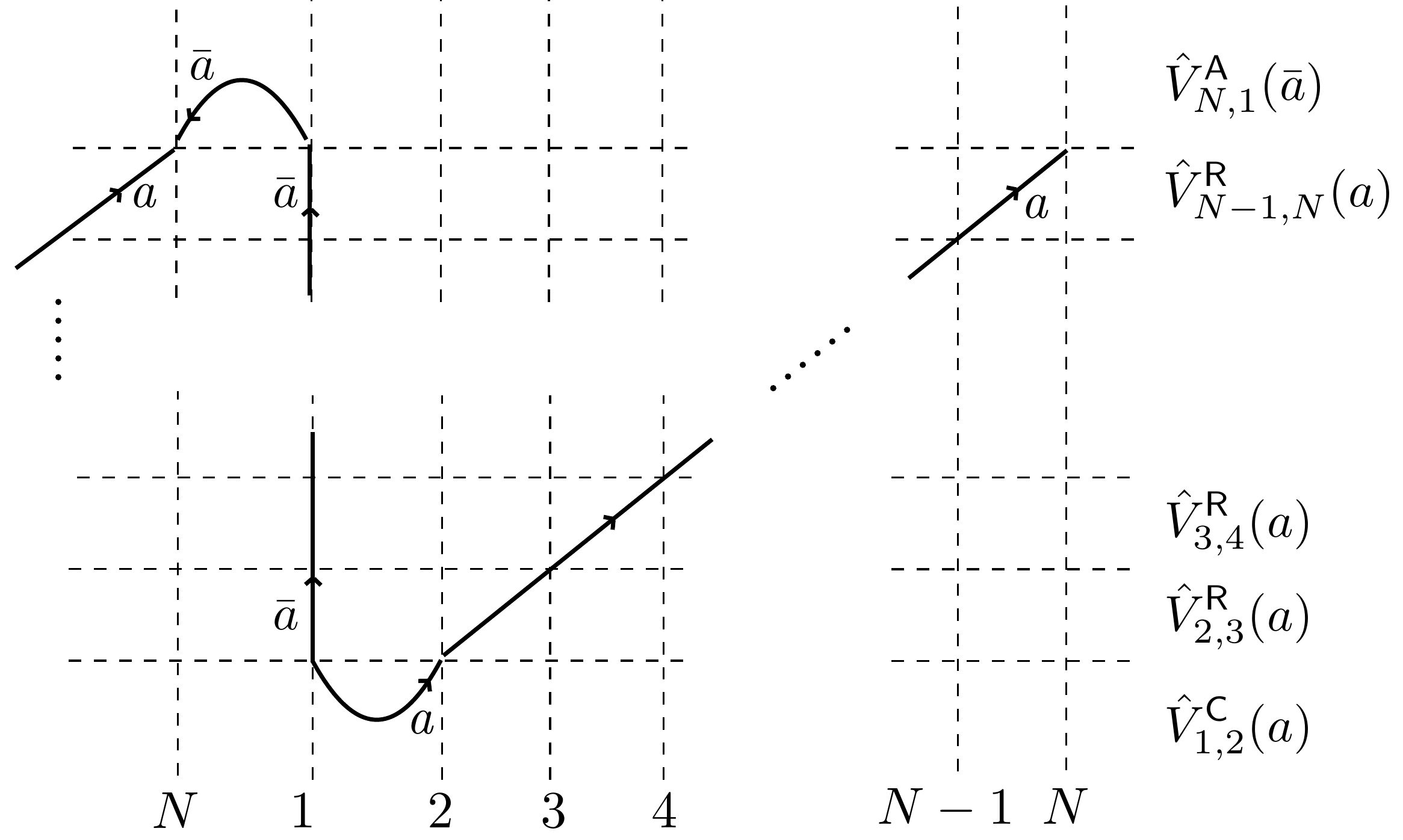}
\end{center}
\caption{ The operator $O_3=V^\textsf{A}(\bar{a})_{N,1}\cdots V^\textsf{R}(a)_{N-1,N}....V^\textsf{R}(a)_{3,4}V^\textsf{R}(a)_{2,3}V^\textsf{C}(a)_{1,2}$ corresponds to a process where a pair $(a,\bar{a})$ of particles is created, and they propagate all the way around the chain before annihilating. Its action on the ground space is given by the string-operator $P_0O_2P_0=F_a$.\label{fig:exampleprocessb} }
\end{subfigure}
\end{center}
\caption{This figure illustrates different processes in the diagrammatic formalism. Each process corresponds to an operator and is a product of elementary processes (diagrams). Ground space matrix elements vanish if the process leaves behind excitations (corresponding to endpoints of open strings). 
}
\end{figure}

\subsection{Perturbation theory for an effective anyon model\label{sec:perturbationtheoryanyon}}
In this section, we consider a $1$-dimensional translation-invariant system of anyons described by the Hamiltonian~$H_0$ introduced in~\eqref{eq:unperturbedhamiltoniananyon}.
We further consider a translation-invariant two-local perturbation $V=\sum_{j}\hat{V}_{j,j+1}$ with local terms~$\hat{V}_{j,j+1}$ of the form~\eqref{eq:vtwolocaloperator} given by
\begin{align}
\hat{V} &=\sum_a \left(\gamma_a V^C(a)+\overline{\gamma_a}V^A(a)\right) +\sum_a \left(\tau_a V^L(a)+\overline{\tau_a}V^R(a)\right) +V_{R}\ ,\label{eq:localperturbationanyon}
\end{align}
where $V_R$ collects all other two-anyon processes (it will turn out that in lowest order perturbation theory, only creation and propagation are relevant). The choice of complex conjugate pairs of parameters ensures that the perturbation is self-adjoint. We may think of $\gamma_a$ as the `creation ampitude', $\tau_a$ as the `propagation amplitude', and $\epsilon_a$ as the energy of particle $a$.
   
 We now compute the form of the effective Schrieffer-Wolff-Hamiltonian. Our main result is the following:

\begin{lemma}[Effective Hamiltonians for $1$-dimensional anyon chains]\label{lem:effectivehamiltoniansanyon}
Consider $H_0+\epsilon V$, with the perturbation~$V$ as described. Let $P_0$ be the projection onto the ground space of $H_0$. Then  the $L$-th order effective Hamiltonian  has the form
 \begin{align}
 \Heff^{(L)}(\epsilon)&=\sum_{a} f_L(\epsilon_a,\gamma_a,\tau_a) F_a+cP_0\ ,\label{eq:heffleps}
 \end{align}
for some constant $c\in\mathbb{R}$, and some function $f_L$ which is independent of the particle label~$a$ and is a homogeneous polynomial of degree $L$ in $\gamma_a$ and $\tau_a$.
\end{lemma} 
Clearly, the form Eq.~\eqref{eq:heffleps} of the effective Hamiltonian
is consistent with the topological superselection rule~\eqref{eq:localoperatorsuperselectionruleanyon}. 
However, Eq.~\eqref{eq:heffleps} provides additional information: for example, the coefficient of the string-operator~$F_a$ only depends on the energy~$\epsilon_a$ of anyon~$a$, as well as its creation/annihilation ($\gamma_a$ respectively $\overline{\gamma_a}$)  and propagation ($\tau_a$) amplitudes. 
There is no dependence on particles distinct from~$a$ (and corresponding braiding processes). 
Such terms only enter in higher orders of the perturbative series. 
This can be thought of as a rigorous derivation of the tunneling amplitude for a particle in the weak perturbation limit.
We note that due to $f_L$ being homogeneous of degree $L$, the dominant tunneling process will be highly sensitive to the perturbation strengths associated to different anyon labels $a$ for large system sizes~$L$. In the absence of a symmetry or fine tuning, it should be possible to order the terms~$f_L(\epsilon_a,\gamma_a,\tau_a)$ by absolute value, with different orders of magnitude being expected for each term (see Section~\ref{sec:perturbation_TQFT} for further discussion).

\begin{proof}
 It is easy to check that the conditions of Theorem~\ref{thm:effectivehamiltoniaschrieffer} are satisfied with $L$ equal to the length of the chain. Indeed, $(L-1)$-local terms have trivial action on the ground space as discussed in Section~\ref{sec:productlocaloperatorslogical}.  
 It thus suffices to consider expressions of the form
 \begin{align}
 P_0(VG)^{L-1}VP_0
 \end{align}
 involving $L$ factors of $V$. Inserting the definition~\eqref{eq:localperturbationanyon} of~$V$, and diagrammatically expanding each term as in Section~\ref{sec:anyonlocalops}, we are left with 
a linear combination of   terms of the form
 \begin{align}
 P_0V_{\alpha_1}GV_{\alpha_2}GV_{\alpha_3}\cdots GV_{\alpha_L}P_0\ ,
 \end{align}
where $V_{\alpha_j}$ is a  local operator given by an elementary (two-anyon) diagram (not a linear combination). Since such operators $V_{\alpha_j}$ map eigenstates of $H_0$ to eigenstates, and the energies of excited states reached from the ground space by applying such operators is independent of the ground state considered, each operator~$G$ merely adds a scalar, i.e., we have
\begin{align}
 P_0V_{\alpha_1}GV_{\alpha_2}GV_{\alpha_3}\cdots GV_{\alpha_L}P_0&=
 \theta(V_{\alpha_1},\ldots,V_{\alpha_L})\cdot P_0V_{\alpha_1}V_{\alpha_2}V_{\alpha_3}\cdots V_{\alpha_L}P_0
\end{align}
for some constant $\theta$ depending on the perturbations $\{V_{\alpha_j}\}$. But the rhs.~of this equation is 
a product of local operators as considered in Section~\ref{sec:productlocaloperatorslogical}. According to  the expression~\eqref{eq:localoperatorsuperselectionruleanyon}, this is a linear combination of string-operators, i.e.,
\begin{align}
P_0V_{\alpha_1}V_{\alpha_2}V_{\alpha_3}\cdots V_{\alpha_L}P_0&=\sum_{a} o_aF_a\ .
\end{align}
Furthermore, since each $V_{\alpha_j}$ is an elementary two-local operator, and we consider only products of length~$L$, the only terms 
$P_0V_{\alpha_1}V_{\alpha_2}V_{\alpha_3}\cdots V_{\alpha_L}P_0$ that have non-trivial action on the ground space are those associated with processes where a single particle (say of type~$a$) winds around the whole chain. We will call such a process {\em topologically non-trivial}. Its action on the ground space is given by a single string-operator~$F_a$.

In summary (rearranging the sum), we conclude that the $L$-th order effective Hamiltonian has the form~\eqref{eq:heffleps}, where the coefficient
$f_L(\epsilon_a,\gamma_a,\tau_a)$ has the form
\begin{align}
f_L(\epsilon_a,\gamma_a,\tau_a)&=\sum_{(V_{\alpha_1},\ldots,V_{\alpha_L})\in \Theta_a} \theta(V_{\alpha_1},\ldots,V_{\alpha_L}) \nu(V_{\alpha_1},\ldots,V_{\alpha_L})\ ,
\end{align}
and where the sum is over the set
\begin{align}
\Theta_a=\{(V_{\alpha_1},\ldots,V_{\alpha_L})\ |\ P_0V_{\alpha_1}\cdots V_{\alpha_L}P_0\in \mathbb{C}P_0\}
\end{align}
of all length-$L$-topologically non-trivial processes (consisting of elementary terms) involving particle~$a$. The coefficient  $\nu(V_{\alpha_1},\ldots,V_{\alpha_L})$  is defined by 
 $P_0V_{\alpha_1}\cdots V_{\alpha_L}P_0=\nu(V_{\alpha_1},\ldots,V_{\alpha_L})F_a$.
Furthermore, $\nu(V_{\alpha_1},\ldots,V_{\alpha_L})$ can only be non-zero when all $L$ operators $V_{\alpha_j}$ are either pair creation/anihilation or hopping terms involving the particle~$a$. 
This implies the claim. 
\end{proof}

\section{2D topological quantum field theories\label{sec:twodimensionalsystems}}

As discussed in Section~\ref{sec:majoranachain}, adding a local perturbation to a Majorana chain leads to an effective Hamiltonian given by the parity (string)-operator. 
Similarly, in the case of a general anyon chain (discussed in Section~\ref{sec:anyonchains}), the effective Hamiltonian is a linear combination of string-operators $F_a$, associated with different particle labels $a$. 
Here we generalize these considerations to arbitrary systems described by a 2-dimensional topological quantum field theory (TQFT) and subsequently specialize to microscopic models, including the toric code and the Levin-Wen string-net models~\cite{levin2005string}. 

Briefly, a TQFT associates a ``ground space''~$\cH_\Sigma$ to a two-dimensional surface~$\Sigma$  -- this is e.g., the ground space of a microscopic model of spins embedded in~$\Sigma$ with geometrically local interactions given by some Hamiltonian~$H_0$ (see Section~\ref{sec:microscopicmodels}). 
In other words,~$\cH_{\Sigma}\subset \cH_{phys,\Sigma}$ is generally a subspace of a certain space~$\cH_{phys,\Sigma}$ of physical degrees of freedom embedded in~$\Sigma$. The system has localized excitations (anyons) with (generally) non-abelian exchange statistics. In particular, there are well-defined physical processes involving creation, propagation, braiding and annihilation of anyons, with associated operators as in the case of $1$-dimensional anyon chains (see Section~\ref{sec:anyonchains}). Contrary to the latter, however, the particles are not constrained to move along a $1$-dimensional chain only, but may move arbitrarily on the surface~$\Sigma$. Nevertheless, the description of these processes is analogous to the case of spin chains, except for the addition of an extra spatial dimension. 
For example, this means that local operators acting on a region~$\cR\subset\Sigma$ are now represented by a linear combination of string-nets (directed trivalent graphs with labels satisfying the fusion rules) embedded in~$\cR\times [0,1]$. We refer to e.g.,~\cite{freedman2003topological} for more examples of this representation.

As before,  there are distinguished ground-space-to-ground-space (or ``vacuum-to-vacuum'') processes which  play a fundamental role. 
These are processes where a particle-anti-particle pair~$(a,\bar{a})$ is created, and the particles fuse after some propagation (tunneling), i.e., after tracing out a closed loop~$C$ on~$\Sigma$. 
Non-trivial logical operators must necessarily include topologically non-trivial loops $C$ on~$\Sigma$ in their support (the spatial region in which they are physically realized). 
In particular, for any such loop~$C$, there is a collection~$\{F_a(C)\}_a$ of string-operators associated with different particle labels. More precisely, a loop is a map~$C:[0,1]\rightarrow\Sigma$ satisfying $C(0)=C(1)$. Reversing direction of the loop gives a new loop~$\bar{C}(t):=C(1-t)$, and this is equivalent to interchanging particle- and antiparticle labels: we have the identity $F_a(C)=F_{\bar{a}}(\bar{C})$. In Section~\ref{sec:stringoperatorstqft}, we state some general properties of the string-operators~$\{F_a(C)\}_a$, and, in particular, explain how to express them in suitable bases of the ground space. 

\subsection{Perturbation theory for Hamiltonians corresponding to a TQFT}
\label{sec:perturbation_TQFT}
In general, the anyon model associated with a TQFT is emergent from a microscopic spin Hamiltonian $H_0$.
The anyon notion of site, as discussed in Section~\ref{sec:anyonchains}, does not necessarily coincide with the spin notion of site associated with the microscopic spin model.
Nevertheless, the following statements are true: 
\begin{enumerate}[(i)]
\item any non-trivial logical operator must include at least one non-contractible loop in its support.
\item given a perturbation~$V$ consisting of geometrically local operators, there exists some minimum  integer $L$ such that $H_0,V$ satisfy the topologically ordered condition with parameter~$L$.
\end{enumerate}
In general, the value of $L$ will depend on the length of the shortest non-contractible loop(s), and the resulting effective Hamiltonian will be of the form
\begin{align}
 \Heff^{(L)}(\epsilon)&= \epsilon^L \sum_{a, C : |C|=L}  f_L(a,C)  F_a(C) + c(\epsilon)P_0\ , \label{eq:hefflepstqft}
\end{align}
where the dependence on $H_0$ and the coefficients in $V$ has been left implicit. The sum is over all non-trivial loops~$C$ of length~$L$ (where length is defined in terms of the spin model), as well as all particle labels~$a$.

Computing the coefficients $\{f_L(a,C)\}$ may be challenging in general. Here we discuss a special case, where anyon processes associated with a single particle~$a$ (respectively its antiparticle~$\bar{a}$) are dominant (compared to processes involving other particles). 
That is, let us assume that we have a translation-invariant perturbation~$V$ of the form
\begin{align*}
V=\sum_{(j,j')}\left(\hat{V}^{(1)}_{j,j'}+\eta V^{(2)}_{j,j'}\right)\ ,
\end{align*}
where the sum is over all pairs $(j,j')$ of nearest-neighbor (anyonic) sites,
and $\hat{V}^{(1)}_{j,j'}=\hat{V}^{(1)}$ and 
$\hat{V}^{(2)}_{j,j'}=\hat{V}^{(2)}$ are both $1$- and $2$-local operators on the same anyon site lattice -- this is a straightforward generalization of anyon chains to 2D. Our specialization consists in the assumption that  all local creation, propagation and annihilation processes constituting the operator $\hat{V}^{(1)}_{j,j'}=\hat{V}^{(1)}$ only  correspond to a single anyon type $a$ (and $\bar{a}$), and that these processes are dominant in the sense that the remaining terms satisfy $\|\eta \hat{V}^{(2)}\| \ll \|\hat{V}^{(1)}\|$.  In the limit~$\eta\rightarrow 0$, perturbation theory in this model only involves the particles~$(a,\bar{a})$.

Assuming that the  shortest non-contractible loops have length~$L$ in this anyonic lattice, we claim that 
\begin{align}
\Heff^{(L)}(\epsilon)&= \epsilon^L \left(\sum_{C : |C|=L}  f_L(a,C)  F_a(C) + \eta^L G_\textsf{eff}^{(L)} \right)+ c(\epsilon)P_0 \ ,\label{eq:singleparticledominant}
\end{align}
where $G_\textsf{eff}^{(L)}$ is an effective Hamiltonian with the same form as $\Heff^{(L)}(\epsilon)$, but only contains string operators~$F_b(C)$ with $b\neq a$.
The reason  is that in order to generate a string operator~$F_b(C)$ in~$L$ steps (i.e., at $L$-th order in perturbation theory), we need to apply local operators corresponding to anyon $b$ $L$ times, as discussed in Lemma~\ref{lem:effectivehamiltoniansanyon}.
Such local operators can only be found in~$\eta V_2$, therefore we obtain the coefficient $\eta^L$ of $G_\textsf{eff}^{(L)}$.
Thus if we fix the system size and slowly increase~$\eta$ from~$0$, the (relative) change of the total effective Hamiltonian is exponentially small with respect to~$L$.
This implies that the ground state of the effective Hamiltonian is stable when $\eta$ is in a neighbourhood of $0$.
We will see in Section~\ref{sec:numerics} that  the final states of Hamiltonian interpolation are indeed stable in some regions of initial Hamiltonians.
The above discussion can be viewed as a partial explanation\footnote{Note that in the cases we consider in Section~\ref{sec:numerics}, $\hat{V}^{(1)}$ and $\hat{V}^{(2)}$ often do not live on the same anyon site lattice.} for this phenomenon.

 \newcommand*{\mcg}{\mathsf{MCG}}
\subsection{String-operators,  flux bases and the mapping class group\label{sec:stringoperatorstqft}}
In the following, we explain how to compute
effective Hamiltonians of the form~\eqref{eq:singleparticledominant} in the case where the perturbation is isotropic, resulting in identical coefficients $f_L(a,C)=f_L(a,C')$  for all loops~$C$ of identical length.  This 
 will be guaranteed by symmetries. We give explicit examples in Section~\ref{sec:numerics}.

For this purpose, we need a more detailed description of the action of string-operators on the ground space. Consider a fixed (directed) loop~$C:[0,1]\rightarrow\Sigma$ embedded in the surface~$\Sigma$. The process of creating a particle-anti-particle pair~$(a,\bar{a})$, then propagating $a$ along~$C$, and subsequently fusing with~$\bar{a}$ defines an operator~$F_a(C)$ which preserves the ground space~$\cH_\Sigma$.  The family of 
operators~$\{F_a(C)\}_a$ is mutually commuting and defines a representation of the Verlinde algebra. It is sometimes convenient to consider the associated (images of the) idempotents, which are explicitly given by (as a consequence of the Verlinde formula~\eqref{eq:verlindeformulasmatrix})  
\begin{align}
\label{eq_F_to_idempotents}
P_a(C)&=S_{1a}\sum_{b}\overline{S_{ba}}F_b(C)\ .
\end{align}
The operators $P_a(C)$ are mutually orthogonal projections $P_a(C)P_b(C) = \delta_{ab}P_a(C)$. 
The inverse relationship (using the unitarity of~$S$) reads 
\begin{align}
\label{eq_idempotents_to_F}
F_b(C)&=\sum_{a} \frac{S_{ba}}{S_{1a}}P_a(C)\ 
\end{align}
and is the generalization of \eqref{eq:stringoperatorsdiagonalanyon}: indeed, specializing to the case where~$\Sigma$ is the torus (this will be our main example of interest), and $C$ is a fundamental loop, the operators~$P_a(C)$ are rank-one projections (when restricted to the ground space), and determine (up to phases) an orthonormal basis of $\cB_{C}=\{\ket{a_C}\}_a$ of~$\cH_\Sigma$ by $P_a(C)=\proj{a_C}$.
In physics language, the state $\ket{a_C}$ has ``flux $a$'' through the loop~$C$.   (More generally, one may define ``fusion-tree'' basis for higher-genus surfaces~$\Sigma$ by considering certain collections of loops and the associated idempotents, see e.g.,~\cite{koenig2010quantum}. However, we will focus on the torus for simplicity.)

Consider now a pair of distinct loops $C$ and $C'$. 
Both families~$\{F_a(C)\}_a$ and $\{F_a(C')\}_{a}$ of operators act on the ground space, and it is natural to ask how they are related.
There is a simple relationship between these operators if $C'=\vartheta(C)$ is the image of~$C$ under an element $\vartheta:\Sigma\rightarrow\Sigma$ of the mapping class group~$\mcg_\Sigma$ of~$\Sigma$ (i.e., the group of orientation-preserving diffeomorphisms of the surface): The TQFT defines a projective unitary representation $V:\mcg_\Sigma\rightarrow\mathsf{U}(\cH_\Sigma)$ of this group on $\cH_\Sigma$, and we have 
\begin{align}
F_a(C')=V(\vartheta)F_a(C)V(\vartheta)^\dagger\quad\textrm{ for all }a\textrm{ if }C'=\vartheta(C)\ .\label{eq:basischangemappingclassgroup}
\end{align} 
In general, while the topology of the manifold is invariant under the mapping class group, the specific lattice realization may not be.
For this reason, if we desire to lift the representation $V$ to the full Hilbert space $\cH_{\Sigma} \supset \cH_{phys,\Sigma}$, such that 
the resulting projective unitary representation preserves the microscopic Hamiltonian $H_0$ under conjugation, we may need to restrict to a finite subgroup of the mapping class group~$\mcg_\Sigma$.
If the lattice has sufficient symmetry, such as for translation-invariant square or rhombic lattices, one may exploit these symmetries to make further conclusions about the resulting effective Hamiltonians.

\subsubsection{String-operators and the mapping class group for the torus \label{sec:stringoperatorstorus}}
For the torus, 
the mapping class group $\mcg_{\Sigma}$
is the group~$SL(2,\mathbb{Z})$. To specify how a group element  maps the torus to itself, it is convenient to parametrize the latter as follows: we fix complex numbers~$(e_1,e_2)$ and identify points~$z$ in  the complex plane according to 
\begin{align}
z\equiv z+n_1e_1+n_2e_2\qquad\textrm{ for }n_1,n_2\in\mathbb{Z}\ .
\end{align}
In other words, $(e_1,e_2)$ defines a lattice in~$\mathbb{C}$, whose unit cell is the torus (with opposite sides identified).  A group element
$A=\left(\begin{matrix} a & b \\
c & d\end{matrix}\right)\in SL(2,\mathbb{Z})$  then
 defines parameters $(e'_1,e'_2)$ by 
\begin{align}
e'_1&=a e_1+be_2\\
e'_2&=ce_1+de_2\ ,
\end{align}
 which a priori appear to be associated with a new torus. 
 However, the constraint that~$A\in SL(2,\mathbb{Z})$ ensures
that  $(e'_1,e'_2)$ and $(e_1,e_2)$ both define the same lattice, and this therefore defines a map from the torus to itself: The action of~$A$ is given by $\alpha e_1+\beta e_2\mapsto \alpha e_1'+\beta e_2'$ for $\alpha,\beta\in\mathbb{R}$, i.e., it is simply a linear map determined by~$A$.

The group $SL(2,\mathbb{Z})=\langle t,s\rangle$  is generated by the two elements 
\begin{align}
t=
\textrm{Dehn twist }\quad\left(
\begin{matrix}
1 & 1\\
0 & 1
\end{matrix}
\right)\qquad\textrm{ and }\qquad 
\pi/2\textrm{ rotation }\quad s=\left(
\begin{matrix}
0 & 1\\
-1 &0
\end{matrix}
\right)\label{eq:tsgenerators}
\end{align}
which are equivalent to the M\"obius transformations  $\tau\mapsto \tau+1$ and $\tau\mapsto -1/\tau$.
Clearly,  $t$  fixes~$e_1$ and hence the
loop $C:t\mapsto C(t)=t e_1$, $t\in [0,1]$ on the torus (this loop is one of the fundamental cycles). 
The matrices representing the unitaries $V(t)$ and $V(s)$ in  the basis~$\cB_{C}=\{\ket{a_C}\}_a$ of~$\cH_\Sigma$ (where $\ket{a_C}$ is an eigenstate of $P_a(C)=\proj{a_C}$)
are denoted~$T$ and $S$, respectively. These matrices are given by the modular tensor category: $T$ is a diagonal matrix with $T_{aa}=e^{i\theta_a}$ (where $\theta_a$ is the topological phase of particle~$a$), whereas $S$ is the usual~$S$-matrix. This defines the mapping class group representation on the Hilbert space~$\cH_\Sigma$ associated with the torus~$\Sigma$.

In the following, we compute explicit relationships between string-operators of minimal length. 
We consider two cases: a square torus and a rhombic torus. 
This allows us to express terms such as those appearing in Eq. \eqref{eq:hefflepstqft} in a fixed basis.

\paragraph{Square torus.}
\begin{figure}
\begin{center}
\includegraphics[scale=0.6]{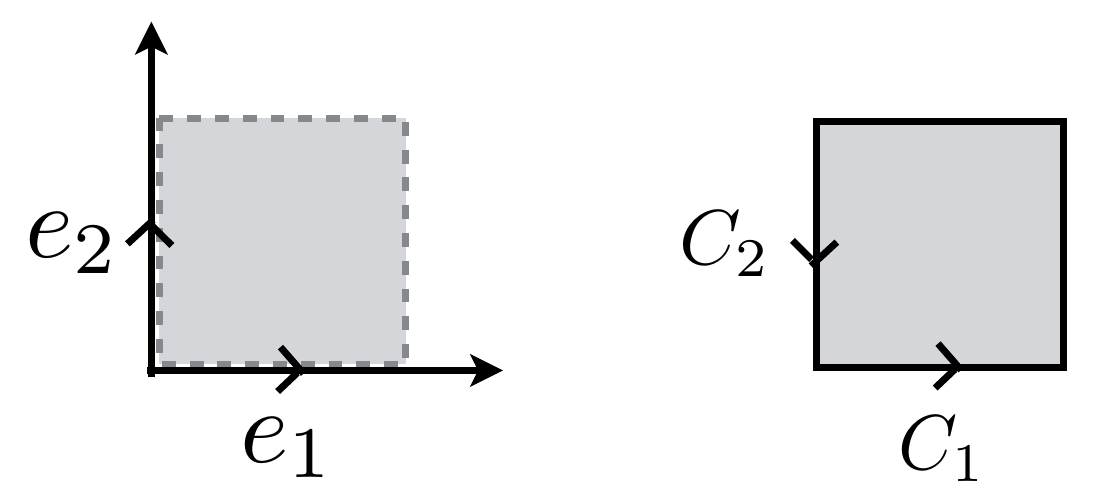}
\end{center}
\caption{Minimal loops on the square torus\label{fig:squaretorus}}
\end{figure}
Here we have 
\begin{align}
e_1=1\qquad\textrm{  and  }\qquad e_2=i\ .
\end{align} 
There are (up to translations) two loops of minimal length, 
\begin{align}
C_1(t)&=te_1\\
C_2(t)&=(1-t)e_2,
\end{align}
which may be traversed in either of two directions namely 
 for $t\in [0,1]$, see Fig.~\ref{fig:squaretorus}.
 Since $se_1=-e_2$ and $se_2=e_1$, we conclude that
 \begin{align}
C_2(t) &= s(C_1(t))   &    
\overline{C_1}(t) &=s^2(C_1(t)) &
\overline{C_2}(t) &= s^3(C_1(t))   &    
C_1(t) &=s^4(C_1(t)) &
\end{align}
In particular, expressed in the basis~$\cB_{\cC_1}$, we have
\begin{align}
\sum_{j=1,2} \left(F_a(C_j) + F_a(\overline{C_j})\right)& = \sum_{j=0}^3 S^j F_a(C_1) S^{-j} .\label{eq:SquareForm}
 \end{align} 
Thus,  when the lattice and Hamiltonian $H_0$ obey a $\pi/2$~rotation symmetry,
the  effective perturbation Hamiltonian will be proportional
to~\eqref{eq:SquareForm}. 
 This is the case for the toric code on a square lattice.

\paragraph{Rhombic torus.}
\begin{figure}
\begin{center}
\includegraphics[scale=0.6]{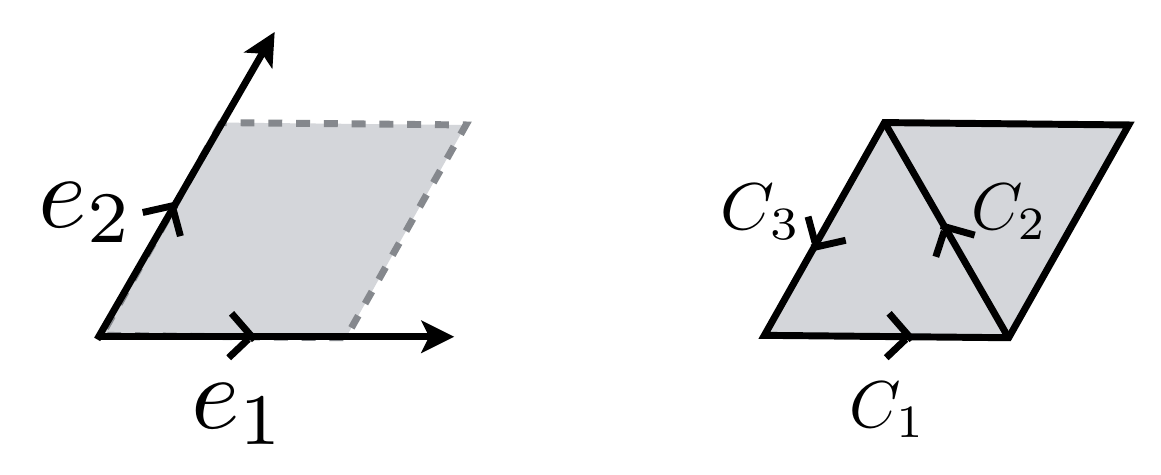}
\end{center}
\caption{Minimal loops on the rhombic torus\label{fig:rhombictorus}}
\end{figure}
We set 
\begin{align}
e_1=1\qquad\textrm{ and }\qquad 
e_2=\cos(2\pi/6)+i \sin (2\pi/6)\ .
\end{align}
Minimal loops of interest are shown in Fig.~\ref{fig:rhombictorus} and can be defined as 
\begin{align}
C_1(t) &=t e_1\\
C_2(t) & =e_1+t(e_2-e_1)\\
C_3(t) &=(1-t)e_2\ .
\end{align}
for $t\in [0,1]$.  Observe that these can be related by a $\pi/3$ rotation~$u$ (if we use the periodicity of the lattice), i.e.,
\begin{align}
\overline{C_3}(t)&=u(C_1(t))   & C_2(t) &= u^2(C_1(t)) & \overline{C_1}(t) &= u^3(C_1(t))   \\    
C_3(t)&=u^3(C_1(t))   & \overline{C_2}(t) &= u^5(C_1(t)) & C_1(t) &= u^6(C_1(t)).
\end{align}
Since such a rotation~$u$ maps $e_1,e_2$ to
\begin{align}
e_1' &= e_2\\
e_2' &= e_2-e_1\ ,
\end{align}
it is  realized by the element 
$u=\begin{pmatrix} 0 & 1\\ -1 & 1 \end{pmatrix}\in SL(2,\mathbb{Z})$, 
which decomposes into the generators~\eqref{eq:tsgenerators} as $u=ts^3ts$. 
We conclude that, expressed in the basis~$\cB_{\cC_1}$, we have 
\begin{align}
\label{eq_Heff_rhombic}
\sum_{j=1}^3 \left(F_a(C_j) + F_a(\overline{C_j})\right) &= \sum_{j=0}^5 U^j F_a(C_1) U^{-j} \quad\text{where}\quad U = TS^3TS .
\end{align}
Again, if the lattice and  Hamiltonian $H_0$ are invariant under a $\pi/3$ rotation, we may conclude that the effective perturbation Hamiltonian will have the form~\eqref{eq_Heff_rhombic}. This is the case for the Levin-Wen model on a  honeycomb lattice embedded in a rhombic torus (see also Section~\ref{sec:symmetryv}).

\subsection{Microscopic models\label{sec:microscopicmodels}}
The purpose of this section is two-fold: First, we briefly review the construction of the microscopic models we use in our numerical experiments in Section~\ref{sec:numerics}: these include the toric code (see Section~\ref{sec:toriccodemicroscopic}) as well as the doubled semion and the doubled Fibonacci model, both instantiations of the Levin-Wen construction (see Section~\ref{sec:Short_Introduction_LW}). Second, we define  single-qubit operators in these models 
and discuss their  action on quasi-particle excitations (i.e., anyons). 
This translation of local terms  in the microscopic spin Hamiltonian into operators in the effective anyon models is necessary to apply the perturbative arguments presented in Section~\ref{sec:perturbation_TQFT}. We will use these local terms  to define translation-invariant perturbations (respectively trivial initial Hamiltonians) in Section~\ref{sec:numerics}).

\subsubsection{The toric code\label{sec:toriccodemicroscopic}}
Kitaev's toric code~\cite{kitaev2003fault} is arguably the simplest exactly solvable model which supports anyons. It can be defined on a variety of lattices, including square and honeycomb lattices. Here we will introduce the Hamiltonian corresponding to honeycomb lattice. On each edge of the lattice resides a qubit. The Hamiltonian consists of two parts and takes the form
\begin{align}
\Htop=- \sum_v A_v -\sum_p B_p\ ,\label{eq:toriccodehamiltonian}
\end{align} where $B_p=X^{\otimes 6}$ is the tensor product of Pauli-$X$ operators  on the six edges of the plaquette~$p$, and $A_v=Z^{\otimes 3}$ is the tensor product of Pauli-$Z$ operators on the three edges connected to the vertex~$v$.

Note that in terms of its anyonic content, the toric code is
described by the double of~$\Integer_2$; hence a model with the same 
type of topological order could be obtained following the prescription 
given by Levin and Wen (see Section~\ref{sec:Short_Introduction_LW}). Here we are not following this route, but instead exploit that 
this has the structure of a quantum double (see~\cite{kitaev2003fault}). The resulting construction, given by~\eqref{eq:toriccodehamiltonian}, results in a simpler plaquette term~$B_p$ as opposed to the Levin-Wen construction.

The anyonic excitations supported by the toric code are labeled by $\{ \bm{1}, \bm{e},\bm{m},\bm{\epsilon}\}$. 
The $\bm{e}$ anyon or electric excitation corresponds to vertex term excitations.
The $\bm{m}$ anyon or magnetic excitations correspond to plaquete term excitations.
Finally, the $\bm{\epsilon}$ anyon corresponds to an excitation on both plaquete and vertex and has the exchange statistics of a fermion.
We can write down the string operators $F_a(C)$  for a closed loop~$C$ on the lattice explicitly (see~\cite{kitaev2003fault}).  Without loss of generality, we can set $F_{\bm{e}}(C)= P_0 \bigotimes_{i\in C} X_i P_0$ and $F_{\bm{m}}(C)= P_0 \bigotimes_{i\in D} Z_i P_0$, where $D$ is a closed loop on the dual lattice corresponding to $C$.  
Finally, the operator $F_{\bm{\epsilon}}(C)=F_{\bm{e}}(C)\times F_{\bm{m}}(C)$ can be written as a product of~$F_{\bm{e}}(C)$ and~$F_{\bm{m}}(C)$, since~$\bm{e}$ and~$\bm{m}$ always fuse to~$\bm{\epsilon}$. With respect to the ordering $(\bm{1},\bm{e},\bm{m},\bm{\epsilon})$ of the anyons, the $S$- and $T$-matrices described in Section~\ref{sec_modular_tensor_cat} are given by
\begin{align}
\begin{matrix}
  T = \mathsf{diag}(1, 1, 1, -1)
 \qquad   
  &S =1/2
  \begin{pmatrix}
  1 & 1 & 1 & 1 \\ 
 1 &  1 & -1 & -1 \\
 1 & -1 & 1 & -1 \\
 1 & -1  & -1 & 1            
 \end{pmatrix}
\end{matrix}\label{eq:sttoric}
\end{align}
for the toric code. 

\paragraph{Local spin operators. }A natural basis of (Hermitian) operators on a single qubit is given by the Pauli operators. For the toric code, each of these operators has a natural interpretation in terms of the underlying anyon model.

Consider for example a single-qubit~$Z$-operator. The ``anyonic lattice'' associated with $\bm{m}$-anyons is the dual lattice (i.e., these anyons `live' on plaquettes), and a single-qubit $Z$-operator acts by either creating or annihilating a $(\bm{m},\bar{\bm{m}})=(\bm{m},\bm{m})$ on the neighboring plaquettes, or propagating an existing $\bm{m}$ from one plaquette to the other. That is, in the  terminology of  Section~\ref{sec:anyonlocalops}, a $Z$-operator acts as a local term
\begin{align}
Z\qquad \longleftrightarrow\quad \hat{V}^{\mathsf{C}}(\bm{m})+\hat{V}^{\mathsf{A}}(\bm{m})+\hat{V}^{\mathsf{L}}(\bm{m})+\hat{V}^{\mathsf{R}}(\bm{m})\label{eq:Zoperatorsinglequbit}
\end{align}
in the  effective anyon model. An analogous identity holds for $X$, which is associated with~$\bm{e}$-anyons: the latter live on vertices of the spin lattice. Finally, $Y$-operators act on~$\bm{\epsilon}$-anyons in the same manner; these anyons live on `supersites', consisting of a plaquette and and an adjacent vertex.

\subsubsection{Short introduction to the Levin-Wen model}
\label{sec:Short_Introduction_LW}
Levin and Wen~\cite{levin2005string} define a family of frustration-free commuting Hamiltonian with topologically ordered ground space and localized anyonic excitations. Their construction is based on interpreting the state of spins residing on the edges of a trivalent lattice (such as a honeycomb lattice) as configurations of string-nets. 

To specify a string-net model, we need algebraic data associated with an anyon model as described in Section~\ref{sec_modular_tensor_cat}. This specifies, in particular,  a set of anyon labels $\cF=\{a_i\}$, associated fusion rules, as well as $S$- and $F$-matrices. 
The Levin-Wen model then associates a qudit to each edge of the lattice, where the local dimension of each spin corresponds to the number of anyon labels in~$\mathcal{F}$. One chooses an orthonormal basis $\{\ket{a}\}_{a\in\mathcal{F}}\subset\mathbb{C}^{|\mathcal{F}|}$ indexed by anyon labels; in the following, we usually simply write~$a$ instead of~$\ket{a}$ to specify a state of a spin in the microscopic model.
The  Levin-Wen spin Hamiltonian can be divided into two parts, 
\begin{align}
\Htop=- \sum_v A_v -\sum_p B_p\ ,\label{eq:htopvb}
\end{align}
where each $B_p$ is a projector acting on the $12$~edges around a plaquette~$p$, and each $A_v$ is a projector acting on the~$3$ edges around a vertex $v$.
In particular, we can construct the spin Hamiltonian for the doubled semion and the doubled Fibonacci models in this way by choosing different initial data. 

As long as all the particles in the underlying model $\cF$ are their own antiparticles (i.e., the involution $a\mapsto \bar{a}$ is the identity), it is not necessary to assign an orientation to each edge of the lattice.
This affords us an important simplification, which is justified for the models under consideration: these only have a single non-trivial anyon label, which is itself its own antiparticle (recall that the trivial label satisfies $\bar{1}=1$).
With this simplification, which we will use throughout the remainder of this paper, the vertex operator~$A_v$ can be written as
\begin{align*}
\includegraphics[height=10ex]{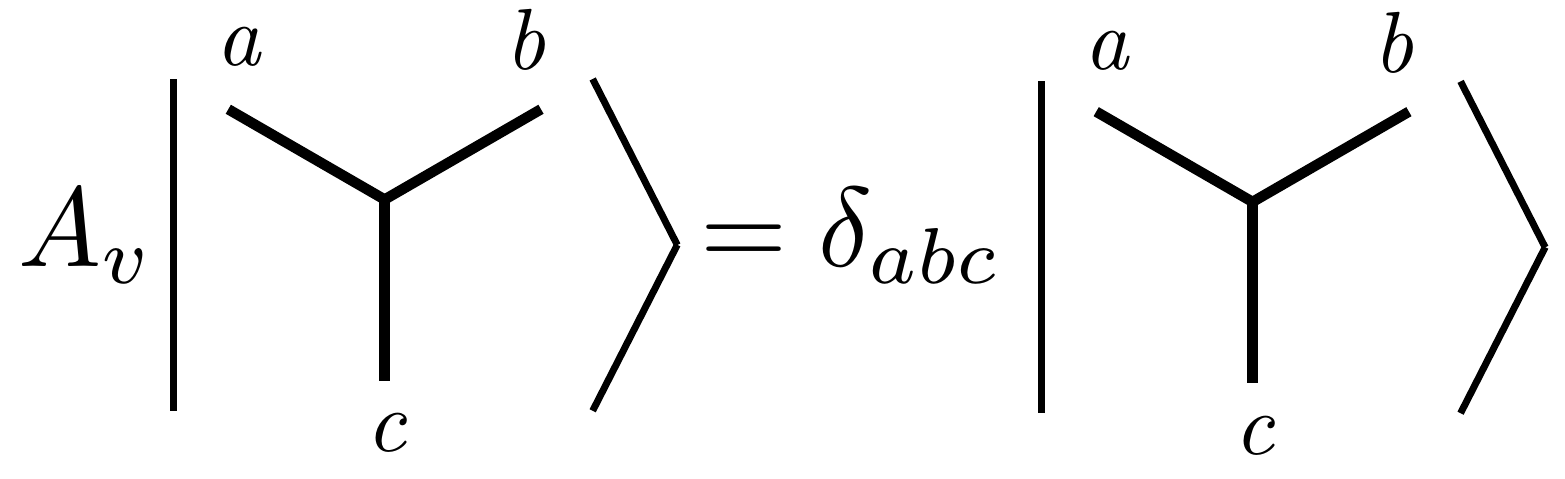}
\end{align*}
where $\delta_{abc}=1$ if $a$ and $b$ can fuse to $c$ and $\delta_{abc}=0$ otherwise. The plaquette operator~$B_p$ is more complicated compared to $A_v$. 
We will give its form without further explanation
\begin{align*}
\includegraphics[height=15ex]{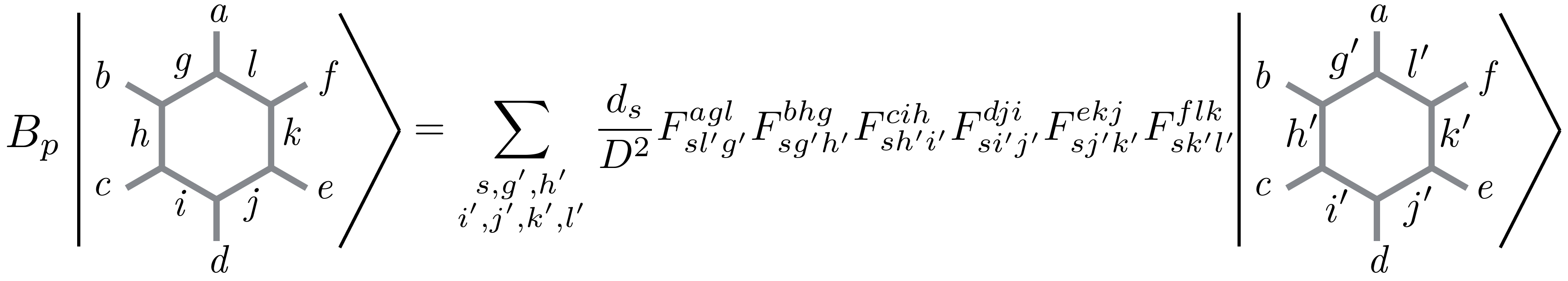},
\end{align*}
where $d_s$ is the quantum dimension of the anyon label $s$, and $D=\sqrt{\sum_j d_j^2}$ is the total quantum dimension.

Having specified the spin Hamiltonian, we stress that the anyon labels~$\cF$ used in this construction should not be confused with the anyon labels~$\mathrm{D}(\mathcal{F})$ describing the local excitations in the resulting Hamiltonian~\eqref{eq:htopvb}. The latter can be described as `pairs' of anyons from~$\cF$, i.e., $\mathrm{D}(\mathcal{F})=\{(a_i,a_j)\}_{a_i,a_j \in \mathcal{F}}$.
Their fusion, twist and braiding properties are described by the double of the original theory.
The $S_{\mathrm{D}(\mathcal{F})}$- and $F_{\mathrm{D}(\mathcal{F})}$- matrices of $\mathrm{D}(\mathcal{F})$ can be obtained from the $S$- and $T$-matrix associated with $\mathcal{F}$  (see~\cite{levin2005string}). String operators~$F_{a_i,a_j}(C)$ acting on the spin lattice have also been explicitly constructed in~\cite{levin2005string}

Below, we present some of the specifics of two models constructed in this way: the  doubled semion and doubled Fibonacci model. In addition to Kitaev's toric codes~$\mathrm{D}(\Integer_2)$,  these are the only models defined on two labels (i.e., with microscopic qubit degrees of freedom).

\subsubsection{The doubled semion model\label{sec:doubledsemion}}
The underlying string-net model of the doubled semion  model only consists of one non-trivial label $\mathbf{s}$ and the trivial label $\mathbf{1}$. To specify the spin Hamiltonian, we have $d_{\bfs}=1$, and $\delta_{abc}=1$ if and only if an even number of~$a,b,c$ are $\bfs$. The $F$-matrix is given by $F^{\bfs \bfs \bfone}_{\bfs \bfs \bfone}=-1$ and otherwise~$F^{abc}_{def}$ is $0$ or $1$ depending on whether $(a,b,c,d,e,f)$ is a legal configuration (see~\cite{kitaev2006anyons} for more detailed explanation).
As we explained above, to construct a spin Hamiltonian, we put a qubit on each edge of the lattice with orthonormal basis $\ket{\bfone},\ket{\bfs}$.
 The spin Hamiltonian obtained this way is similar to the toric code and it also supports Abelian anyons.
The excitations of the spin model can be labeled by~$\mathrm{D}(\cF)= \{ (\bfone,\bfone),(\bfone,\bfs),(\bfs,\bfone),(\bfs,\bfs)\}$, which is the quantum double of $\cF =\{\bfone, \bfs\}$. With respect to the given ordering of anyons, the $S$- and $T$-matrices of these excitations are given by 
\begin{align}
  S =1/2
  \begin{pmatrix}
  1 & 1 & 1 & 1 \\ 
 1 &  -1 & 1 & -1 \\
 1 & 1 & -1 & -1 \\
 1 & -1  & -1 & 1            
 \end{pmatrix}
 \qquad     T = \mathsf{diag}(1, i, -i, 1)
 \label{eq:stsemion}
\end{align}

\paragraph{Local operators.} 
Identifying $\ket{\bfone}$ with the standard basis state~$\ket{0}$ and
$\ket{\bfs}$ with $\ket{1}$, we can again use Pauli operators 
to  parametrize single-spin Hamiltonian terms. 

Here we will discuss the effect of single qubit operators~$X$ and~$Z$ on the ground states of the resulting topologically ordered Hamiltonian. 
The goal is to interpret single spin operators in terms of effective anyon creation, annihilation and hopping operators.

When  $Z$-operator is applied to  an edge of the system in a ground state, only the neighboring plaquete projectors  $B_p$ will become excited.
More specifically, a pair of  $(\bfs,\bfs)$ anyons are created if none were present. 
Since $(\bfs,\bfs)$ is an abelian anyon, in fact a boson, and is the anti-particle of itself, a $Z$ operator could also move an $(\bfs,\bfs)$ anyon or annihilate two such particles if they are already present. Thus we conclude that single-qubit $Z$-operators have a similar action as in the toric code (cf.~\eqref{eq:Zoperatorsinglequbit}), with $\bfs$ playing the role of the anyon~$\bm{m}$.

When an $X$ operator is applied on edge of the system in a ground state, it excites the two neighboring vertex terms~$A_v$ (in the sense that the state is no longer a $+1$-eigenstate any longer).
Since the plaquete terms~$B_p$ are only defined within the subspace stabilized by~$A_v$, the four plaquette terms~$B_p$ terms around the edge also become excited.  It is unclear how to provide a full interpretation of $X$ operators in terms of an effective anyon language. In order to provide this, a full interpretation of the spin Hilbert space and its operators in the effective anyonic language is required; such a description is currently not known. 

In summary, this situation is quite different from the case of the toric code, where~$X$ and~$Z$ are dual to each other.

\subsubsection{The doubled Fibonacci\label{sec:doubledfib}}
Again, the underlying string-net model of doubled Fibonacci contains only one non-trivial label~$\bftau$, with quantum dimension $d_{\bftau}=\varphi$, where $\varphi=\frac{1+\sqrt{5}}{2}$.  The fusion rules are given by~$\delta_{abc}=0$ if only one of the $a, b, c$ is~$\bftau$, and otherwise~$\delta_{abc}=1$. Non-trivial values of $F$ are
\begin{align*}
F^{\bftau \bftau \bfone}_{\bftau \bftau \bfone}=\varphi^{-1},&\ F^{\bftau \bftau \bfone}_{\bftau \bftau \bftau}=\varphi^{-1/2} \\
F^{\bftau \bftau \bftau}_{\bftau \bftau \bfone}=\varphi^{-1/2},&\ F^{\bftau \bftau \bftau}_{\bftau \bftau \bftau}=-\varphi^{-1}, \\
\end{align*}
and otherwise $F^{abc}_{def}$ is either $0$ or $1$ depending on whether $(a,b,c,d,e,f)$ is a legal configuration.

Many aspects of the doubled Fibonacci spin Hamiltonian are similar to the doubled semion model:
\begin{itemize}
\item There is one qubit on each edge, with orthonormal basis states associated with the anyon labels $\cF=\{\bfone,\bftau\}$.
\item The anyons supported by the spin Hamiltonian carry labels $\mathrm{D}(\mathrm{\cF}) = \{(\bfone,\bfone), (\bfone,\bftau), (\bftau,\bfone), (\bftau, \bftau)\}$.
\end{itemize}
With respect to the given ordering of anyons, the $S$- and $T$-matrices are given by 
\begin{align}
  S =
  \begin{pmatrix}
  1 & \varphi & \varphi & \varphi^2 \\ 
 \varphi &  -1 & \varphi^2 & -\varphi \\
 \varphi & \varphi^2 & -1 & -\varphi \\
 \varphi^2 & -\varphi  & -\varphi & 1            
 \end{pmatrix}
 /(1+\varphi^2) 
 \qquad   
  T = \mathsf{diag}(1, e^{-4\pi/5}, e^{4\pi/5}, 1) \ .
 \label{eq:doubledfibst}
\end{align}

A substantial difference to the doubled semion model is that the non-trivial anyons supported by the model are non-abelian. One manifestation of this fact we encounter concerns the $(\bftau,\bftau)$-anyon:
\begin{itemize}
\item While $(\bftau, \bftau)$ is its own anti-particle, it is not an abelian particle so in general two $(\bftau, \bftau)$ particles will not necessarily annihilate with each other.
In other words, the dimension of the subspace carrying two localized $(\bftau, \bftau)$ charges is larger than the dimension of the charge-free subspace.
\item Two intersecting string operators $F_{(\bftau, \bftau)}(C_1)$ and $F_{(\bftau, \bftau)}(C_2)$ corresponding to the $(\bftau, \bftau)$ particle do not commute with each other.
\end{itemize}
Neither of these properties holds for the $(\bfs,\bfs)$-anyon in the case of the doubled semion model.

\paragraph{Local operators.} Similarly, as before,
we identify $\ket{\bfone}$ with the standard basis state~$\ket{0}$ and~$\ket{\bftau}$ with $\ket{1}$, enabling us to express single-qubit operators in terms of the standard Pauli operators.

Again, we want to consider the effect of single qubit operators in terms of anyons. This is generally rather tricky, but for single-qubit $Z$-operators, we can obtain partial information from an analysis presented in appendix~\ref{sec:stringFibonacci}:  Let $\ket{\psi}$ be a ground state. Then $Z\ket{\psi}=\frac{1}{\sqrt{5}}\ket{\psi}+\frac{4}{5}\ket{\varphi}$, 
where $\ket{\varphi}$ is a $\psi$-dependent excited state with a pair of $(\bftau, \bftau)$ on the plaquettes next to the edge~$Z$ acts on. Thus the resulting state after application of a single~$Z$ operator has support both on the excited and as well as the ground subspace. Again, this is in contrast to the doubled semion model, where a single-qubit~$Z$ operator applied to the ground space always results in an excited eigenstate of the Hamiltonian.

\newpage

\section{Numerics \label{sec:numerics}}
\begin{figure}
\centering
\includegraphics[scale=0.5]{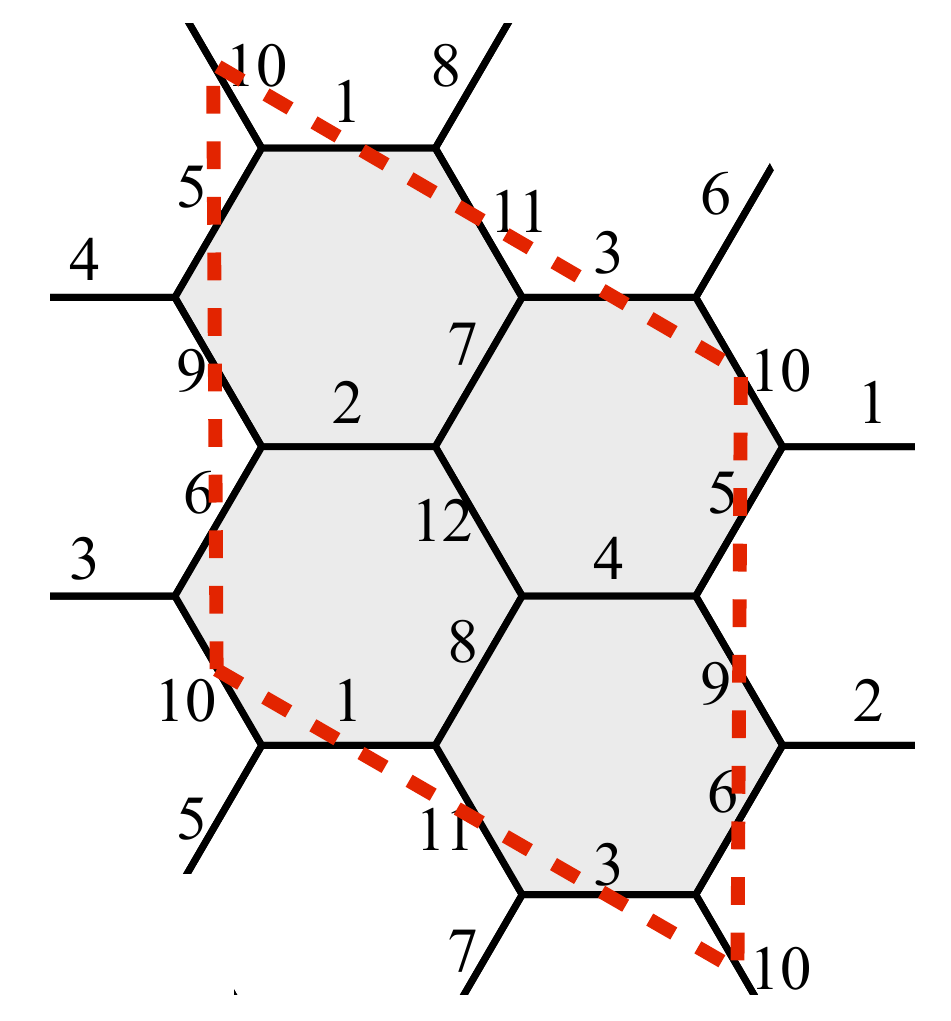}
\caption{The 12-qubit-torus~\label{fig_lattice} we use for numerical simulation (qubits are numbered~$1$ to~$12$.  
It is a rhombic torus and we can identify three minimal loops $\{1,2\}, \{5,7\}, \{9,11\}$ (and their inverses) which are related by $\pi/3$ rotations. }
\end{figure}
In this section, we present results obtained by numerically simulating Hamiltonian interpolation for small systems. Specifically, we consider three topologically ordered systems on the 12-qubit  honeycomb lattice  of Fig~\ref{fig_lattice}: the toric code,  the doubled semion  and the doubled Fibonacci Levin-Wen models. That is, the target Hamiltonian~$\Htop$ is given either by~\eqref{eq:toriccodehamiltonian} (with stabilizer plaquette- and vertex-operators $A_v$ and $B_p$) in the toric code case, and expression~\eqref{eq:htopvb} specified in Section~\ref{sec:Short_Introduction_LW} (with projection operators $A_v$ and $B_p$)  for the doubled semion and the doubled Fibonacci case. As initial Hamiltonian~$\Htriv$, we choose certain translation-invariant Hamiltonians consisting of single-qubit Pauli-$X$, Pauli-$Y$ and Pauli-$Z$ operators (see Sections~\ref{sec:doubledsemion} and~\ref{sec:doubledfib}  for their definition and a discussion of the effect of these operators in the two Levin-Wen models.) For concreteness and ease of visualization, we will consider the following families of such Hamiltonians: the one-parameter family  
\begin{align}
\Htriv(\theta)&=\cos\theta \sum_j Z_j+\sin\theta\sum_j X_j\label{eq:htrivialthetdef}
\end{align}
where $\theta\in [0,2\pi[$, 
and two two-parameter families of the form
\begin{align}
\Htriv^{\pm}(a,b)&=a\sum_j X_j+b\sum_j Y_j\pm (1-a^2-b^2)^{1/2}\sum_j Z_j\ ,\label{eq:hplusminusdef}
\end{align}
where $(a,b)\in\mathbb{R}^2$ belongs to the unit disc, $a^2+b^2\leq 1$.  (In some instances, we will permute the roles of $X$, $Y$ and $Z$, and use an additional superscript to indicate this.)

For different parameter choices~$\theta$ respectively~$(a,b)$, we study Hamiltonian interpolation (i.e., the evolution~\eqref{eq:adiabaticprep}) along the linear interpolation path $H(t)$ (cf.~\eqref{eq:hlinearinterpol})  with a total evolution time~$T$. In order to numerically simulate the evolution under the time-dependent Schr\"odinger equation, we perform a time-dependent Trotter expansion using the approximation 
\begin{align}
\mathcal{T}\exp \left(i\int_0^t H(s)ds \right)
\approx 
\prod_{j=1}^{\lfloor T/\Delta t\rfloor}  e^{iH(j\cdot \Delta T)\Delta t}
\quad\textrm{ and }\quad 
e^{iH(t) \Delta t}\approx e^{i\frac{(T-t)}{T}\Htriv \Delta t}
e^{i\frac{t}{T}\Htop\Delta t}\ .\label{eq:trottertimes}
\end{align}
Unless otherwise specified, the time discretization is taken to be~$\Delta t =0.1$.

\newcommand*{\nonadiabaticity}{\epsilon_{\textrm{adia}}}
\subsection{Quantities of interest and summary of observations}
Recall that our initial state $\Psi(0)=\varphi^{\otimes 12}$ is the unique 12-qubit ground state of the chosen trivial Hamiltonian~$\Htriv$. We are interested in the states~$\Psi(t)$ along the evolution, and, in particular, the final state~$\Psi(T)$ for a total evolution time~$T$. For notational convenience, we
will write~$\Psi_{\theta}(t)$, respectively $\Psi^{\pm}_{a,b}(t)$ to indicate which of the initial Hamiltonians~$\Htriv$ is considered (cf.~\eqref{eq:htrivialthetdef} and~\eqref{eq:hplusminusdef}). We consider the following two aspects:
\begin{description}
\item[(Non)-adiabaticity:]
We investigate whether the state~$\Psi(t)$ follows the instantanenous ground space along the evolution~\eqref{eq:adiabaticprep}. We quantify this using the {\em adiabaticity error}, which we define (for a fixed total evolution time~$T$, which we suppress in the notation) as
\begin{align}
\nonadiabaticity(t):=1-|\langle \Psi(t)|P_0(t)|\Psi(t)\rangle|^2\qquad\textrm{ for }0\leq t\leq T\ , \label{eq:adiabaticityerrordef}
\end{align}
where $P_0(t)$ is the projection onto the ground space of $H(t)$ (note that except for $t=T$, where $P_0(T)$ projects onto the degenerate ground space of~$\Htop$, this is generally a rank-one projection). The function $t\mapsto\nonadiabaticity(t)$ quantifies the overlap with the instantaneous ground state of~$H(t)$ along the Hamiltonian interpolation~$t\mapsto H(t)$, and hence directly reflects adiabaticity.

Ultimately, we are interested in whether the
evolution reaches a ground state of~$\Htop$. This is measured by the expression~$\nonadiabaticity(T)$, which quantifies the deviation of the final state~$\Psi(T)$ from the ground space of~$\Htop$. Clearly, the quantity~$\nonadiabaticity(T)$ depends on the choice of initial Hamiltonian~$\Htriv$ (i.e., the parameters $\theta$ respectively~$(a,b)$) and the total evolution time~$T$. For sufficiently large choices of the latter, we expect the adiabaticity assumption underlying Conjecture~\ref{claim:targetstates} to be satisfied, and this is directly quantifiable by means of the adiabaticity error. We will also discuss situations where, as discussed in Observation~\ref{obs:groundstatenotreached}, symmetries prevent reaching the ground space of~$\Htop$ as reflected in a value of~$\nonadiabaticity(T)$ close to~$1$.

\item[Logical state:] 
assuming the ground space of~$\Htop$ is reached (as quantified by~$\nonadiabaticity(T)$), we will identify the logical state~$\Psi(T)$ and investigate its stability under perturbations of the the initial Hamiltonian~$\Htriv$ (i.e., changes of the parameters~$\theta$ respectively $(a,b)$). For this purpose, we employ the following measures:

\begin{itemize}
\item
We argue (see Section~\ref{sec:symmetryv}) that symmetries constrain the projection of the resulting state~$\Psi(T)$  onto the ground space of~$\Htop$ to a two-dimensional subspace (see Section~\ref{sec:symmetryv}).  For the toric code, the state is then fully determined by the 
expectation values~$\langle\bar{X}\rangle_{\Psi(T)}$, $\langle\bar{Z}\rangle_{\Psi(T)}$ of two logical operators $\bar{X}$ and $\bar{Z}$. To investigate stability properties of the prepared state, we can therefore consider~$(\langle\bar{X}\rangle_{\Psi(T)},\langle\bar{Z}\rangle_{\Psi(T)})$ as a function of parameters of the initial Hamiltonian. 
\item
for the Levin-Wen models, we proceed as follows: 
we pick a suitable reference state~$\ket{\psi_R}\in (\mathbb{C}^2)^{\otimes 12}$
in the ground space of~$\Htop$, and then study
how the overlap $|\langle \Psi^{\pm}_{a,b}(T)|\psi_R\rangle|^2$ changes as the parameters~$(a,b)$ of the initial Hamiltonian are varied. In particular, 
if we fix a pair $(a_0,b_0)$ and choose $\ket{\psi_R}$ as the normalized projection of the state $\ket{\Psi^{\pm}_{a_0,b_0}(T)}$ onto the ground space of~$\Htop$, this allows us to study the stability of the prepared state~$\ket{\Psi^\pm_{a,b}(T)}$ as a function of the Hamiltonian parameters~$(a,b)$ in the neighborhood of~$(a_0,b_0)$. 

According to the reasoning in Section~\ref{sec:adiabaticprepgroundstates} (see Conjecture~\ref{claim:targetstates}), the specific target state~$\ket{\psi^{\textrm{ref}}_{a_0,b_0}}$ chosen in this way should correspond to the ground state of~$\Htop+\epsilon H^{\pm}(a_0,b_0)$ in the limit~$\epsilon\rightarrow 0$ of infinitesimally small perturbations (or, more precisely, the corresponding effective Hamiltonian).  
Furthermore, according to the reasoning in Section~\ref{sec:perturbation_TQFT}, the family of effective Hamiltonians associated with~$\Htop+\epsilon H^{\pm}(a,b)$ has a very specific form. This should give rise to a certain stability of the ground space as a function of the parameters~$(a,b)$. 

To support this reasoning, we numerically compute the (exact) ground state~$\ket{\psi^{\textrm{pert}}_{a,b}}$ of~$\Htop+\epsilon H^{\pm}(a,b)$ for the choice~$\epsilon=0.001$ (as a proxy for the effective Hamiltonian), and study the overlap~$|\langle \psi^{\textrm{pert}}_{a,b}|\psi^{\textrm{ref}}_{a_0,b_0}\rangle|^2$ as a function of the parameters~$(a,b)$ in the neighborhood of~$(a_0,b_0)$. 
\end{itemize}
\end{description}

The results of our numerical experiments support the following two observations:
\begin{itemize}
\item Hamiltonian interpolation is generically able to prepare approximate ground states of these topological models for sufficiently long total evolution times~$T$.
\item Specific final state(s) show a certain degree of stability with respect to changes in the initial Hamiltonian.  
The theoretical reasoning based on perturbation theory presented in Section~\ref{sec:twodimensionalsystems} provides a partial explanation of this phenomenon.
\end{itemize}

\subsection{A symmetry of the 12-qubit rhombic torus\label{sec:symmetryv}}
As discussed in Section~\ref{sec:microscopicmodels}, the ground space of $\Htop$ on a torus is $4$-dimensional for the toric code, the doubled semion- and the Fibonacci model. In this section, we argue that adiabatic interpolation starting from a translation-invariant Hamiltonian (as considered here) yields states belonging to a two-dimensional subspace of this ground space, thus providing a simplification.

Consider again the 12-qubit rhombic torus illustrated in Fig.~\ref{fig_lattice}. A  $\pi/3$ rotation permuting the physical qubits according to
\begin{align}
(1,2,3,4,5,6,7,8,9,10,11,12)\mapsto (5,7,8,6,9,12,11,10,2,4,1,3)
\end{align} defines a unitary~$U_{\pi/3}$ on $\mathbb{C}^{\otimes 12}$. 
Because of translation-invariance, this is a symmetry of the trivial Hamiltonian, $U_{\pi/3} \Htriv U_{\pi/3}^\dagger=\Htriv$, and it can easily be verified that
for the models considered here, the unitary~$U_{\pi/3}$ also commutes with $\Htop$. Because of the product form of the initial state~$\Psi(0)$, it thus follows that $U_{\pi/3}\Psi(t)=\Psi(t)$ along the whole trajectory $t\mapsto \Psi(t)$ of adiabatic interpolation. In particular, the projection of the final state~$\Psi(T)$ onto the ground space of~$\Htop$ is supported on the $+1$-eigenspaces  space of~$U_{\pi/3}$.

As discussed in Section~\ref{sec:stringoperatorstorus}, a $\pi/3$-rotation of the rhombic torus corresponds to the modular transformation~$ts^3ts$. Since $U_{\pi/3}$ realizes this transformations, its restriction to the ground space of~$\Htop$ can be computed from the~$T$ and $S$-matrices. That is, expressed in the flux bases discussed in Section~\ref{sec:microscopicmodels}, the action of $U_{\pi/3}$ on the ground space is given by the matrix~$TS^3TS$, where 
$(S,T)$ are given by~\eqref{eq:sttoric} for the toric code, as well as~\eqref{eq:stsemion} and~\eqref{eq:doubledfibst} for the doubled semion and Fibonacci models, respectively. The specific form of $T S^3 T S$ or its eigenvectors is not particularly elucidating, but may be computed explicitly.

Importantly, the $+1$~eigenspace
of $T S^3 T S$ is two-dimensional for the toric code, the doubled semion and the Fibonacci models. (In the case of the toric code, it can be verified that this eigenspace  is contained in the logical symmetric subspace. The latter is the subspace invariant under swapping the two logical qubits in the standard computational basis.) As a result, the projection of the state~$\Psi(T)$ onto the ground space of~$\Htop$ belongs to a known two-dimensional subspace which can be explicitly computed. This means that we may characterize the resulting state in terms of a restricted reduced set of logical observables, a fact we will exploit in Section~\ref{sec:toric_numerics}.

\subsection{The toric code \label{sec:toric_numerics}}
As discussed in Section~\ref{sec:toriccodemicroscopic},  
for  the toric code on the honeycomb lattice (see Fig.~\ref{fig_lattice}), the Hamiltonian of the model is $\Htop=-(\sum_p B_p+\sum_v A_v)$, where $B_p=X^{\otimes 6}$ is a tensor product of Pauli-$X$~operators on the six edges of the plaquette~$p$, and $A_v=Z^{\otimes 3}$ is a tensor product of Pauli-$Z$~operators  on the three edges incident on the vertex~$v$.
We point out that the toric code on a honeycomb lattice has several differences compared to a toric code on a square lattice (which is often considered in the literature). Assuming that both lattices are defined with periodic boundary conditions,
\begin{enumerate}[(i)]
\item there are twice as many vertices compared to plaquettes on a honeycomb lattice (as opposed to the same number on a square lattice)
\item 
the vertex terms $A_v=Z^{\otimes 3}$ of the Hamiltonian have odd weights (as opposed to even weight for the square lattice)
\item the weight of a logical minimal $\bar{X}$-string operator (i.e. the number of spins it acts on) is roughly twice as large compared to the corresponding minimal $\bar{Z}$-string operator on the dual lattice (as opposed to the
 square lattice, where both operators have the same weight). For the 12-qubit code of Fig.~\ref{fig_lattice}, an example of such a pair $(\bar{X},\bar{Z})$ of lowest-weight logical operators is given below in Eq.~\eqref{eq:logicalstateexpectation}. 
\end{enumerate}
Properties (i) and (ii) imply that the usual symmetries $X\leftrightarrow Z$ 
and $Z\leftrightarrow -Z$ of the toric code on the square lattice are not present in this case. The absence of these symmetries is reflected in our simulations.  Property~(iii) also directly affects the final state, as can be seen by the perturbative reasoning of Section~\ref{sec:perturbation_TQFT}:  $\bar{Z}$-string operators appear in lower order in perturbation theory compared to $\bar{X}$-string operators.

\paragraph{(Non)-adiabaticity.} 
\begin{figure}
\centering
\includegraphics[scale=0.6]{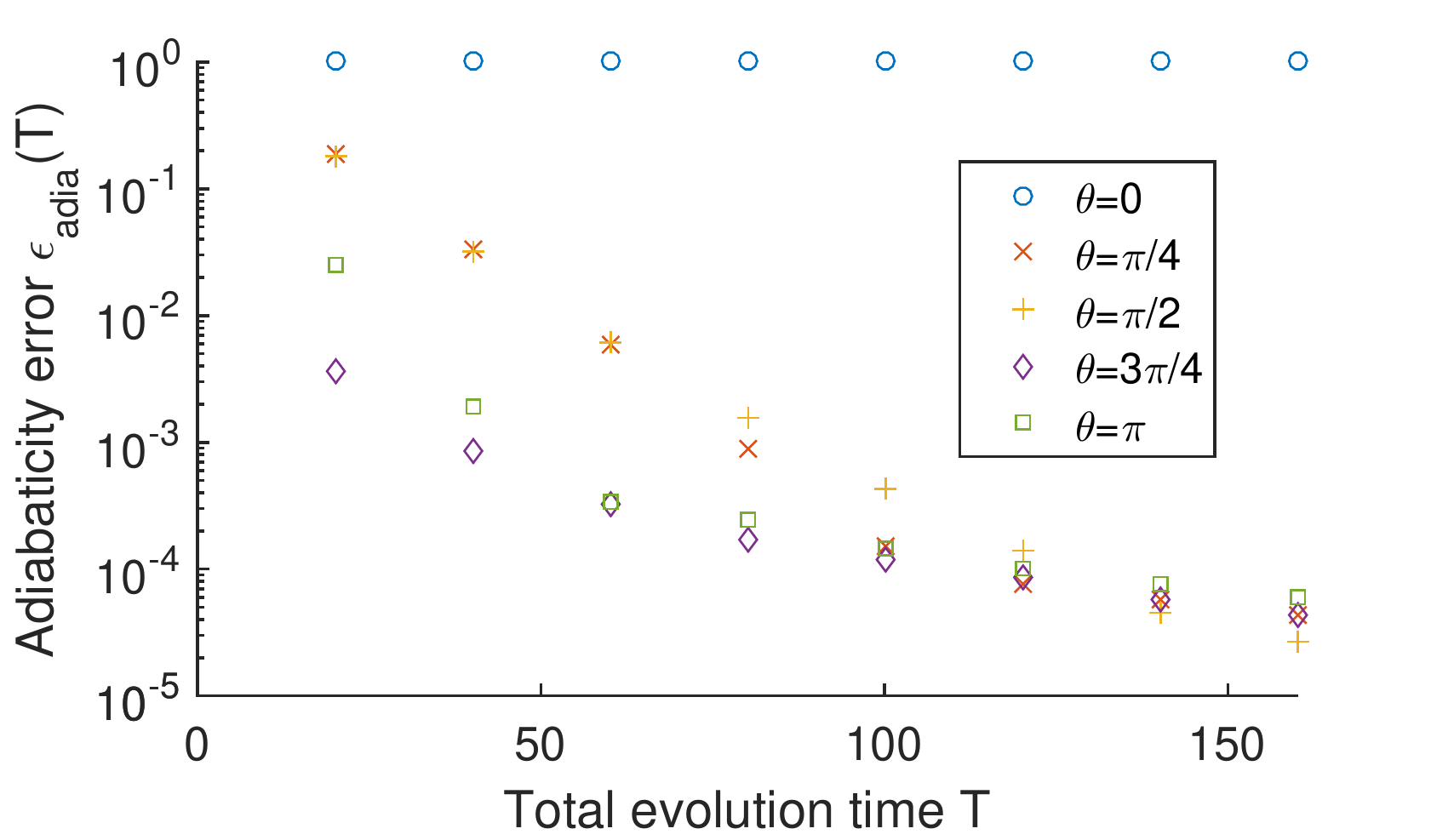}
\caption{\label{fig_TC_groundspaceoverlap} 
This figure gives the adiabaticity error $\nonadiabaticity(T)=1-\langle\Psi(T)|P_0(T)|\Psi(T)\rangle$ (cf.~\protect\eqref{eq:adiabaticityerrordef}) as a function of the total evolution time~$T$ and the initial Hamiltonian chosen. For the latter, we consider the one-parameter family~$\Htriv(\theta)$ given by~\protect\eqref{eq:htrivialthetdef}.  For $\theta=0$, the adiabatic evolution is not able to reach the final ground space because initially $\langle A_v \rangle=-1$ for every vertex operator $A_v=Z^{\otimes 3}$, and this quantity is conserved during the evolution.  This is a feature of the honeycomb lattice because the vertex terms~$A_v$ have odd weights. For other values of~$\theta$, the ground space is reached for sufficiently large total evolution times~$T$. }
\end{figure}
We first present the adiabaticity error $\nonadiabaticity(T)$ for the Hamiltonian~$\Htriv(\theta)$ given by~\eqref{eq:htrivialthetdef}  (for different values of~$\theta$) as a function of the total evolution time~$T$. 
 Fig.~\ref{fig_TC_groundspaceoverlap} illlustrates the result. It shows that for sufficiently long total evolution times~$T$, the Hamiltonian interpolation  reaches the ground space of the toric code when the  initial Hamiltonian is
 $\Htriv(\theta=\pi)=-\sum_i Z_i$; this is also the case for~$\theta\in \{\pi/4,\pi/2,3\pi/4\}$. 

However, if the initial Hamiltonian is $\Htriv(\theta=0)=\sum_i Z_i$, then the final state~$\Psi(T)$ is far from the ground space of the toric code Hamiltonian~$\Htop$. This phenomenon has a simple explanation along the lines of Observation~\ref{obs:groundstatenotreached}. Indeed, if~$\theta=0$, then every  vertex terms~$A_v=Z^{\otimes 3}$ commutes with both~$\Htriv$ as well as $\Htop$ (and thus
all intermediate Hamiltonians~$H(t)$). In particular, the expectation value of the vertex terms remains constant throughout the whole evolution, and this leads to an adiabaticity error~$\nonadiabaticity(T)$ of~$1$ in the case of $\Htriv(\theta=0)=\sum_i Z_i$.

In Figs.~\ref{fig_TC_plusZhemisphere_gsoverlap},
~\ref{fig_TC_minusZhemisphere_gsoverlap}, we consider neighborhoods of  Hamiltonians of the form (cf.~\eqref{eq:hplusminusdef})
\begin{align}
\Htriv^+(a,b)\qquad&\textrm{ around }\qquad\Htriv^+(0,0)=\Htriv(\theta=0)=\sum_j Z_j\qquad \textrm{ and }\\
\Htriv^-(a,b)\qquad&\textrm{ around }\qquad \Htriv^-(0,0)=\Htriv(\theta=\pi)=-\sum_j Z_j\ .
\end{align}
The initial Hamiltonians~$\Htriv(\theta=0)$
and $\Htriv(\theta=\pi)$ correspond 
 to the center points in  Fig.~\ref{fig_TC_plusZhemisphere_gsoverlap} and~\ref{fig_TC_minusZhemisphere_gsoverlap}, respectively. 
\begin{itemize}
\item In the first case (Fig.~\ref{fig_TC_plusZhemisphere_gsoverlap}), 
we observe that for all 
initial Hamiltonians of the form~$\Htriv^+(a,b)$ in a small neighborhood of~$\Htriv^+(0,0)$, the adiabaticity error~$\nonadiabaticity(T)$ is also large, but drops off quickly outside that neighborhood. This is consistent with the relevant level crossing(s) being avoided by introducing generic perturbations to the initial Hamiltonian. 
\item
In contrast, almost all initial Hamiltonians in the family $\Htriv^-(a,b)$ (around the initial Hamiltonian~$\Htriv^-(0,0)$) lead to a small adiabaticity error~$\nonadiabaticity(T)$ (Fig.~\ref{fig_TC_minusZhemisphere_gsoverlap}), demonstrating the stability of the adiabatic preparation.
\end{itemize}

\begin{figure}
\begin{subfigure}[t]{0.5\textwidth}
\centering\captionsetup{width=.8\linewidth}
\includegraphics[width=\textwidth]{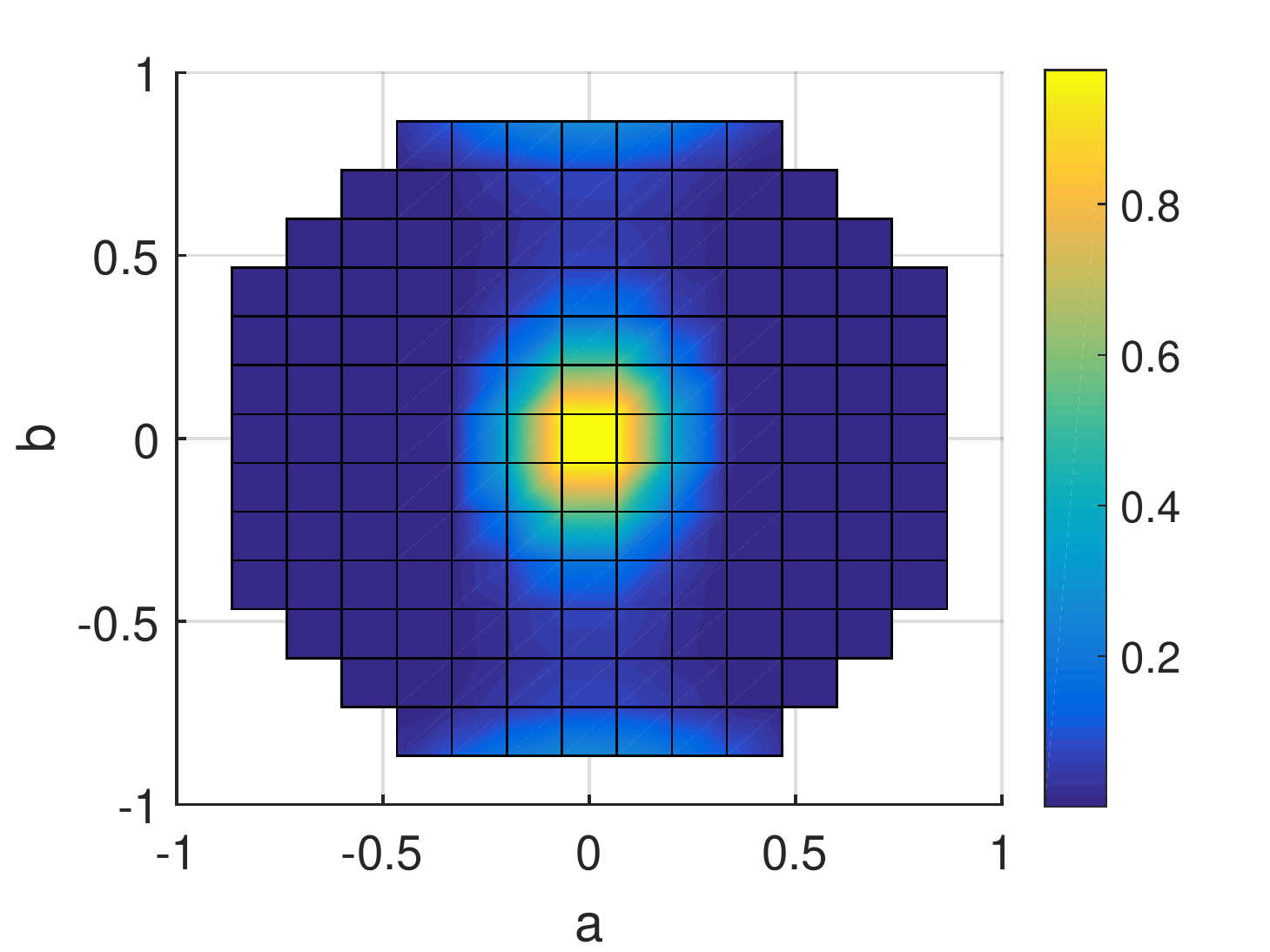}
\caption{
The adiabaticity error~$\nonadiabaticity(T)$ in the neighborhood around $\Htriv^+(0,0)=\sum_i Z_i$ for different Hamiltonians~$\Htriv^+(a,b)$.
As explained, the evolution cannot reach the ground space of the toric code around~$(a,b)=(0,0)$  because the expectation values of plaquette-operators are preserved. \label{fig_TC_plusZhemisphere_gsoverlap}}
\end{subfigure}
\begin{subfigure}[t]{0.5\textwidth}
\centering\captionsetup{width=.8\linewidth}
\includegraphics[width=\textwidth]{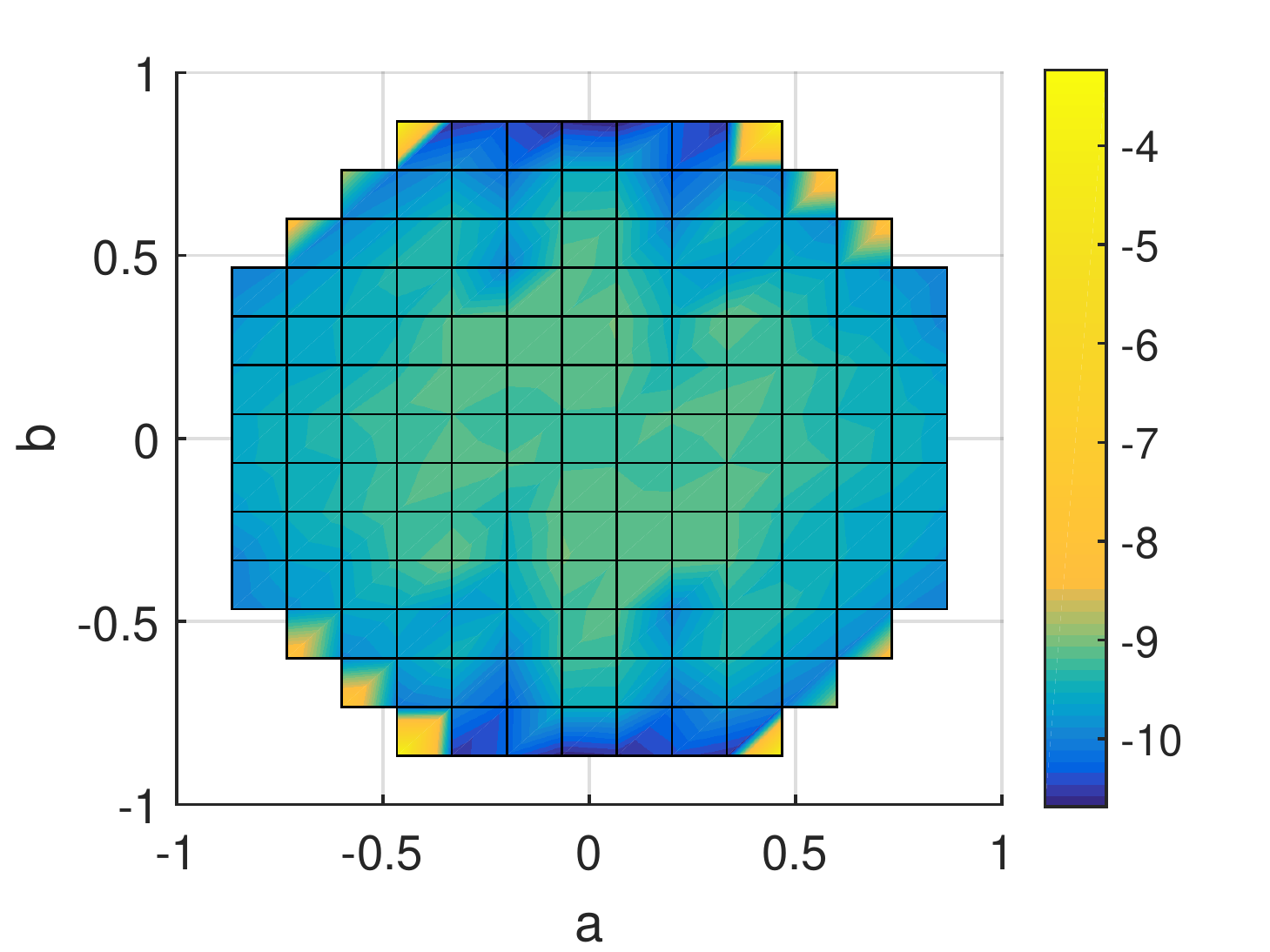}
\caption{The logarithm~$\ln\nonadiabaticity(T)$ of the adiabaticity error 
in the neighborhood around~$\Htriv^-(0,0)=-\sum_i Z_i$ for different Hamiltonians~$\Htriv^-(a,b)$. Here we use a log-scale because the variation in values is small. The ground space of the toric code Hamiltonian~$\Htop$ is reached for 
almost the entire parameter region.\label{fig_TC_minusZhemisphere_gsoverlap}}
\end{subfigure}\\
\begin{center}
\begin{subfigure}[t]{1.0\textwidth}
\centering\captionsetup{width=.8\linewidth}
\includegraphics[width=0.5\textwidth]{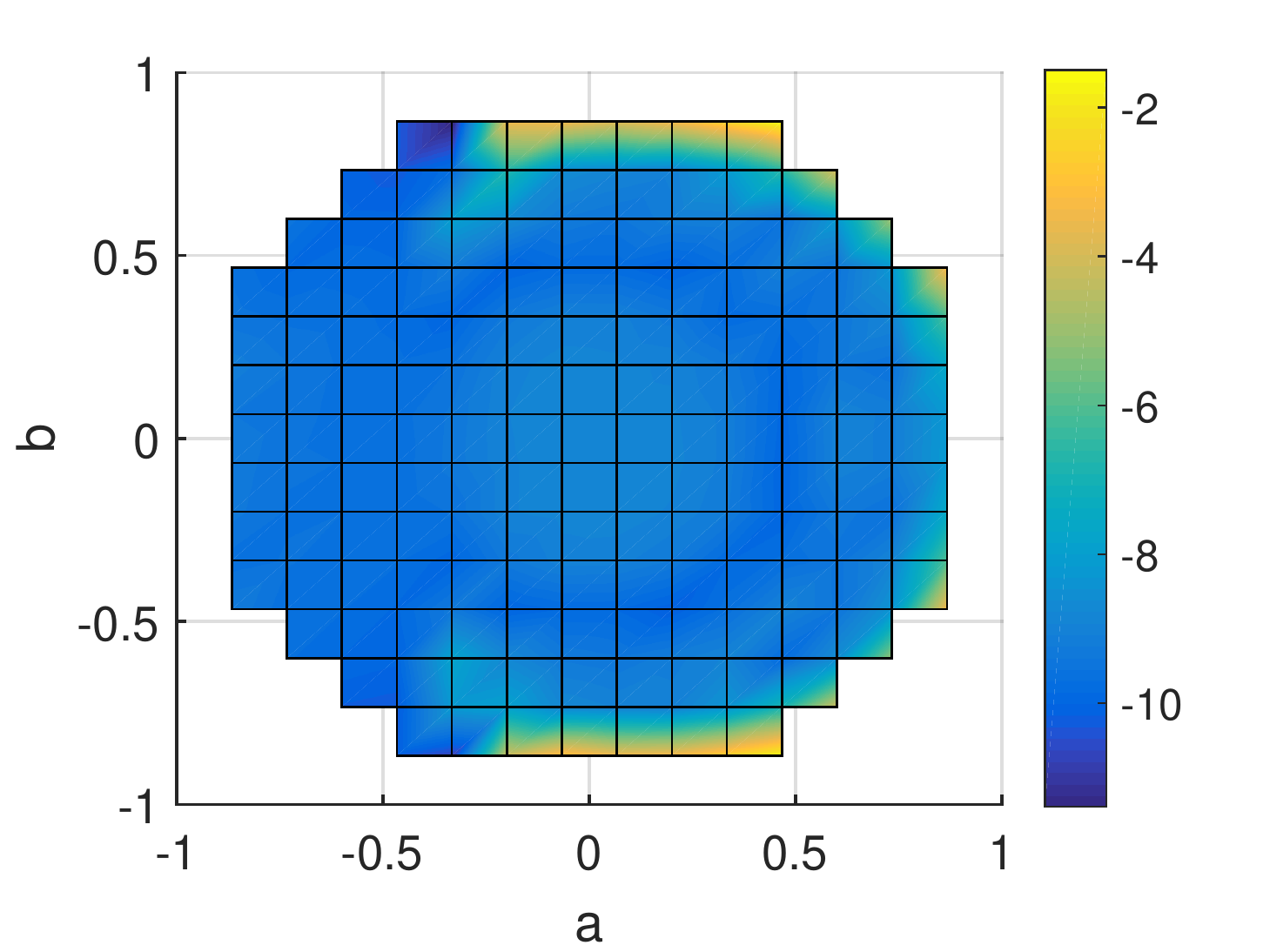}
\caption{
The logarithm of adiabaticity error~$\protect\ln \protect\nonadiabaticity(T)$ in the neighborhood around $\protect\Htriv^{-X}(0,0)=-\protect\sum_j X_j$ for different Hamiltonians~$\protect\Htriv^{-X}(a,b)$. 
Note that the resulting figure would look identical for the Hamiltonians 
$\protect\Htriv^{+X}(a,b)$ because of the $-X \protect\leftrightarrow +X$ symmetry.
}
\label{fig:nonadiabaticitygeneralfX}
\end{subfigure}
\end{center}

\caption{The adiabaticity error~$\nonadiabaticity(T)=1-\langle\Psi(T)|P_0(T)|\Psi(T)\rangle$, measuring how well the final state~$\Psi(T)$ overlaps with the ground space of the  toric code.  All three figures are for a total evolution time~$T=120$. In Fig.~\protect\ref{fig_TC_plusZhemisphere_gsoverlap}, we consider the family of initial Hamiltonians $\Htriv^+(a,b)$ in the neighborhood of $\Htriv^+(0,0)=\Htriv(\theta=0)=\sum_j Z_j$. 
In contrast, Fig.~\protect\ref{fig_TC_minusZhemisphere_gsoverlap} illustrates
different choices of initial Hamiltonians~$\Htriv^-(a,b)$ around $\Htriv^-(0,0)=\Htriv(\theta=\pi)=-\sum_j Z_j$. 
The values $(a,b)\subset\mathbb{R}^2$ are restricted to the unit disc~$a^2+b^2\leq 1$; the center points of the two figures correspond respectively to~$\theta=0$ and $\theta=\pi$ in Fig.~\protect\ref{fig_TC_groundspaceoverlap}. 
Finally, Fig.~\protect\ref{fig:nonadiabaticitygeneralfX} gives the non-adiabaticity error for initial Hamiltonians  of the form~$\Htriv^{-X}(a,b)$  (as defined in Eq.~\protect\eqref{eq:Hpmdef}).
\label{fig:nonadiabaticitygeneralf}
}
 \end{figure}

In a similar vein, Fig.~\ref{fig:nonadiabaticitygeneralfX} illustrates the non-adiabaticity for the family of Hamiltonian
\begin{align}
\Htriv^{- X}(a,b)&=-(1-a^2-b^2)^{1/2}\sum_j X_j+b\sum_j Y_j+a\sum_jZ_j \ . \label{eq:Hpmdef}
\end{align}
The family $\Htriv^{+ X}(a,b)$ (defined with a positive square root) would behave exactly the same due to the symmetry $+X \leftrightarrow -X$.

\paragraph{Logical state.} 
For the  12-qubit rhombic toric code (Fig.~\ref{fig_lattice}), logical observables associated with the two encoded logical qubits can be chosen as 
\begin{align}
\begin{matrix}
\bar{X}_1&=&X_7X_8X_{11}X_{12}\\
\bar{Z}_1&=&Z_{10}Z_{12}\\
\end{matrix}\qquad\textrm{ and }\qquad
\begin{matrix}
\bar{X}_2&=&X_4X_0X_2X_{12}\\
\bar{Z}_2&=&Z_1Z_2
\end{matrix}\ .
\end{align}
Because of the symmetry~\eqref{sec:symmetryv}, however, these are not independent for a state~$\Psi(T)$  (or more precisely, its projection $P_0(T)\Psi(T)$) prepared by Hamiltonian interpolation from a product state: their expectation values satisfy the identities
\begin{align}
\langle \bar{Z}_1\rangle=\langle\bar{Z}_2\rangle\qquad\textrm{ and }\qquad \langle \bar{X}_1\rangle=\langle\bar{X}_2\rangle\ .
\end{align}
We will hence use the two (commuting) logical operators 
\begin{align}
\bar{X}=\bar{X}_1=X_7X_8X_{11}X_{12}\qquad\textrm{  and }\qquad \bar{Z}=\bar{Z}_2=Z_1Z_2\label{eq:logicalstateexpectation}
\end{align} to describe the obtained logical state.

In Fig.~\ref{fig_TC_minusXZhemishpere_XZlogical}, we plot the expectation values of $\bar{Z}$ and $\bar{X}$ in the final state~$\Psi(T)$ for initial Hamiltonians of the form  (cf.~\eqref{eq:hplusminusdef} and ~\eqref{eq:Hpmdef})
\begin{align}
\Htriv^-(a,b)\qquad &\textrm{ around }\qquad\Htriv^-(0,0)=-\sum_j Z_j\\
\Htriv^{-X}(a,b) \qquad&\textrm{ around }\qquad \Htriv^{-X}(0,0)=-\sum_j X_j
\end{align}
 We again discuss the center points in more detail.  It is worth noting that the single-qubit~$\{Z_i \}$ operators correspond to the local creation, hopping  and annihilation of $\bm{m}$ anyons situated on plaquettes, whereas the operators~$\{X_i\}$ are associated with creation, hopping and annihilation of~$\bm{e}$ anyons situated on vertices. In particular, this means that the initial Hamiltonians associated with the center points  in the two figures each generate  processes involving only either type of anyon.
\begin{itemize}
\item
For $\Htriv^-(0,0)=-\sum_i Z_i$, we know that $\langle \bar{Z} \rangle =1$ during the entire evolution because~$\bar{Z}$ commutes with the Hamiltonians~$H(t)$,
and the initial ground state~$\Psi(0)$ is a $+1$~eigenstate of~$\bar{Z}$. 
In Figs.~\ref{fig_TC_minusZhemishpere_Xlogical} and~\ref{fig_TC_minusZhemishpere_Zlogical}, we can see that there is a large region of initial Hamiltonians~$\Htriv^-(a,b)$ around $\Htriv^-(0,0)=-\sum_i Z_i$ which lead to approximately the same final state. 
\item
On the other hand, as shown in Figs.~\ref{fig_TC_minusXhemishpere_Xlogical} and~\ref{fig_TC_minusXhemishpere_Zlogical}, the stable region of Hamiltonians~$\Htriv^{-X}(a,b)$ around the initial Hamiltonian $\Htriv^{-X}(0,0)=-\sum_i X_i$ is much smaller.  This is due to the fact that the operator~$\bar{X}$  appears in higher order perturbation expansion compared to~$\bar{Z}$, and the evolution time~$T$ is taken to be quite long. Given sufficiently large total evolution time~$T$,
in the neighborhood of~$\Htriv^{-X}(0,0)=-\sum_i X_i$,
 the lower order term~$\bar{Z}$ 
in the effective Hamiltonian will dominate the term~$\bar{X}$ associated with~$V=-\sum_i X_i$.
\end{itemize}

However, in both cases considered in Fig.~\ref{fig_TC_minusXZhemishpere_XZlogical}, we observe that one of two specific logical states is prepared with great precision within a significant fraction of the initial Hamiltonian parameter space. 
\begin{figure}
\begin{subfigure}	[t]{0.5\textwidth}
	\includegraphics[width=\textwidth]{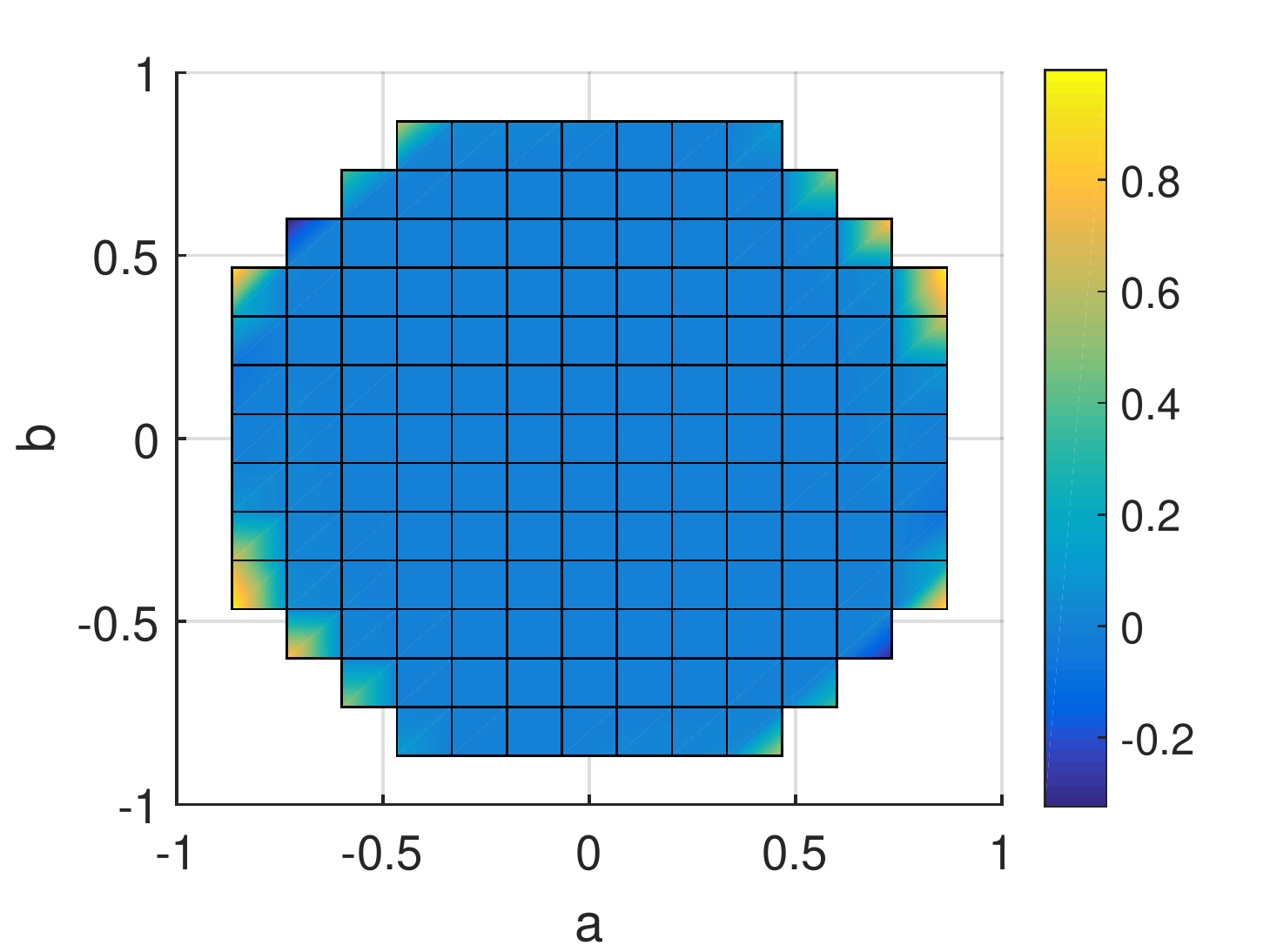}
	\caption{
		The expectation value
		$\langle \bar{X}\rangle$ of the final state~$\Psi(T)$,  
		for initial Hamiltonians~$\Htriv^-(a,b)$ in the neighborhood of $\Htriv^-(0,0)=-\sum_j Z_j$. 
Note that, as illustrated in Fig.~\protect\ref{fig_TC_minusZhemisphere_gsoverlap}, the ground space of the toric code is reached for the whole parameter range; hence these values, together with the expectation values shown in Fig.~\protect\ref{fig_TC_minusZhemishpere_Zlogical} uniquely determine the  state~$\Psi(T)$.
		 \label{fig_TC_minusZhemishpere_Xlogical}}
\end{subfigure}
\hspace*{\fill}\qquad 
\begin{subfigure}	[t]{0.5\textwidth}
\includegraphics[width=\textwidth]{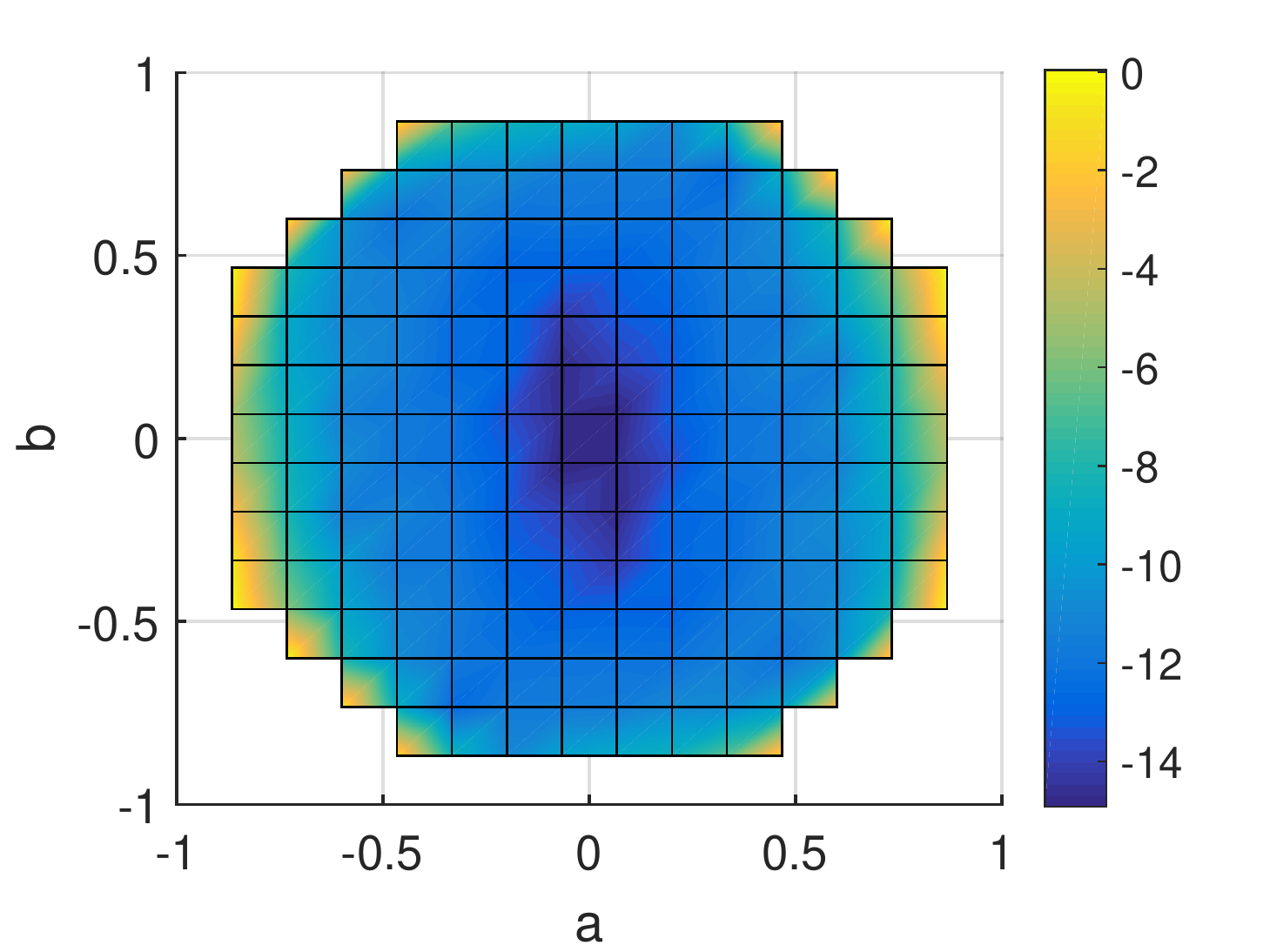}
\caption{
The quantity $\ln (1-\langle \bar{Z}\rangle)$
	for initial Hamiltonians~$\Htriv^-(a,b)$ (we plot the logarithm because the variation is small) as in Fig.~\protect\ref{fig_TC_minusZhemishpere_Xlogical}. \label{fig_TC_minusZhemishpere_Zlogical}
} 
\end{subfigure}

\begin{subfigure}[t]{0.5\textwidth}
\includegraphics[width=\textwidth]{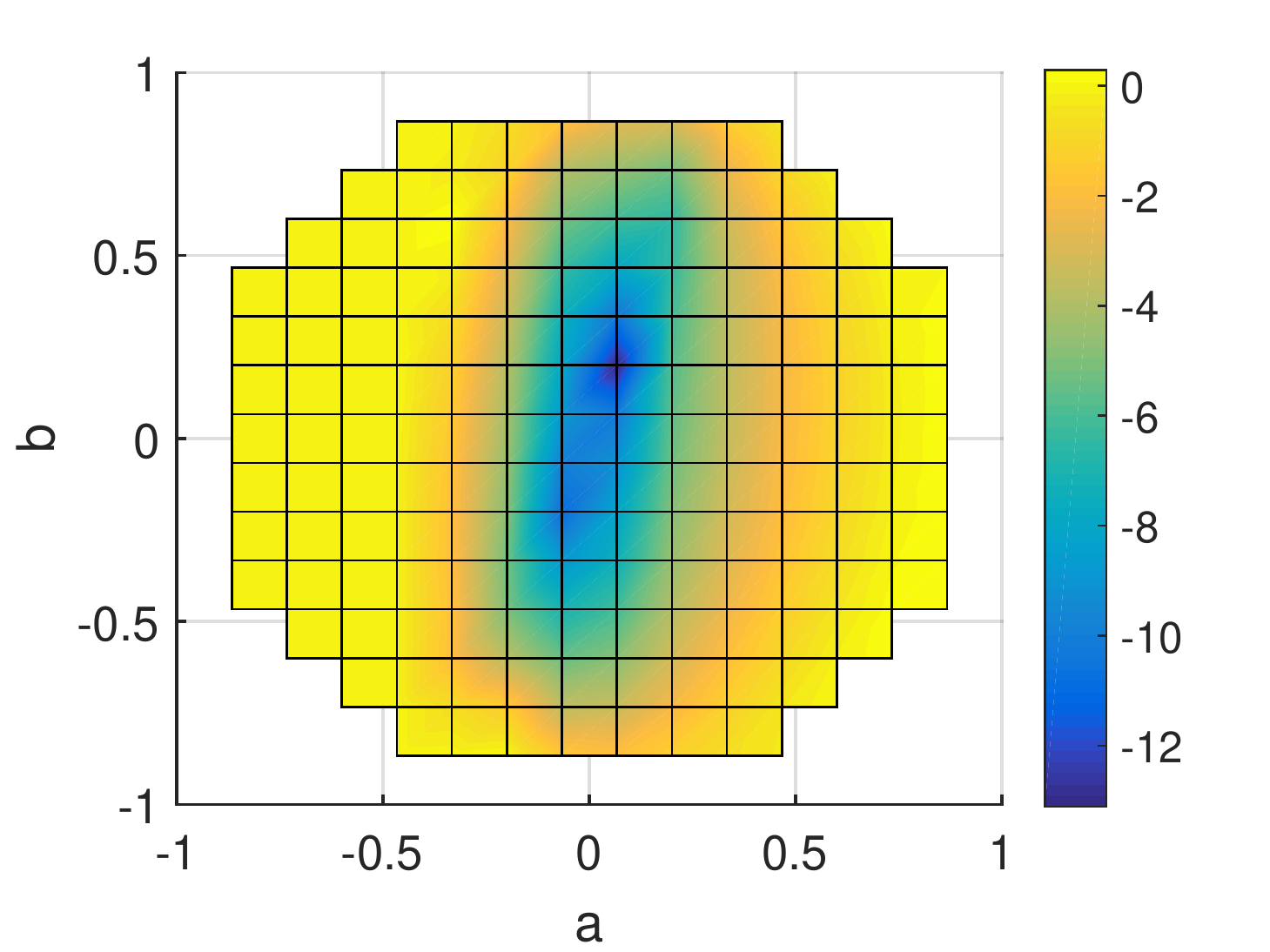}
\caption{
The quantity
$\ln (1-\langle \bar{X}\rangle)$ 
for initial Hamiltonians~$\Htriv^{-X}(a,b)$ in the neighborhood of $\Htriv^{-X}(0,0)=-\sum_j X_j$. The corresponding adiabaticity error is shown in Fig.~\protect\ref{fig:nonadiabaticitygeneralfX}. 
\label{fig_TC_minusXhemishpere_Xlogical}}
\end{subfigure}
\hspace*{\fill}\qquad 
\begin{subfigure}[t]{0.5\textwidth}
	\includegraphics[width=\textwidth]{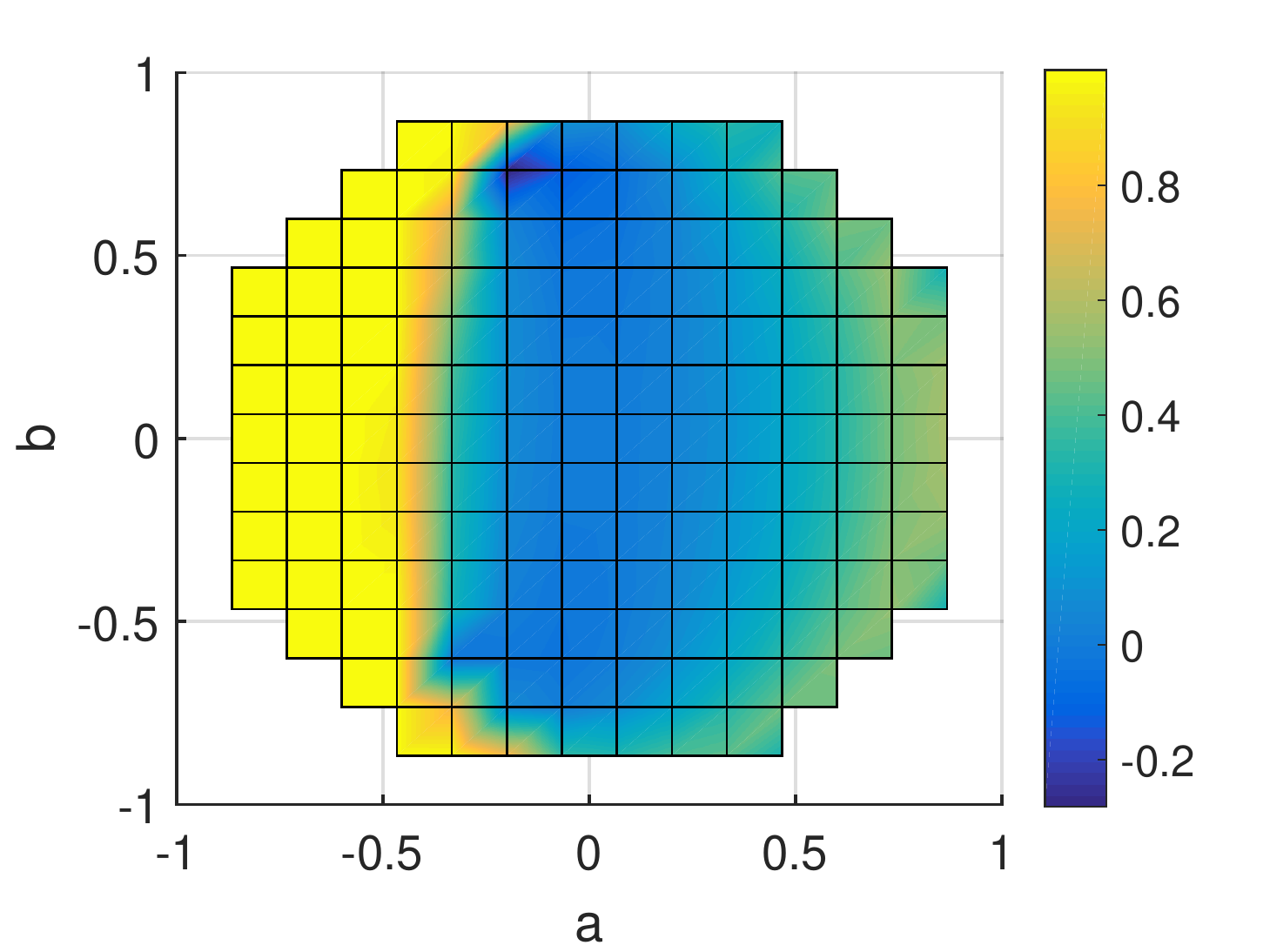}
	\caption{
		The quantity
		$\langle \bar{Z}\rangle$
		for initial Hamiltonians~$\Htriv^{-X}(a,b)$ as in Fig.~\protect\ref{fig_TC_minusXhemishpere_Xlogical}.
		\label{fig_TC_minusXhemishpere_Zlogical}}
\end{subfigure}

\caption{ These figures 
illustrate the expectation values~$\langle \bar{X}\rangle$ and $\langle \bar{Z}\rangle$ of string-operators (cf.~\protect\eqref{eq:logicalstateexpectation}) of the final state~$\Psi(T)$, for different choices of the initial Hamiltonian.  The total  evolution time is $T=120$.
} \label{fig_TC_minusXZhemishpere_XZlogical}
\end{figure}

\clearpage
\subsection{The doubled semion model\label{sec:doubledsemionnumerics}}
In this section, we 
present our numerical results for Hamiltonian interpolation 
in the case of the doubled semion model (see Section~\ref{sec:doubledsemion}).

\paragraph{(Non)-adiabaticity.} 
We first consider the total evolution time~$T$ necessary to reach the final ground space of~$\Htop$, for different initial Hamiltonians~$\Htriv$. Specifically,
Fig.~\ref{fig_DS_groundspaceoverlap} shows the adiabaticity error~$\nonadiabaticity(T)$ (cf.~\eqref{eq:adiabaticityerrordef}) as a function of the total evolution time~$T$ for  the three initial Hamiltonians~$\Htriv(\theta)$, $\theta\in\{\pi,\pi/3,2\pi/3\}$ (cf.~\eqref{eq:htrivialthetdef}).  The case of $\theta=0$, corresponding to the initial Hamiltonian $\Htriv(0)=\sum_j Z_j$ is not shown in Fig.~\ref{fig_DS_groundspaceoverlap} since the situation is the same as in the toric code: No overlap with the  ground space of~$\Htop$ is achieved 
because the vertex-operators $A_v=Z^{\otimes 3}$ are conserved quantities with  $\langle A_v \rangle=-1$.

In Fig.~\ref{fig_DS_plusZhemishpere_gsoverlap}, we plot the adiabaticity error~$\nonadiabaticity(T)$  with   initial Hamiltonian among the family of Hamiltonians~$\Htriv^+(a,b)$ in the vicinity of $\Htriv^+(0,0)=\sum_i Z_i$. Similarly, Fig.~\ref{fig_DS_minusZhemishpere_gsoverlap} provides the adiabaticity error for initial Hamiltonians~$\Htriv^-(a,b)$ in the vicinity of $\Htriv^-(0,0)=-\sum_i Z_i$. 
\begin{figure}
\centering
\includegraphics[scale=0.6]{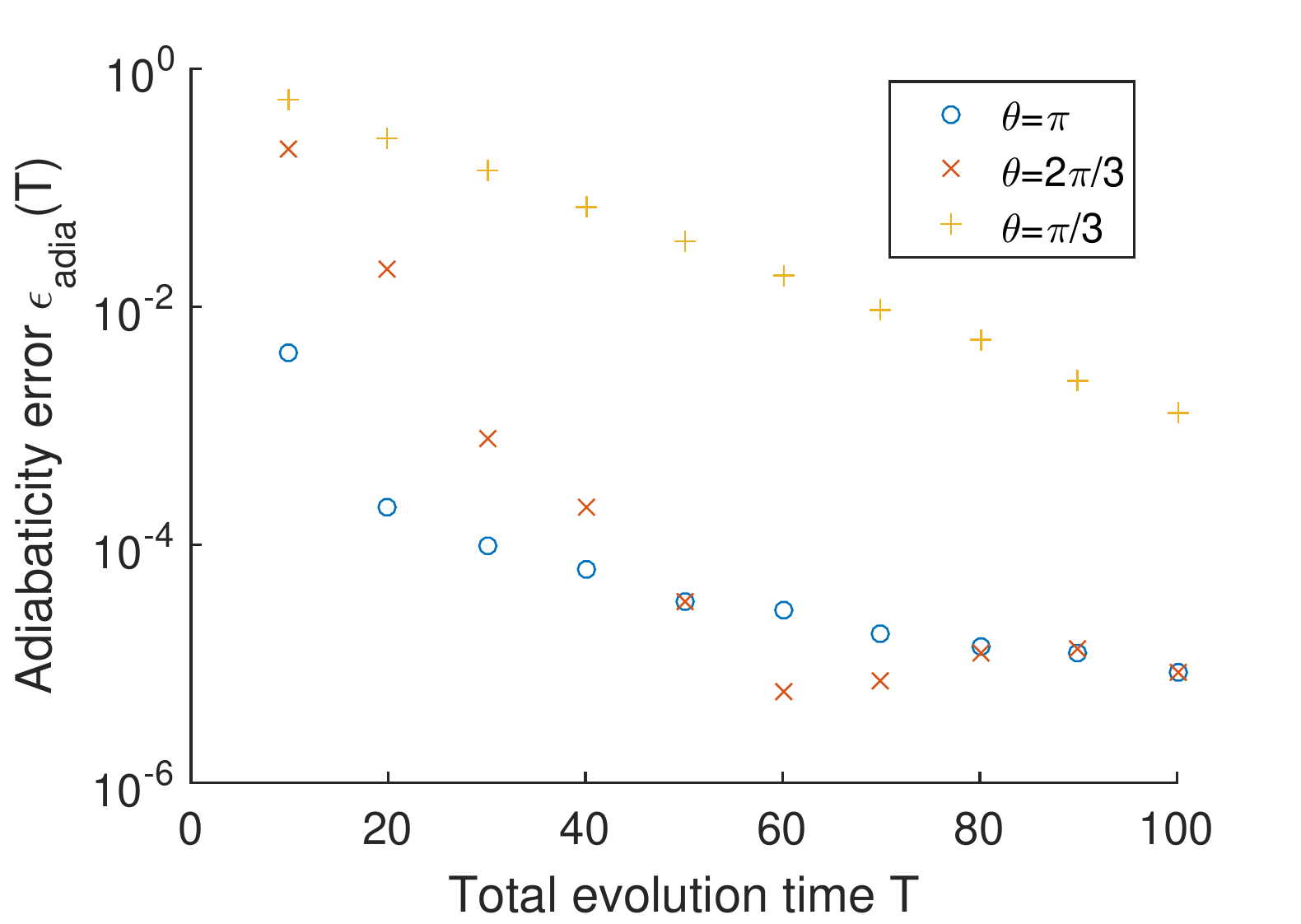}
\caption{
The adiabaticity error~$\nonadiabaticity(T)$ 
for the doubled semion model as a function of the total evolution time~$T$. 
Initial Hamiltonians~$\Htriv(\theta)$  with $\theta\in\{\pi/3,2\pi/3,\pi\}$ are considered.}
\label{fig_DS_groundspaceoverlap}
\end{figure}

\begin{figure}
\begin{subfigure}[t]{0.5\textwidth}
\centering\captionsetup{width=.8\linewidth}
\includegraphics[width=\textwidth]{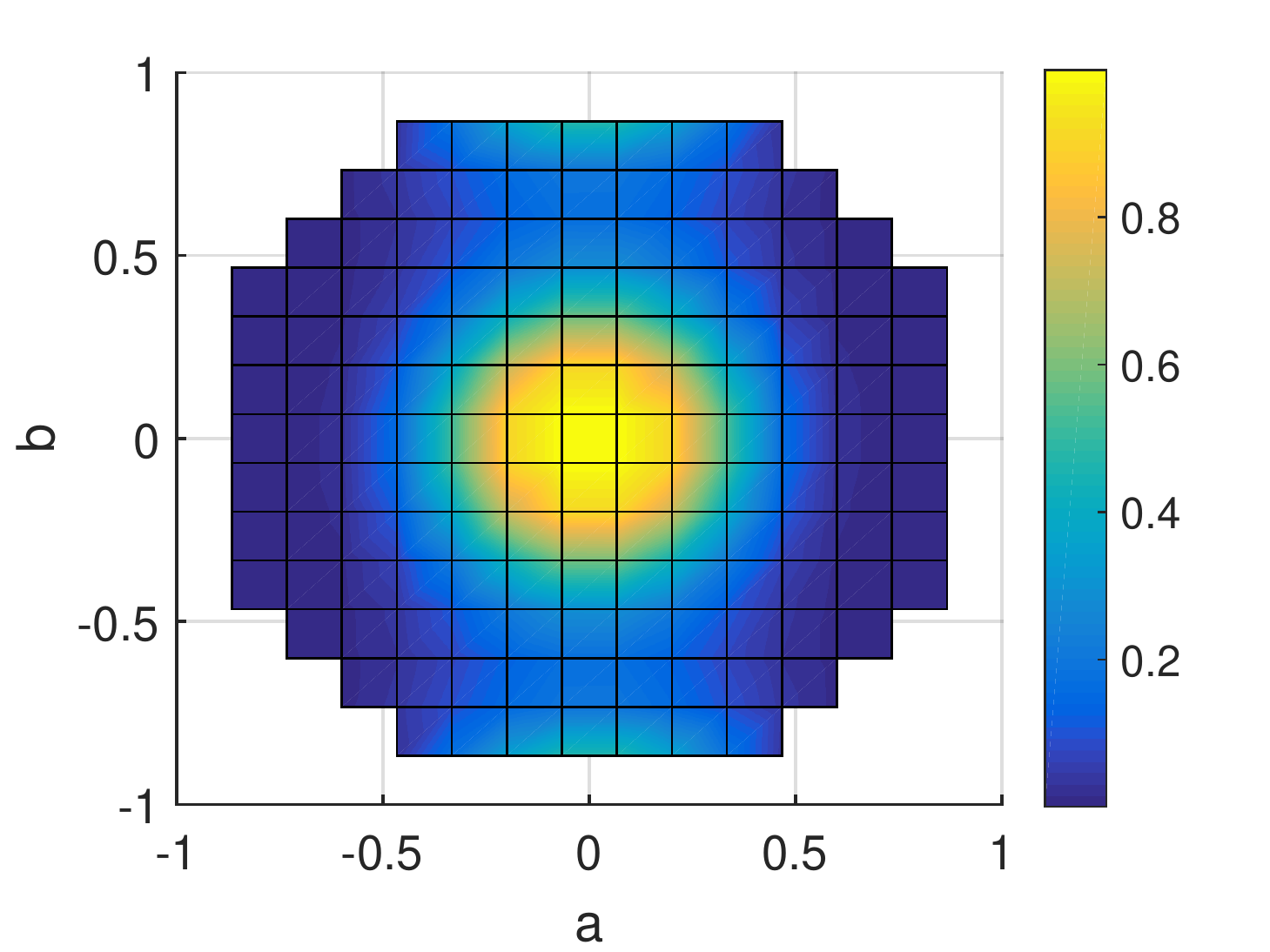}
\caption{
The adiabaticity error~$\nonadiabaticity(T)$  for different Hamiltonians~$\Htriv^+(a,b)$ in the vicinity of~$\Htriv^+(0,0)=\sum_j Z_j$.
The adiabaticity error is maximal for the latter because of conserved quantities; however, it decays rapidly outside this center region.  
This situation is analogous to Fig.~\protect\ref{fig_TC_plusZhemisphere_gsoverlap} for the toric code. \label{fig_DS_plusZhemishpere_gsoverlap}}
\end{subfigure}
\begin{subfigure}[t]{0.5\textwidth}
\centering\captionsetup{width=.8\linewidth}
\includegraphics[width=\textwidth]{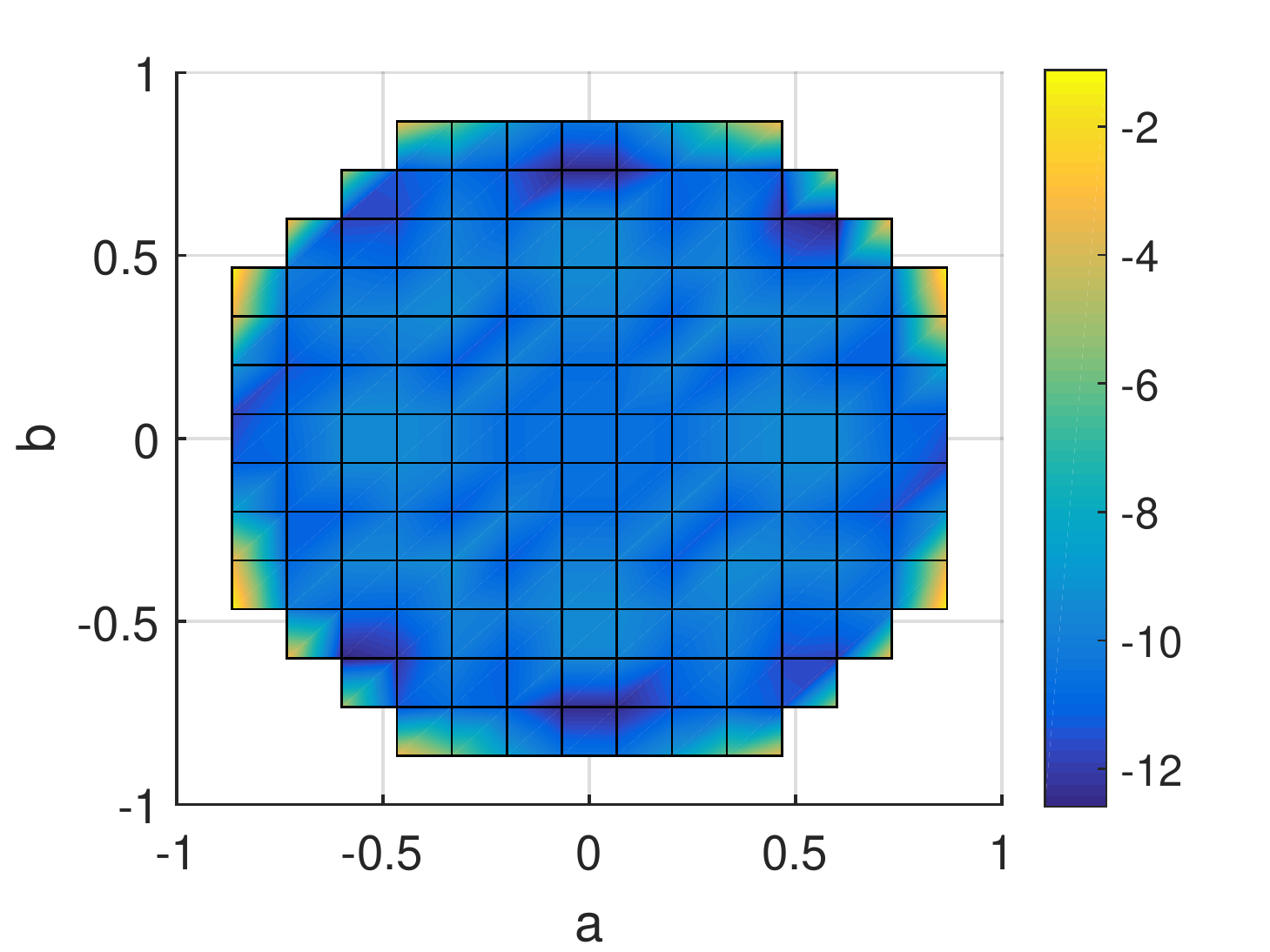}
\caption{The logarithmic adiabaticity error~$\ln \nonadiabaticity(T)$
among the family of Hamiltonians~$\Htriv^-(a,b)$
around~$\Htriv^-(0,0)=-\sum_j Z_j$.\label{fig_DS_minusZhemishpere_gsoverlap}}
\end{subfigure}
\caption{
The adiabaticity error~$\nonadiabaticity(T)$ for different initial Hamiltonians~$\Htriv$ and the doubled semion model as~$\Htop$. In both cases,  the total evolution  time is $T=120$.}
\end{figure}

\begin{figure}
\begin{subfigure}[t]{0.5\textwidth}
\centering\captionsetup{width=.8\linewidth}
\includegraphics[width=\textwidth]{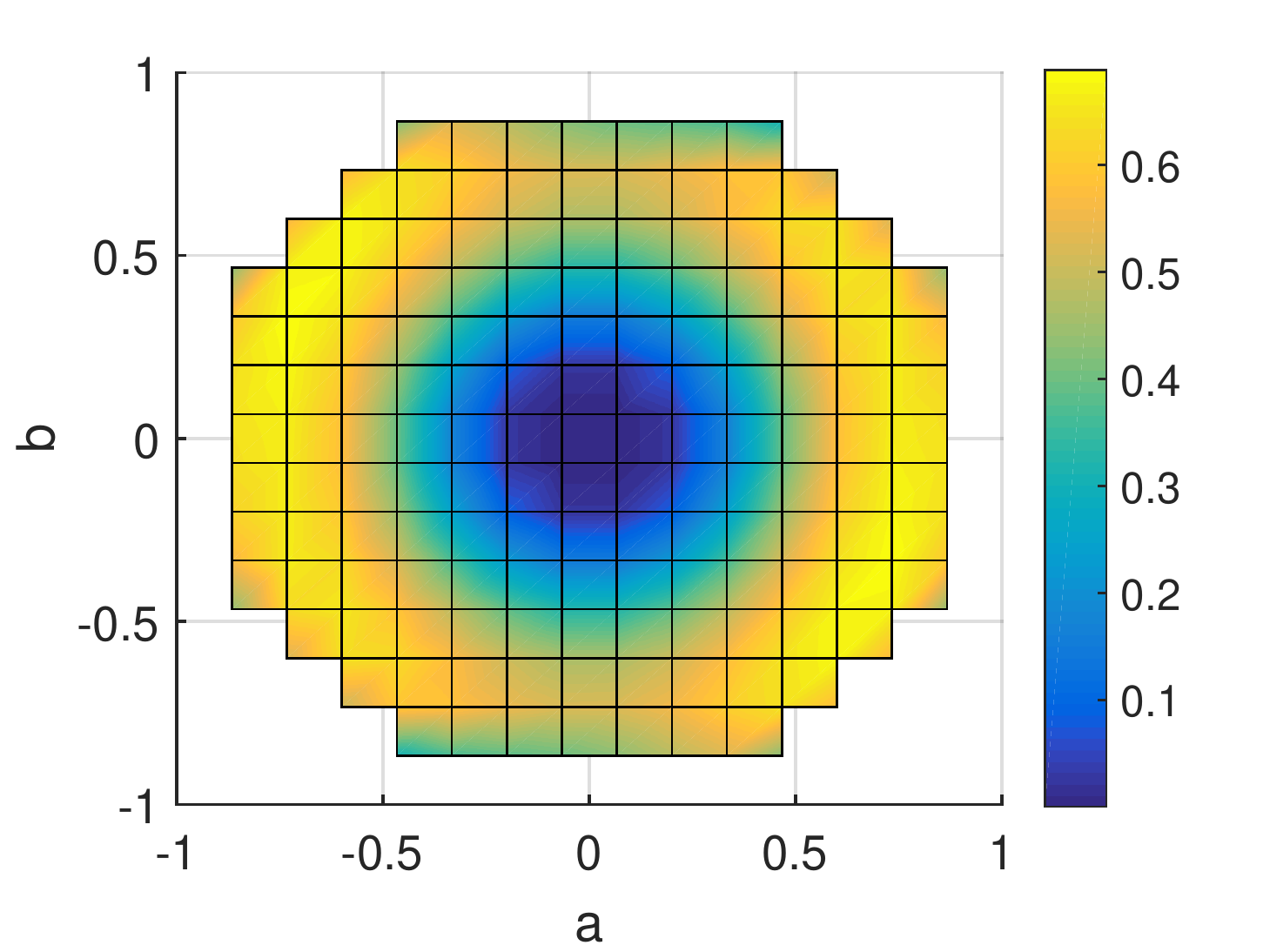}
\caption{The overlap $|\inp{\Psi^{+}_{a,b}(T)}{\psi_R}|^2$ 
for initial Hamiltonians~$\Htriv^+(a,b)$ around $\Htriv^+(0,0)=\sum_i Z_i$.
 We observe that outside the center region (where the ground space of~$\Htop$ is not reached, see Fig.~\protect\ref{fig_DS_plusZhemishpere_gsoverlap}), 
the prepared state~$\Psi^{+}_{a,b}(T)$ is not too far from the reference state~$\psi_R$. 
Note that definition of the latter does not
correspond to any Hamiltonian in this plot, but rather the centerpoint of Fig.\protect~\ref{fig_DS_minusZhemisphere120}. 
 \label{fig_DS_plusZhemishpere_stateoverlap120}}
\end{subfigure}
\begin{subfigure}[t]{0.5\textwidth}
\centering\captionsetup{width=.8\linewidth}
\includegraphics[width=\textwidth]{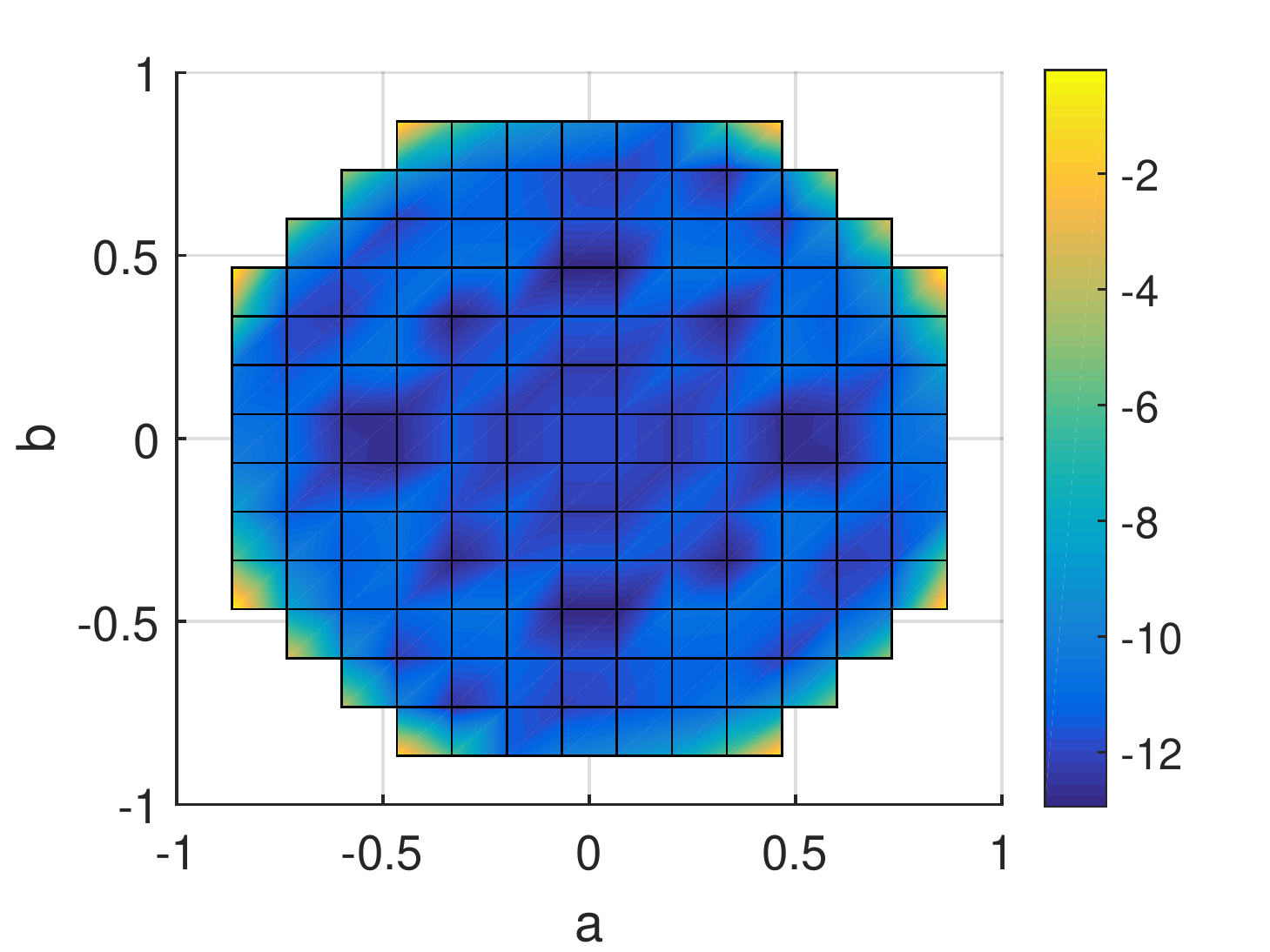}
\caption{The quantity $\ln (1-|\inp{\Psi^{-}_{a,b}(T)}{\psi_R}|^2)$
 for initial Hamiltonians~$\Htriv^-(a,b)$ around $\Htriv^-(0,0)=-\sum_i Z_i$.
We plot the logarithm of this quantity because the variation is small. As 
illustrated, the resulting state is close to the reference state~$\psi_R$  throughout most of the parameter region. Observe that, while $\psi_R$ corresponds to the center point in this figure, it still deviates from~$\Psi^{-}_{a,b}$ since the latter has support outside the ground space of~$\Htop$ (cf. Fig.~\protect\ref{fig_DS_minusZhemishpere_gsoverlap}).   \label{fig_DS_minusZhemisphere120}}
\end{subfigure}

\caption{The overlaps~$|\langle \Psi^{\pm}_{a,b}(T)\ket{\psi_R}|^2$ between the final states $\Psi^{\pm}_{a,b}(T)$ of Hamiltonian interpolation and the reference state $\psi_R$ (cf.~\protect\eqref{eq:referencestatedefsem}). Observe that the same reference state is used in both figures even though~$\psi_R$ is naturally associated with the centerpoint in Fig.~\protect\ref{fig_DS_minusZhemisphere120}.  The total  evolution time is $T=120$ in both cases.\label{fig:DSZhmispherestateoverlap} Comparing with Figs.~\protect\ref{fig_DS_plusZhemishpere_gsoverlap} and~\protect\ref{fig_DS_minusZhemishpere_gsoverlap}, we conclude that throughout the region where the ground space of~$\Htop$ is reached, approximately same state is prepared.}
\end{figure}

\paragraph{Logical state.} 
To explore the stability of the resulting final state, we 
consider the family of initial Hamiltonians~$\Htriv^{\pm}(a,b)$ and
compute the overlap~$|\langle \Psi^{\pm}_{a,b}(T)\ket{\psi_R}|^2$  of the resulting final state~$\Psi^{\pm}_{a,b}(T)$ with a suitably chosen reference state~$\psi_R$.  We choose the latter as follows:~$\psi_R$ is the result of projecting the   final state~$\Psi^-_{0,0}(T)$ of the Hamiltonian interpolation, starting from  the initial Hamiltonian~$\Htriv^-(0,0)=-\sum_i Z_i$ onto the ground space of the doubled semion model~$\Htop$ and normalizing, i.e.,
\begin{align}
\psi_R= \frac{P_0\Psi^-_{0,0}(T)}{\|P_0\Psi^-_{0,0}(T)\|}\ . \label{eq:referencestatedefsem}
\end{align}
We briefly remark that the state~$\psi_R$ is uniquely determined (up to a phase) as the unique simultaneous~$+1$-eigenvector of $TS^3TS$ (see Section~\ref{sec:symmetryv}) and the string operator~$\bar{Z}=Z_1Z_2$ (which 
is the string-operator $F_{(s,s)}(C)$ for the associated loop~$C$ when acting on the ground space of~$\Htop$): indeed, the latter operator commutes with both~$\Htriv^-(0,0)$ and $\Htop$. We also point out that, similarly to the toric code,  the local $Z_i$-operators correspond to a combination of pair creation, hopping and pair annihilation of $(\bm{s},\bm{s})$~anyons.

The preparation stability of the reference state $\psi_R$ with respect to the initial Hamiltonians~$\Htriv^\pm(a,b)$ with negative and positive $Z$~field component is illustrated in Fig.~\ref{fig:DSZhmispherestateoverlap}.
For negative~$Z$ field (Fig.~\ref{fig_DS_minusZhemisphere120}) the resulting state~$\Psi^-_{a,b}(T)$  has large overlap with the reference state~$\psi_R$ for almost the entire parameter range.   Even when starting from initial Hamiltonians with positive $Z$ field component (Fig.~\ref{fig_DS_plusZhemishpere_stateoverlap120}), where the final state does not have a large overlap with the topological ground space (see Fig.~\ref{fig_DS_plusZhemishpere_gsoverlap}), the ground space contribution comes almost exclusively from the reference state. Thus, for doubled semion model, we  identify a single stable final state~$\psi_R$ corresponding to the initial Hamiltonian $\Htriv=-\sum_i Z_i$.

\subsection{The doubled Fibonacci model}
As our last case study of Hamiltonian interpolation, we consider the doubled Fibonacci model described in Section~\ref{sec:doubledfib}. 

\paragraph{(Non)-adiabaticity.} 
 Fig.~\ref{fig_groundspaceoverlap} shows the adiabaticity error~$\nonadiabaticity(T)$ as a function of the total evolution time~$T$ for the initial Hamiltonians $\Htriv^\pm=\pm\sum_j Z_j$.  Note that to achieve the same error, the total evolution time~$T$ needs to be much longer compared to the toric code and the doubled semion models. It also illustrates that an error of around~$\nonadiabaticity(T)\approx 10^{-3}$ is obtained for~$T=320$: the final state~$\Psi(T)$ overlaps well with the ground space of~$\Htop$.

In Fig.~\ref{fig:fibzplusminus}, we consider the non-adiabaticity~$t\mapsto\nonadiabaticity(t)$ along the evolution, again for 
the initial Hamiltonians~$\Htriv^{\pm}=\pm\sum_j Z_j$. In particular, Fig.~\ref{fig_instantoverlap3200}, which is for a total evolution time of $T=320$,  we see that the deviation of the state~$\Psi(t)$ from the instantaneous ground state of~$H(t)$ can be much larger (compared to the non-adiabaticity~$\nonadiabaticity(T)$) along the evolution, even when approaching the end of Hamiltonian interpolation: we have $\nonadiabaticity(t)\gtrsim 10^{-2}$ for $t\approx 280$. The fact that the ground space of the final Hamiltonian~$\Htop$ is reached nevertheless at time~$t\approx T$ is essentially due to the  exact degeneracy in the final Hamiltonian~$\Htop$:
In fact, the system is in a state which has a large overlap with the subspace of `low energy' (corresponding to the $4$-fold degenerate subspace of~$\Htop$) along the trajectory, but not necessarily with the unique instantaneous ground state of~$H(t)$ for $t<T$. For $t=T$, the state has a large overlap with the ground space of~$\Htop$ since the latter is higher-dimensional.

This illustrates that the adiabaticity error~$t\mapsto \nonadiabaticity(t)$ along the evolution (i.e., for $t<T$) does not provide sufficient information to conclude that the ground space of $\Htop$ is reached at the end of the evolution.  Due to the small energy splitting within the topological ``phase'' it is more fruitful to view the part of the interpolation close to $t\approx T$ in terms of degenerate adiabatic perturbation theory~\cite{Rigolin2014} instead of the traditional adiabatic theorem.

Fig.~\ref{fig_instantoverlap3200} also shows that for $T=320$, changing the Trotter time steps~$\Delta t$ (cf.~\eqref{eq:trottertimes}) from $\Delta t=0.1$ to $\Delta t=0.01$ does not significantly change the behavior, particularly for the initial Hamiltonian~$\sum_i Z_i$.
On the other hand, by increasing the Hamiltonian interpolation time for $\Htriv=\sum_i Z_i$ to $T=1280$, as in Fig.~\ref{fig_instantoverlap12800}, we see the evolution closely follows the instantaneous ground state.
The discrepancy can be seen as a ``lag'' or delay of the evolved state and the instantaneous ground state and is largest at the ``phase transition'', $H(t)\approx 1/4 \Htop+3/4\sum_i Z_i$, where the gap closes.
\begin{figure}
\centering
\includegraphics[scale=0.6]{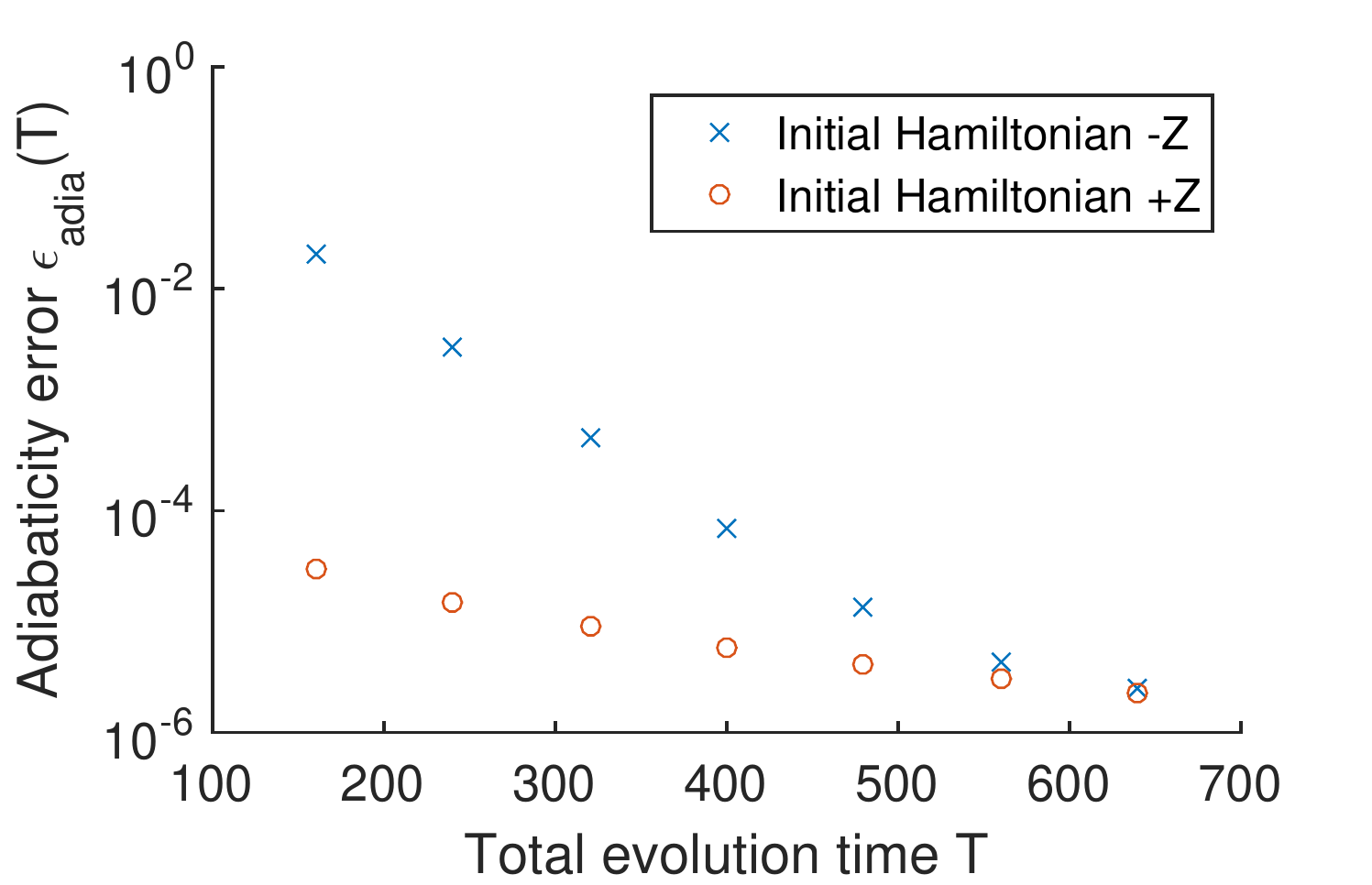}
\caption{The adiabaticity error~$\nonadiabaticity(T)$  with respect to different total evolution times~$T$ for the Fibonacci model. 
The initial Hamiltonian~$\Htriv$ is either $\Htriv^+(0,0)=\sum_i Z_i$ or $\Htriv^-(0,0)=-\sum_i Z_i$.
Note that for this choice of initial Hamiltonians,  the vertex terms $A_v$ are conserved quantities (as for example in the toric code).
Since both $\ket{1}^{\otimes 3}$ and $\ket{\tau}^{\otimes 3}$ are in the ground space of $A_v$, both signs of the pure~$Z$ field lead to a Hamiltonian interpolation which invariantly remains in the ground space of~$A_v$. In other words, the adiabaticity error stems from the plaquette terms. 
Other fields are computationally more costly, since they lift the block decomposition of the interpolating Hamiltonians~$H(t)$ induced by the conserved vertex terms, reducing the sparsity of the unitary evolution.
\label{fig_groundspaceoverlap}}
\end{figure}

\begin{figure}
\begin{subfigure}[t]{0.5\textwidth}
	\centering\captionsetup{width=.8\linewidth}
\includegraphics[width=\textwidth]{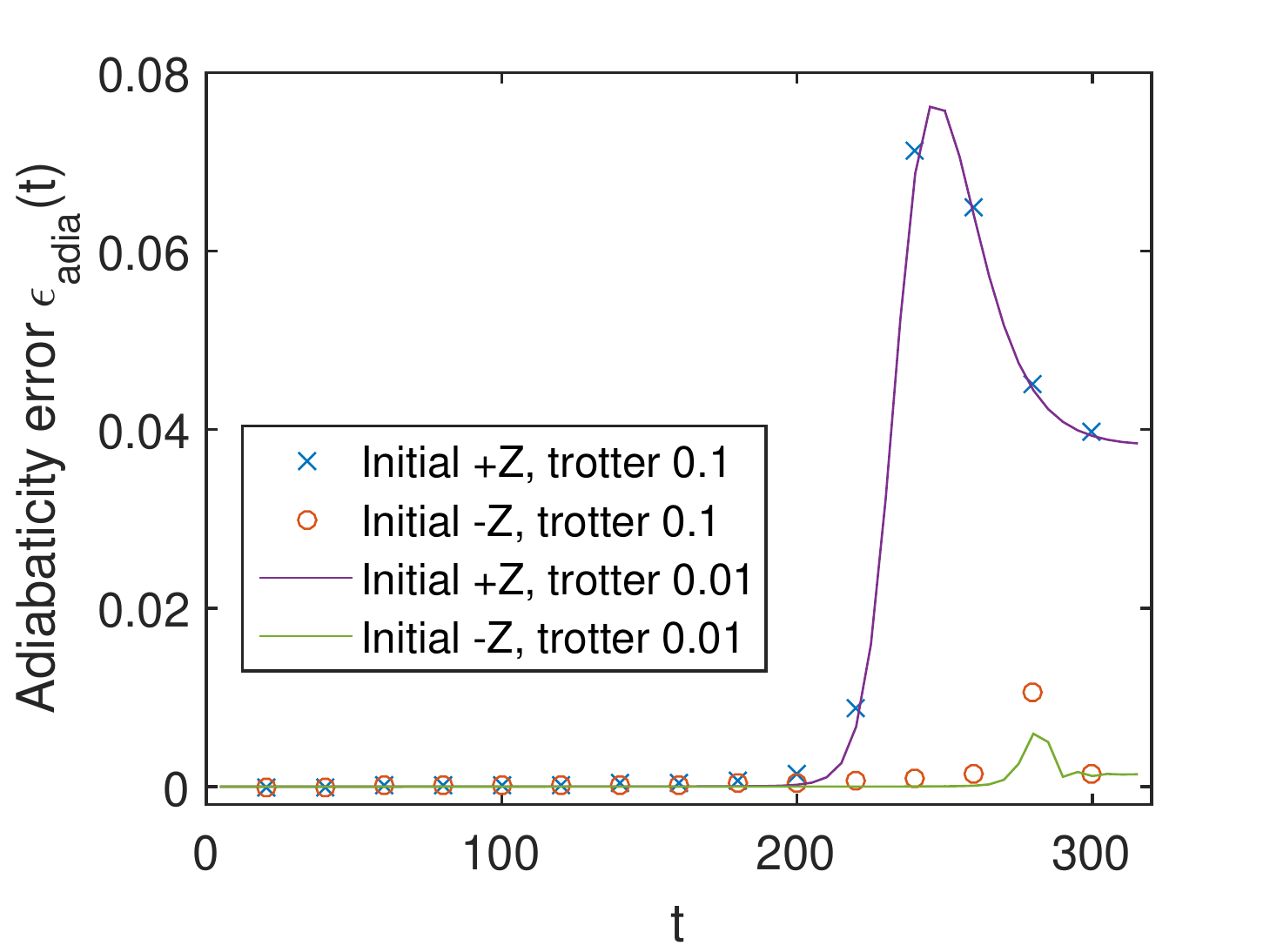}
\caption{The total evolution time is $T=320$. 
As explained in the text, the fact that the overlap with the instantaneous ground state towards the end of the evolution is small does not prevent the system from reaching the degenerate ground space of the final Hamiltonian~$\Htop$ (see~\protect\ref{fig_groundspaceoverlap}). 
Changing the Trotter discretization step from $\Delta t=0.1$ to $\Delta t=0.01$ does not significantly change the behavior.
}\label{fig_instantoverlap3200}
\end{subfigure}\qquad 
\begin{subfigure}[t]{0.5\textwidth}
	\centering\captionsetup{width=.8\linewidth}
\includegraphics[width=\textwidth]{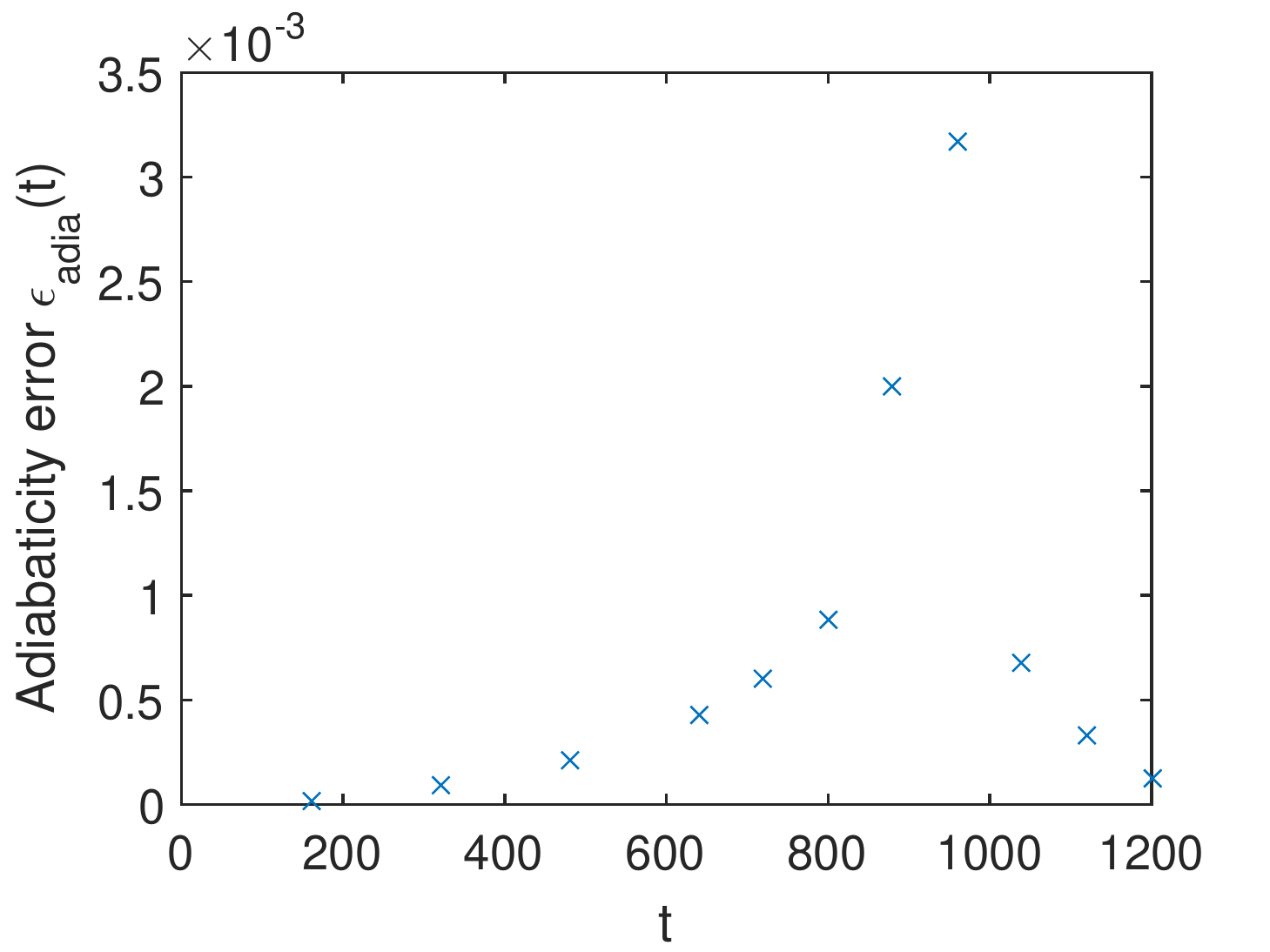}
\caption{
For the initial Hamiltonian $\Htriv=\sum_i Z_i$ and a total evolution  time $T=1280$, the system closely follows the instantaneous ground state of $H(t)$.}\label{fig_instantoverlap12800}
\end{subfigure}
\caption{
The overlap of the state~$\Psi(t)$ at time~$t$ with the instantaneous ground space of $H(t)$, as expressed by the adiabaticity error $t\mapsto \nonadiabaticity(t)$ along the evolution. The initial Hamiltonian is either $\Htriv^+(0,0)=\sum_i Z_i$ or $\Htriv^-(0,0)=-\sum_i Z_i$, and the final Hamiltonian~$\Htop$ is the doubled Fibonacci model.\label{fig:fibzplusminus}}
\end{figure}

\paragraph{Logical state. }
Fig.~\ref{fig_Fib_plusZhemishpere_state} provides information about the final state~$\Psi^{\pm}_{a,b}(T)$ of Hamiltonian interpolation, for the family of initial Hamiltonians~$\Htriv^{\pm}(a,b)$ (cf.~\eqref{eq:Hpmdef}). Again, the figure
gives the overlap with a  single reference state~$\psi_R$. Similarly as before, we choose the latter as the  final state of Hamiltonian interpolation, starting with initial Hamiltonian $\Htriv^-(0,0)=-\sum_i Z_i$, and subsequently projected into the ground space and normalized (cf.~\eqref{eq:referencestatedefsem}).

\begin{figure}
\begin{subfigure}[t]{0.5\textwidth}
\centering\captionsetup{width=.8\linewidth}
\includegraphics[width=\textwidth]{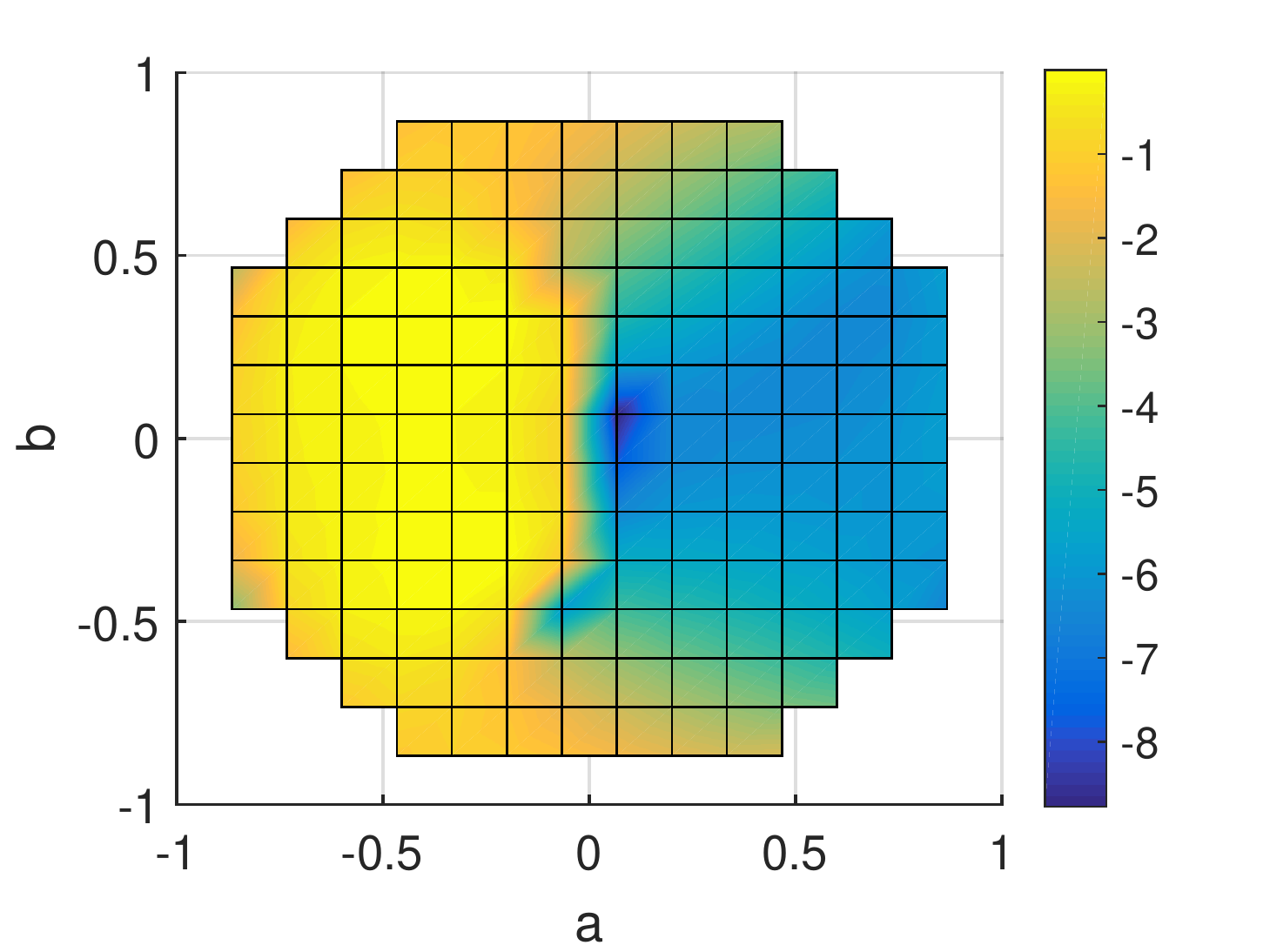}
\caption{The quantity $\ln (1-|\inp{\Psi^+_{a,b}(T)}{\psi_R}|^2)$
for initial Hamiltonian of the form~$\Htriv^+(a,b)$ around $\Htriv^+(0,0)=\sum_j Z_j$. For the whole range of parameters~$(a,b)$, the adiabaticity error
is small, $\nonadiabaticity(T)\leq 10^{-4}$.
The reference state~$\psi_R$ corresponds to the center of Fig.~\protect\ref{fig_minusZhemishpere_overlap} (up to projection onto the ground space of~$\Htop$ and normalization). The figure illlustrates that 
the final state~$\Psi^+_{a,b}(T)$ has non-trivial overlap with the reference state in the region~$a>0$, but is very sensitive to the choice of parameters~$(a,b)$, especially around $(a,b)=(0,0)$.  
\label{fig_plusZhemishpere_overlap}}
\end{subfigure}
\begin{subfigure}[t]{0.5\textwidth}
\centering\captionsetup{width=.8\linewidth}
\includegraphics[width=\textwidth]{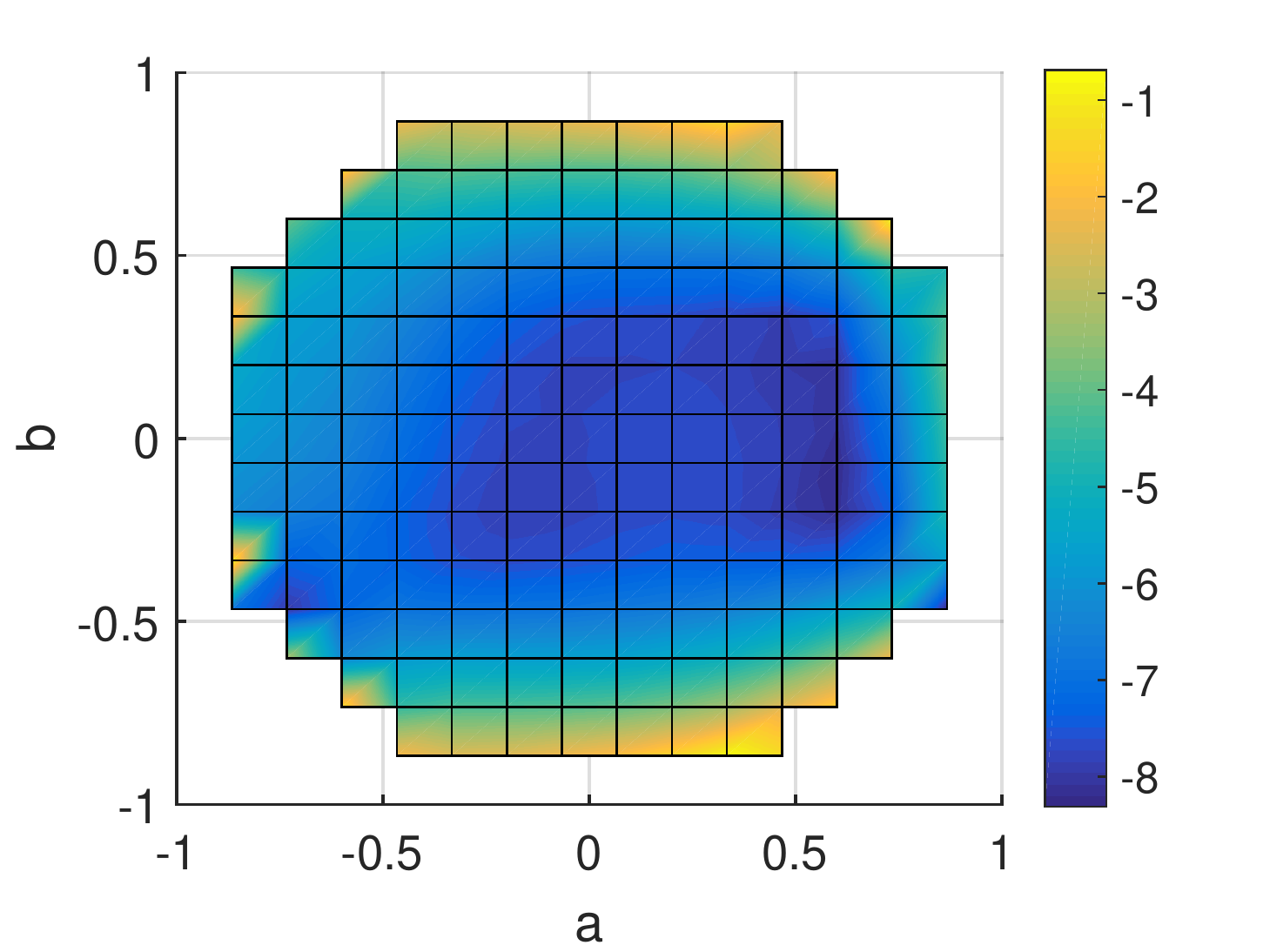}
\caption{
The quantity $\ln (1-|\inp{\Psi^-_{a,b}(T)}{\psi_R}|^2)$,
for initial Hamiltonians $\Htriv^-(a,b)$ around $H^-(0,0)=-\sum_j Z_j$. 
For the whole range of parameters~$(a,b)$, the adiabaticity error
is small, $\nonadiabaticity(T)\leq 0.005$ apart from the point on the boundary of the plot. 
 The Hamiltonian interpolation reaches the reference state~$\psi_R$ essentially for the full parameter range. 
}\label{fig_minusZhemishpere_overlap}
\end{subfigure}
\caption{These figures show the overlap between the final states~$\Psi^\pm_{a,b}$ of Hamiltonian interpolation and the reference state~$\psi_R$.  
This is for the family $\Htriv^\pm(a,b)$ of initial Hamiltonians and the double Fibonacci model~$\Htop$ as the final Hamiltonian. The reference state~$\psi_R$ is chosen in both figures as in~\protect\eqref{eq:referencestatedefsem} (corresponding to the center point in Fig.~\protect\ref{fig_minusZhemishpere_overlap}).  
The total  evolution time is $T=320$ in both cases. } \label{fig_Fib_plusZhemishpere_state}
\end{figure}

We observe significant overlap of the final state 
with the reference state~$\psi_R$ for the whole parameter range for the initial Hamiltonians~$\Htriv^-(a,b)$ (Fig.~\ref{fig_minusZhemishpere_overlap}). In contrast, for the initial Hamiltonians~$\Htriv^+(a,b)$,  the final state depends strongly on the choice of parameters~$(a,b)$ (Fig.~\ref{fig_plusZhemishpere_overlap}). 

To relate this to the discussion in Section~\ref{sec:twodimensionalsystems} (respectively Conjecture~\ref{claim:targetstates}), let us first consider the 
centerpoint of Fig.~\ref{fig_minusZhemishpere_overlap}
associated with the initial Hamiltonian~$\Htriv^-(0,0)=-\sum_{j}Z_j$.
These terms correspond to a combination of local pair creation, hopping and pair annihilation of $(\bm{\tau},\bm{\tau})$ anyons, as explained in Appendix~\ref{sec:stringFibonacci}.  The effective Hamiltonian 
can be computed at this point  based on expression~\eqref{eq:singleparticledominant} and the $S$- and $T$-matrices  given in~Eq.~\eqref{eq:doubledfibst}.
The result is given numerically in Eq.~\eqref{eq_analytical_perturbed_ground_state}
in the appendix. Computing the ground state~$\psi_{\mathsf{eff}}$ of this effective Hamiltonian, we observe that
with respect to the projections $\{ P_{1,1}, P_{\tau,\tau} ,P_{1,\tau},P_{\tau,1} \}$, the expectation values
of the reference state~$\psi_R$ and $\psi_{\mathsf{eff}}$ are similar,
\begin{center}
\begin{tabular}{l|llll}
& $ P_{1,1}$ & $P_{\tau,\tau}$ & $P_{1,\tau}$ & $P_{\tau,1}$\\
\hline\
$\psi_R$ & 0.5096 & 0.4838 & 0.0033 & 0.0033\\
$\psi_{\mathsf{eff}}$ & 0.5125 & 0.4804 & 0.0036 & 0.0036
\end{tabular}
\end{center}

Moving away from the center point in Fig.~\ref{fig_minusZhemishpere_overlap},  we compute the overlaps of the reference state~$\psi_R$ with the ground states~$\psi^{\pm}_{\mathsf{pert}}(a,b)$ of  perturbed Hamiltonians of the  form~$H^{\pm}_{\mathsf{pert}}(a,b)=\Htop\pm 0.001 \Htriv^{\pm}(a,b)$, as illustrated in Fig.~\ref{fig:perturbationcomp} (the latter providing an approximate notion of effective Hamiltonians). The figure illustrates that these perturbed states have, as expected, a certain degree of stability with respect to the parameters~$(a,b)$. Comparison with Fig.~\ref{fig_Fib_plusZhemishpere_state}  thus points to a certain discrepancy between the behavior of perturbed states and states obtained by Hamiltonian interpolation:
Fig.~\ref{fig_plusZhemishpere_overlap} shows high sensitivity 
of the final state to initial parameters~$(a,b)$
 (which is absent in the perturbative prediction), whereas Fig.~\ref{fig_minusZhemishpere_overlap} shows that the final state is close to the reference state~$\psi_R$ throughout (as opposed to the perturbative prediction, where this is not the case along the boundary). To rule out that this discrepancy stems from an insufficiently large choice of  the total evolution time~$T$, we also show that different choices of the total evolution time~$T$ do not significantly affect the overlap with the reference state along the line~$b=0$, see Fig.~\ref{fig:overlap1Dline}. 

In summary, we conclude that while for a large parameter range of initial parameters the reference state~$\psi_R$ is indeed reached, the stability property is less pronounced than for the toric code and the doubled semion models. In addition, a na\"ive comparison with ground states of perturbed Hamiltonians suggests that the description via effective Hamiltonians does not capture all relevant features. We conjecture that higher orders in perturbation theory are needed to provide more information in the case of the Fibonacci model: the state may be ``locked'' in eigenstates of such higher-order Hamiltonians before the lowest order effective Hamiltonian dominates.

\begin{figure}
\begin{subfigure}[t]{0.5\textwidth}
\centering\captionsetup{width=.8\linewidth}
\includegraphics[width=\textwidth]{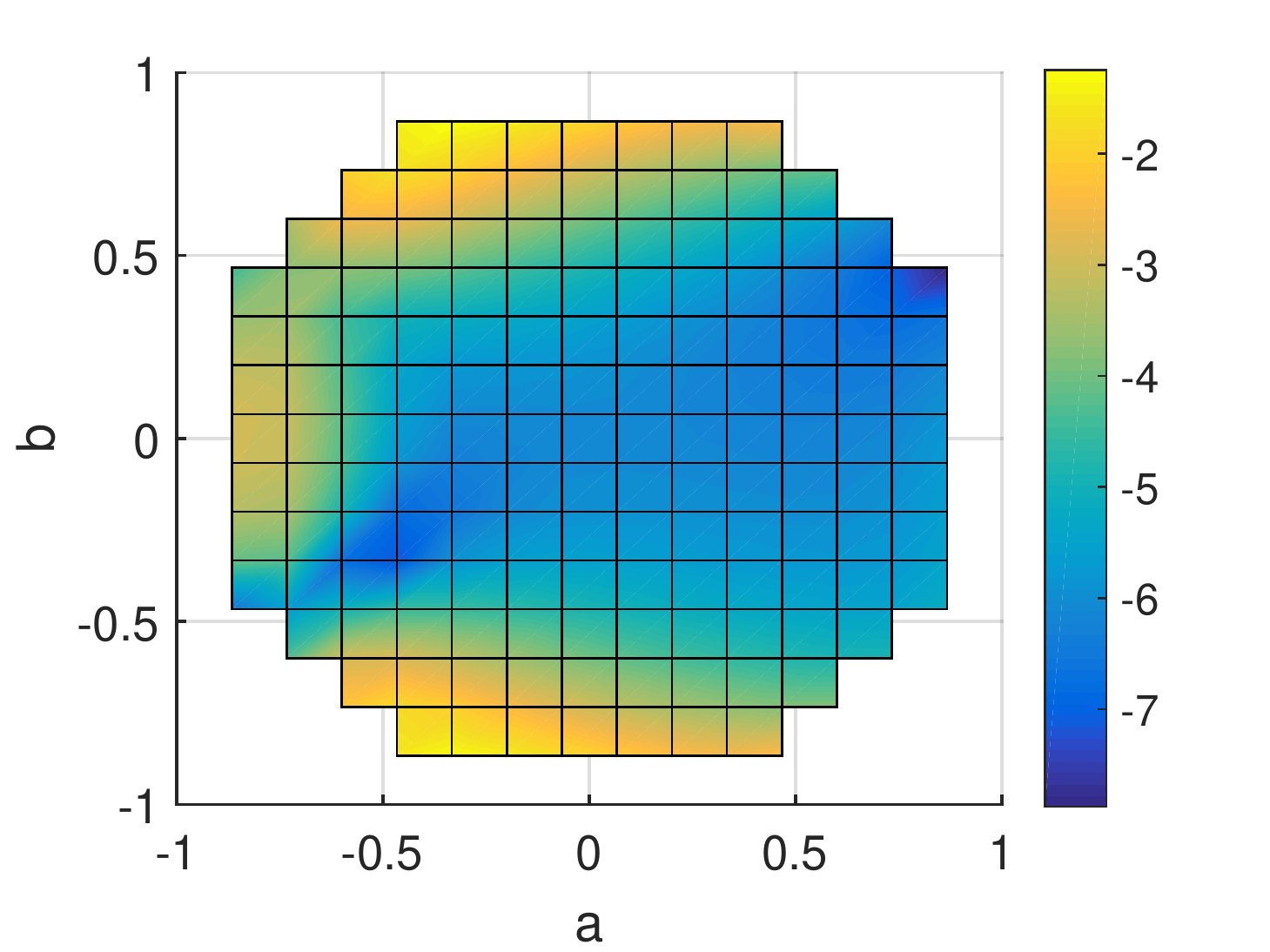}
\caption{The quantity  $\ln (1-|\inp{\psi^+_{\mathsf{pert}}(a,b)}{\psi_R}|^2)$
for  Hamiltonians~$H_{\mathsf{pert}}^+(a,b)$. 
\label{fig_plusZhemishpere_perturbation}}
\end{subfigure}
\begin{subfigure}[t]{0.5\textwidth}
\centering\captionsetup{width=.8\linewidth}
\includegraphics[width=\textwidth]{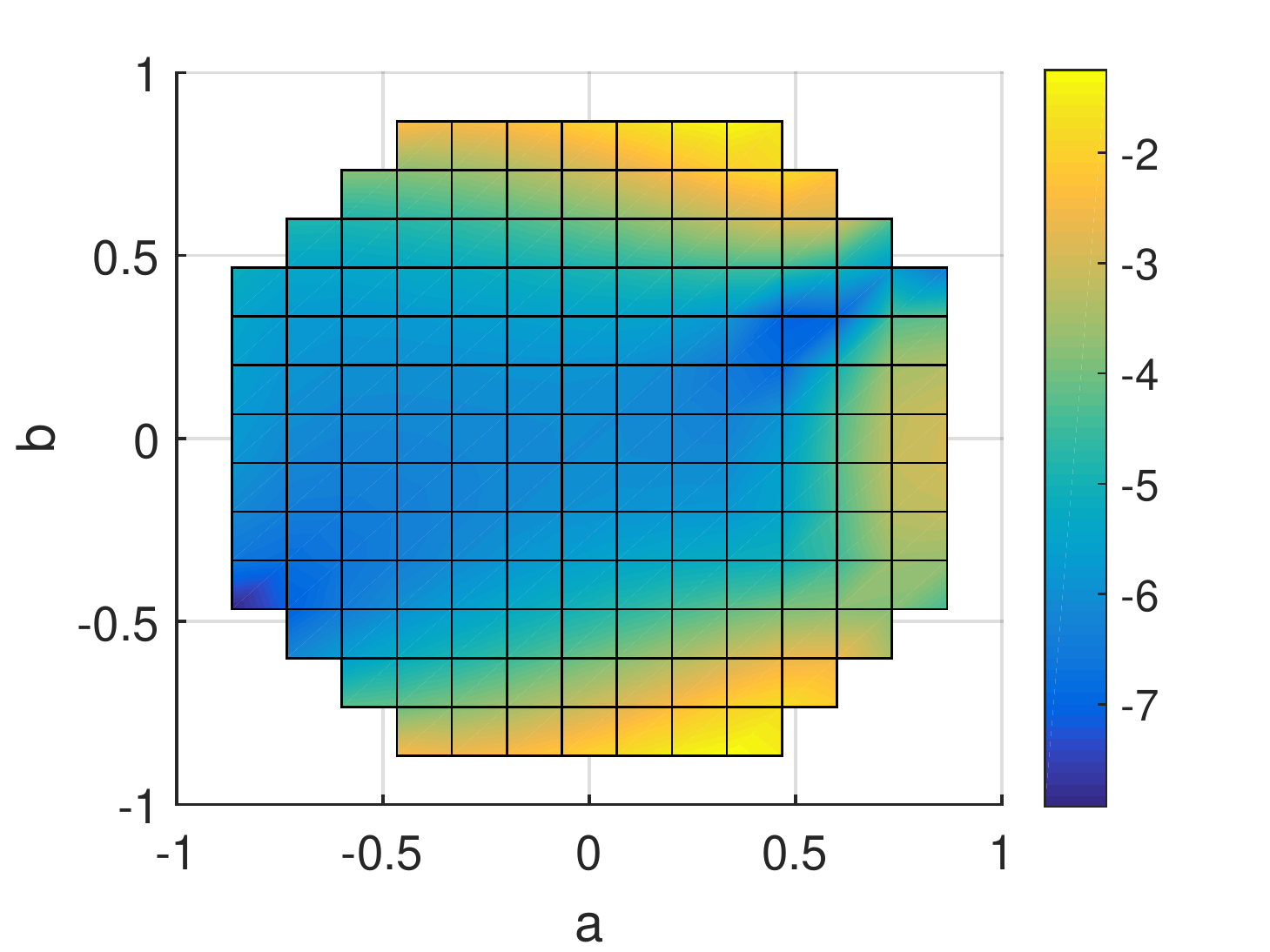}
\caption{
The quantity $\ln (1-|\inp{\psi^-_{\mathsf{pert}}(a,b)}{\psi_R}|^2)$ 
for Hamiltonians~$H_{\mathsf{pert}}^-(a,b)$. The reference state $\psi_R$ has 
overlap $|\langle \psi^-_{\mathsf{pert}}(0,0)|\psi_R\rangle|^2\approx 0.9976$ 
with the  ground state of the perturbed Hamiltonian $H_{\mathsf{pert}}^-(0,0)=\Htop-0.001\sum_j Z_j$.
}\label{fig_minusZhemishpere_perturbation}
\end{subfigure}
\caption{To compare with the perturbative prediction, 
these figures give the overlap between the reference state~$\psi_R$ and 
the ground state
$\psi_{\mathsf{pert}}(a,b)$ of the perturbed Hamiltonians $H^{\pm}_{\mathsf{pert}}(a,b)=\Htop\pm 0.001\Htriv^+(a,b)$.  
\label{fig:perturbationcomp}}
\end{figure}

\begin{figure}
\centering
\includegraphics[scale=0.8]{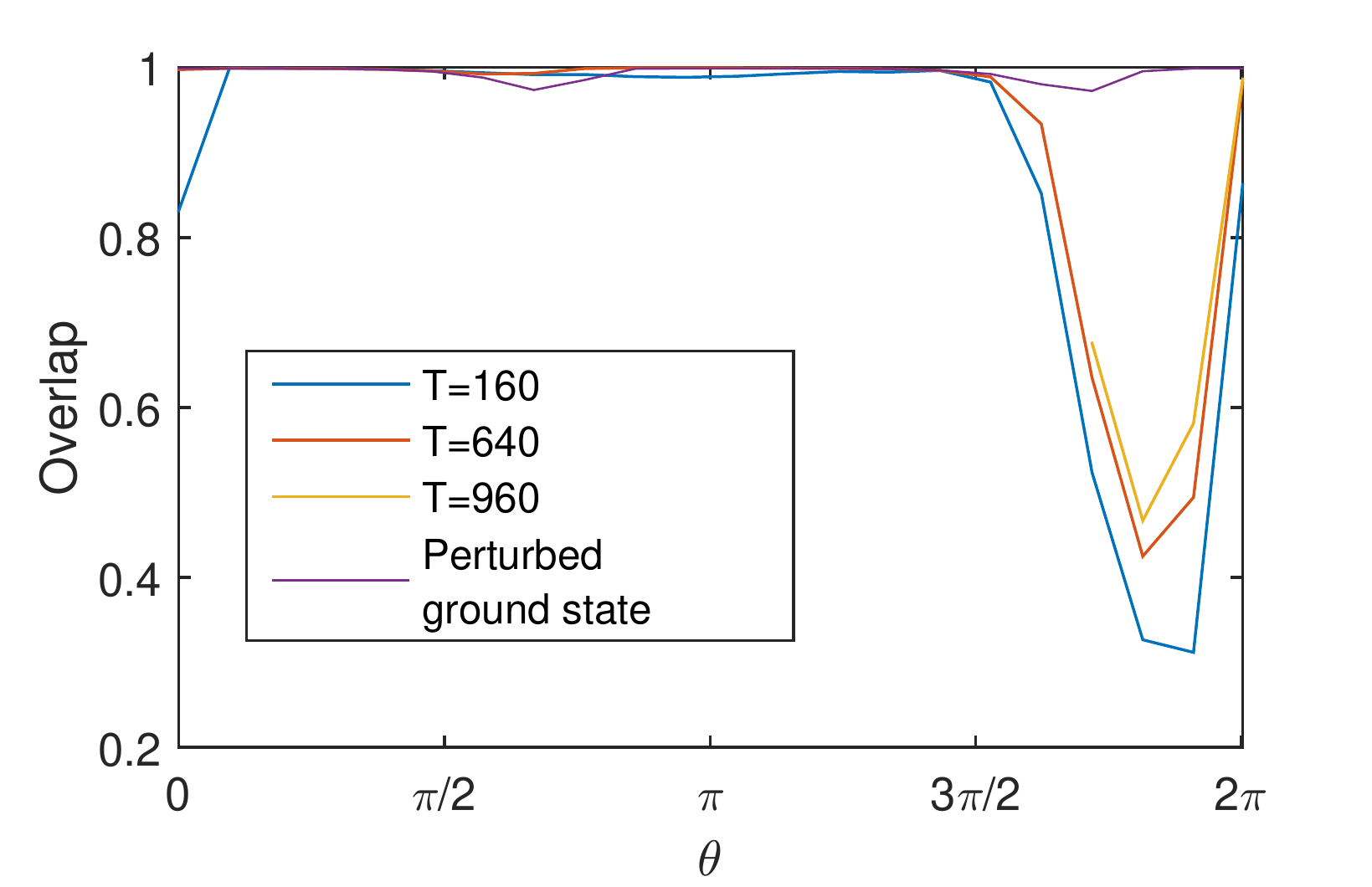}
\caption{
This figure shows the overlap~$|\langle\Psi_{\mathsf{pert}}(\theta)|\psi_R\rangle|^2$
for initial Hamiltonians $\Htriv(\theta)=\sum_{j} \cos \theta Z_j+\sin \theta X_j$ along the line of the horizontal axis in Fig.~\protect\ref{fig_plusZhemishpere_perturbation}, for different values of the total evolution time~$T$.
We only compute $T=960$ on the rightmost region to show that increasing the total evolution time do not significantly change the final states.
It also gives the corresponding overlap~$|\langle\Psi_{\mathsf{pert}}(\theta)|\psi_R\rangle|^2$ between the ground state of $H_{\mathsf{pert}}(\theta)$ and the reference state. The figure illustrates that increasing the evolution time~$T$
does not significantly alter the overlap with the reference state.}\label{fig:overlap1Dline}
\end{figure}

\subsection*{Acknowledgments}
FP acknowledges funding provided by the Institute for Quantum Information and Matter, a NSF Physics Frontiers Center with support of the Gordon and Betty Moore Foundation (Grants No. PHY-0803371 and PHY-1125565). 
FP would alos like to acknowledge insightful discussions with John Preskill, Brian Swingle, Julien Vidal and Chris Laumann.
BY is supported by the David and Ellen Lee Postdoctoral fellowship and the Government of Canada through Industry Canada and by the Province of Ontario through the Ministry of Research and Innovation.
 RK is supported by the Technische Universit\"at
at M\"unchen --  Institute  for  Advanced  Study,  funded  by  the  German  Excellence  Initiative  and  the
European  Union  Seventh  Framework  Programme  under  grant  agreement  no.~291763.
\appendix

\section{Equivalence of the self-energy- and Schrieffer-Wolff methods for topological order\label{sec:equivalenceselfenergyschrieff}}
 As discussed in Section~\ref{sec:perturbativeeffectivetop},
here we show that at lowest non-trivial order,
the expressions obtained from the self-energy-method and the Schrieffer-Wolff method coincide if the Hamiltonian and perturbation satisfies a certain topological order condition.

We begin with a review of the exact Schrieffer-Wolff transformation (Section~\ref{app:exactschriefferwolfdef}), as well as the expressions resulting from the Schrieffer-Wolff perturbative expansion (Section~\ref{sec:seriesexpansion}).  In Section~\ref{sec:prepSWdef}, we present some preliminary computations. In Section~\ref{sec:toporderconstraint}, we introduce the topological order constraint and establish our main result.

\subsection{Exact-Schrieffer-Wolff transformation\label{app:exactschriefferwolfdef}}
As mentioned in Section~\ref{sec:effectivehamiltonians}, the Schrieffer-Wolff method provides a unitary $U$ such that 
\begin{align}
\Heff=U(H_0+\epsilon V)U^\dagger\label{eq:heffdefinitionsw}
\end{align} preserves the ground space $P_0\cH$ of $H_0$, and can be considered as an effective Hamiltonian. The definition of the unitary is as follows:  let~$P$ be the projection onto the ground space of the perturbed Hamiltonian~$H_0+\varepsilon V$. Defining the reflections
\begin{align*}
R_{P_0}&=2P_0-I\\
R_P&=2P-I
\end{align*}
the (exact) Schrieffer-Wolff transformation is defined by the ``direct rotation''
\begin{align}
U&=\sqrt{R_{P_0}R_{P}}\ ,\label{eq:nonvariationalcharacterization}
 \end{align} 
 where the square root is defined with a branch cut along the negative real axis. The effective Hamiltonian is then given by
 \begin{align}
 \Heff(\epsilon)&=P_0 U(H_0+\varepsilon V) U^\dagger P_0\ .\label{eq:schriefferwolfexacth}
 \end{align} 
 A variational characterization (see~\cite{bravyietal})  of the unitary~$U$ (instead of~\eqref{eq:nonvariationalcharacterization}) is often more useful (e.g., for computing the effective Hamiltonian in the case of a two-dimensional ground space, such as for the Majorana chain): we have 
\begin{align}
U&=\arg \min \left\{\|I-U\|_2\ \big|\ U\textrm{ unitary and }UPU^\dagger=P_0\right\}\ ,\label{eq:expressionvariationalSW}
\end{align}
where $\|A\|_2=\sqrt{\tr(A^\dagger A)}$ is the Frobenius norm.

\subsection{The perturbative SW expansion \label{sec:seriesexpansion}}
Since the transforming unitary~\eqref{eq:nonvariationalcharacterization}, as well as expression~\eqref{eq:heffdefinitionsw}, are difficult to compute in general, a standard approach is to derive systematic series in the parameter~$\epsilon$ (the perturbation strength). In this section, we summarize the expressions for this explicit perturbative expansion of the Schrieffer-Wolff effective Hamiltonian obtained in~\cite{bravyietal}. 
The perturbation is split into diagonal and off-diagonal parts according to
\begin{align}
\Vd &=P_0 VP_0+Q_0VQ_0=:\cD(V)\label{eq:vddefinition}\\
\Vod&=P_0V Q_0+Q_0VP_0=:\cO(V)\ .\label{eq:voddefinition}
\end{align}
where $P_0$ is the projection onto the ground space of $H_0$,
and $Q_0=I-P_0$  the projection onto the orthogonal complement.
Assuming that $\{\ket{i}\}_{i}$ is the eigenbasis of $H_0$ with energies $H\ket{i}=E_i\ket{i}$, one introduces the superoperator
\begin{align*}
\cL(X)&=\sum_{i,j}\frac{\bra{i}\cO(X)\ket{j}}{E_i-E_j}\ket{i}\bra{j}\ .
\end{align*}
Then the operators $S_j$ are defined recursively as
\begin{align}
S_1&=\cL(\Vod)\nonumber\\
S_2&=-\cL(\ad{\Vd}(S_1))\nonumber\\
S_n&=-\cL(\ad{\Vd}(S_{n-1}))+\sum_{j\geq 1}a_{2j}\cL( \hat{S}^{2j}(\Vod)_{n-1})\ ,\label{eq:sndefinitiongeneral}
\end{align}  
where
\begin{align}
\hat{S}^k(\Vod)_m &=\sum_{\substack{n_1,\ldots,n_k\geq 1\\
\sum_{r=1}^k n_r=m}} \ad{S_{n_1}}\cdots\ad{S_{n_k}}(\Vod)\ ,\label{eq:skvodmdef}
\end{align}
and where $\ad{S}(X)=[S,X]$. The constants are  $a_{m}=\frac{2^m \beta_m}{m!}$, where $\beta_m$ is the $m$-th Bernoulli number.
Observe that
\begin{align*}
\hat{S}^k(\Vod)_m&=0\qquad\textrm{ for }k>m\ .
\end{align*}
The 
$q$-th order term in the expansion~\eqref{eq:heffexpansion} is
\begin{align}
{\Heff}_{,q}=\sum_{1\leq j\leq \lfloor q/2\rfloor}b_{2j-1} P_0\hat{S}^{2j-1}(\Vod)_{q-1}P_0\ ,\label{eq:qthordereffectiveHamiltoniandef}
\end{align}
where $b_{2n-1}=\frac{2(2^n-1) \beta_{2n}}{(2n)!}$.

Since our main goal is to apply the perturbation theory to topologically ordered (spin) systems, we can try to utilize their properties. In particular, one defining property of such systems is that, if an operator is supported on a topological trivial region, then it acts trivially inside the ground space. A common non-trivial operation in the ground space corresponds to the virtual process of tunneling an anyon around the torus. This property will allow us to simplify the computation when we want to compute the lowest order effective Hamiltonian. In the following subsections, we will show that although $S_n$ is defined recursively based on $S_1,\ldots, S_{n-1}$, only the first term $-\cL(\ad{\Vd}(S_{n-1}))$ on the rhs of~\eqref{eq:sndefinitiongeneral} would contribute to the lowest order effective Hamiltonian. The intuition behind this claim is that the other term $\sum_{j\geq 1}a_{2j}\cL( \hat{S}^{2j}(\Vod)_{n-1})$ corresponds to virtual processes which go through the ground space $\rightarrow$ excited space $\rightarrow$ ground space cycle multiple times (larger than one). It is intuitive that such virtual processes would not happen when we want to consider the lowest order perturbation.

\subsection{Some preparatory definitions and properties\label{sec:prepSWdef}}
Let 
\begin{align}
\label{eq_G_resolvent}
G(z)&=(zI-H_0)^{-1}
\end{align}
be the resolvent of the unperturbed Hamiltonian $H_0$. Let $E_0$ be the ground space energy of~$H_0$. 
We set
\begin{align*}
G&=G(E_0)=G(E_0)Q_0=Q_0G(E_0)Q_0\ ,
\end{align*}
i.e., the inverse is taken on the image of $Q_0$. Then $\cL$ can be written as
\begin{align}
\cL(X)&=P_0XG -GXP_0\ .\label{eq:lpxg}
\end{align}

To organize the terms appearing in the perturbative Schrieffer-Wolff expansion, it will be convenient to introduce the following subspaces of operators.
\begin{definition}\label{def:gammasetsops}
For each $n\in\mathbb{N}$, let $\Gamma(n)$ be
the linear span of operators of the form
\begin{align}
Z_0VZ_1VZ_2\cdots Z_{n-1}VZ_n\ ,\label{eq:zsequencedef}
\end{align}
where  for each  $j=0,\ldots,n$,  the operator $Z_j$ is either one of the projections $P_0$ or $Q_0$, or a positive power of~$G$, i.e.,  $Z_j\in\{P_0,Q_0\}\cup \{G^m\ | m\in\mathbb{N}\}$.

Let $\Gamma^\star(n)\subset\Gamma(n)$
the span of operators of the form~\eqref{eq:zsequencedef} which additionally satisfy
the condition
\begin{align}
Z_0Z_n=Z_nZ_0=0\ ,\label{it:orthogonalityrelationv}
\end{align}
i.e., $Z_0$ and $Z_n$ are orthogonal. 
\end{definition}
For later reference, we remark that operators in $\Gamma^\star(n)$ a linear combinations of certain terms which are off-diagonal with respect the ground space of~$H_0$ (and its orthogonal complement).  In particular, any product of an even number of these operators is diagonal.

The first observation is that the summands 
the effective Hamiltonian~\eqref{eq:qthordereffectiveHamiltoniandef}
have this particular form.

\begin{lemma}\label{lem:basicformsnops}
We have
\begin{align}
\Vod\in \Gamma^\star(1)\label{eq:vodstar}
\end{align}
and
\begin{align}
S_n\in \Gamma^\star(n)\textrm{ for every }n\in\mathbb{N}\ \label{eq:sncontainementclaim}
\end{align}
Furthermore,
\begin{align}
\hat{S}^k(\Vod)_m\in \Gamma(m+1)\label{eq:shatvod} 
\end{align}
for all $k,m$.
\end{lemma}
\begin{proof}
The definition of $\Gamma(n)$ immediately implies  that 
\begin{align}
XY\in \Gamma(n_1+n_2)\qquad\textrm{ for }X\in \Gamma(n_1)\textrm{ and } Y\in \Gamma(n_2)\  \label{eq:compositionlawgamma}.
\end{align}
Thus
\begin{align}
\ad{X}(Y)\in \Gamma(n_1+n_2)\qquad\textrm{ for }X\in \Gamma(n_1)\textrm{ and } Y\in \Gamma(n_2)\  .
\end{align}
Furthermore, 
inspecting the Definitions~\eqref{eq:voddefinition} and~\eqref{eq:vddefinition}, we immediately verify that 
\begin{align}
\Vod\in \Gamma(1)\qquad\textrm{ and }\qquad \Vd\in\Gamma(1)\ \label{eq:vodclaimgamma}
\end{align}
Similarly,~\eqref{eq:vodstar} follows directly from the definitions.

We first argue that 
\begin{align}
S_1\in \Gamma(1)\qquad\textrm{ and }\qquad S_2\in \Gamma(2)\ .\label{eq:sclaimgamma}
\end{align}
Inserting the definition of $\Vod$ and $\cL$ (that is,~\eqref{eq:lpxg}), we have 
\begin{align}
S_1&=\cL(P_0VQ_0+Q_0VP_0)\nonumber\\
&=P_0(P_0VQ_0+Q_0VP_0)G-G(P_0VQ_0+Q_0VP_0)P_0\nonumber\\
&=P_0VG-GVP_0\ ,\label{eq:soneexplicitcomputed}
\end{align}
where we used the fact $GQ_0=Q_0G=G$ and 
that $Q_0$, $P_0$ are orthogonal projections. This proves the claim~\eqref{eq:sncontainementclaim} for $n=1$ and, in particular, shows that $S_1\in \Gamma(1)$.

Similarly, for $n=2$, using the definition of $\Vd$, a straightforward calculation (using~\eqref{eq:soneexplicitcomputed}) gives
\begin{align}
\ad{\Vd}(S_1)&=(P_0VP_0VG-Q_0VGVP_0) + h.c.
\end{align}
(where $h.c.$ denotes the Hermitian conjugate of the previous expression) and thus with~\eqref{eq:lpxg}
\begin{align}
S_2&=(P_0VP_0VG^2+G VGVP_0)-h.c.\ .
\end{align}
We conclude that~\eqref{eq:sncontainementclaim} holds $n=2$ and, in particular, $S_2\in\Gamma(2)$, as claimed (Eq.~\eqref{eq:sclaimgamma}).

With~\eqref{eq:vodclaimgamma}  and~\eqref{eq:sclaimgamma}, we can use the composition law~\eqref{eq:compositionlawgamma}
 to show inductively that
\begin{align}
S_n\in \Gamma(n)\qquad\textrm{ for all }n\in\mathbb{N}\ .\label{eq:sngammaclaim}
\end{align}
Indeed,~\eqref{eq:sngammaclaim} holds for $n=1,2$. Furthermore, assuming $S_m\in \Gamma(m)$ for all $m\leq n-1$, we
can apply~\eqref{eq:compositionlawgamma} and~\eqref{eq:vodclaimgamma} to the Definition~\eqref{eq:skvodmdef} of~$\hat{S}^{2j}(\Vod)_{n-1}$, obtaining
\begin{align*}
\hat{S}^{2j}(\Vod)_{n-1}\in \Gamma(n)\qquad\textrm{ and }\qquad \ad{V_d}(S_{n-1})\in \Gamma(n)\ . 
\end{align*}
Thus~\eqref{eq:sngammaclaim} follows by definition~\eqref{eq:sndefinitiongeneral} of $S_n$,
the easily verified fact  (cf.~\eqref{eq:lpxg}) that $\cL(\Gamma(n))\subset\Gamma(n)$, and linearity. 

Finally, observe that~\eqref{eq:lpxg}
also implies 
\begin{align}
\cL(\Gamma(n))\subset \Gamma^\star(n)\ ,\label{eq:clgammainclu}
\end{align}
hence~\eqref{eq:sncontainementclaim} follows with~\eqref{eq:sngammaclaim}.

The claim~\eqref{eq:shatvod} 
is then immediate from the composition law~\eqref{eq:compositionlawgamma},
as well as~\eqref{eq:sngammaclaim} and~\eqref{eq:vodclaimgamma}.
\end{proof}

\subsection{Topological-order constraint\label{sec:toporderconstraint}}
In the following, we will assume that
\begin{align}
P_0\Gamma(n)P_0\subset\mathbb{C}P_0\qquad\textrm{  for all  }n<L\ .\eqref{eq:pzerotsandwich}
\end{align}
which amounts to saying that $(H_0,V)$ satisfies the topological order condition with parameter~$L$ (see Definition~\ref{def:topologicalorderconditionL}).  In Section~\ref{sec:trivialeffectivelow}, we argue that this implies that the effective Hamiltonian is trivial (i.e., proportional to $P_0$) for all orders $n<L$. In Section~\ref{sec:nontrivialcontr}, we then compute the non-trivial contribution of lowest order. 

\subsubsection{Triviality of effective Hamiltonian at orders $n<L$\label{sec:trivialeffectivelow}}
A simple consequence of  Definition~\ref{def:gammasetsops} then is the following.
\begin{lemma}\label{lem:producttrivialpert}
Suppose that
$P_0\Gamma(n)P_0\subset\mathbb{C}P_0$ for all $n<L$. Then for any $2k$-tuple of  integers $n_1,\ldots,n_{2k}\in\mathbb{N}$ with
\begin{align*}
\sum_{j=1}^{2k} n_j &< L\ ,
\end{align*}
and all $T_{n_j}\in \Gamma^\star(n_j)$, $j=1,\ldots,2k$, 
we have 
\begin{align*}
T_{n_1}\cdots T_{n_{2k}}P_0&\in \mathbb{C}P_0\\
P_0T_{n_1}\cdots T_{n_{2k}}&\in\mathbb{C}P_0\ .
\end{align*}
\end{lemma}
\begin{proof}
It is easy to check that because of property~\eqref{eq:compositionlawgamma},
the expression 
$T_{n_1}\cdots T_{n_{2k}}P_0$
is contained in $P_0\Gamma(n)P_0$, where $n=\sum_{j=1}^{2k}n_j$. The claim follows immediately. The argument for $P_0T_{n_1}\cdots T_{n_{2k}}$ is identical.
\end{proof}

\begin{lemma}\label{lem:auxiliaryexpressions}
Assume that
$P_0\Gamma(n)P_0\subset\mathbb{C}P_0$ for all $n<L$.
Then
\begin{align}
P_0\hat{S}^{2j-1}(\Vod)_{n-1}P_0&\in\mathbb{C}P_0\qquad \textrm{ for all }j \textrm{ and all }n< L\ .\label{eq:firstclaimimm}\\
P_0 \hat{S}^{2j-1}(\Vod)_{L-1}P_0&\in\mathbb{C}P_0\qquad\textrm{ for all }j>1\ .
\end{align}
\end{lemma}
Lemma~\eqref{lem:auxiliaryexpressions} suffices to show that the $n$-th order effective Hamiltonian~$\Heff^{(n)}$ is trivial (i.e., proportional to $P_0$)  for any order~$n<L$ (see Theorem~\ref{thm:effectivehamiltoniaschrieffer} below).

\begin{proof}
The claim~\eqref{eq:firstclaimimm} is an immediate consequence of the assumption since~$\hat{S}^{2j-1}(\Vod)_{n-1}\in \Gamma(n)$ according to~\eqref{eq:shatvod} of  Lemma~\ref{lem:basicformsnops}.

For $j>1$, we use the definition
\begin{align*}
\hat{S}^{2j-1}(\Vod)_{L-1}
&=\sum_{\substack{
n_1,\ldots,n_{2j-1}\geq 1\\
\sum_{r=1}^{2j-1}n_r=L-1
}}\ad{S_{n_1}}\cdots\ad{S_{n_{2j-1}}}(\Vod)\ .
\end{align*}
First summing over $n_1$ (using the linearity of $\ad{S_{n_1}}$), we obtain 
\begin{align}
\hat{S}^{2j-1}(\Vod)_{n-1}
&=\sum_{n_1\geq 1}\ad{S_{n_1}}(Y_{n_1})\quad\textrm{ where }\\
Y_{n_1}&=\sum_{\substack{
n_2,\ldots,n_{2j-1}\geq 1\\
\sum_{r=2}^{2j-1}n_r=L-1-n_1
}}\ad{S_{n_2}}\cdots\ad{S_{n_{2j-1}}}(\Vod)\label{eq:stwojnminusone}
\end{align}
Observe that $Y_{n_1}$ is a linear combination of 
products $T_{1}\cdots T_{2j-1}$ of an odd number~$2j-1$ of elements~$\{T_r\}_{r=1}^{2j-1}$, where $(T_1,\ldots,T_{2j-1})$ is a
permutation of $(S_{n_2},\ldots,S_{n_{2j-1}},\Vod)$. By linearity, it suffices to show that
$P_0\ad{S_{n_1}}(T_1\cdots T_{2j-1})P_0\in\mathbb{C}P_0$ for such a product. 

We will argue that
\begin{align}
T_1\cdots T_{2j-1}P_0&=TP_0\qquad\textrm{ for some }T\in \Gamma(m)\qquad\textrm{ with }m<L-n_1\textrm{ and }\label{eq:tonetwojmus}\\
P_0T_1\cdots T_{2j-1}&=P_0T'\qquad\textrm{ for some }T'\in \Gamma(m')\qquad\textrm{ with }m'<L-n_1\ .\label{eq:ttwotwojm}
\end{align}
This implies the claim since
\begin{align}
P_0\ad{S_{n_1}}(T_1\cdots T_{2j-1})P_0
&=P_0 S_{n_1}T_1\cdots T_{2j-1}P_0-P_0T_1\cdots T_{2j-1} S_{n_1}P_0\\
&=P_0 S_{n_1} TP_0-P_0T' S_{n_1}P_0\\
&\in\mathbb{C}P_0
\end{align}
where we used that $S_{n_1}T\in \Gamma(n_1+m)$
and $T'S_{n_1}\in \Gamma(n_1+m')$, $n_1+m<L$, $n_1+m'<L$ and our assumption in the last step.

To prove~\eqref{eq:tonetwojmus} (the proof of~\eqref{eq:ttwotwojm} is analogous and omitted here), 
we use that
$S_{n_j}\in \Gamma^\star(n_j)$ and $\Vod\in\Gamma^\star(1)$ 
  according to Lemma~\ref{lem:basicformsnops}. In other words, there are numbers $m_1,\ldots,m_{2j-1}\geq 1$ with 
  $\sum_{r=1}^{2j-1}m_r=1+\sum_{r=2}^{2j-1}n_r=L-n_1<L$ such that $T_r\in \Gamma^\star(m_r)$ for $r=1,\ldots,2j-1$. 
  In particular, with Lemma~\ref{lem:producttrivialpert}, we conclude that
  \begin{align}
T_1\cdots T_{2j-1}P_0&=T_1 (T_2\cdots T_{2j-1})P_0\\
&\in\mathbb{C}T_1P_0\ .
  \end{align}
  Since $m_1=L-n_1-\sum_{r=2}^{2j-1}m_r<L-n_1$, the claim~\eqref{eq:tonetwojmus} follows. 
\end{proof}

\subsubsection{Computation of the first non-trivial contribution\label{sec:nontrivialcontr}}
 Lemma~\eqref{lem:auxiliaryexpressions} also implies that the first (potentially) non-trivial term is of order~$L$, and given by~$P_0\hat{S}^1(\Vod)_{L-1}P_0$. Computing this term requires some effort.

\newcommand*{\tVd}{\mathcal{V}_d}
Let us define the superoperator $\tVd=-\cL\circ \ad{\Vd}$, that is,
\begin{align*}
\tVd(X)=\cL(X\Vd-\Vd X)\ 
\end{align*}
For later reference, we note that this operator satisfies
\begin{align}
\tVd(\Gamma^\star(n))\subset \Gamma^\star(n+1)\ .\label{eq:gradingpreserv}
\end{align}
as an immediate consequence of~\eqref{eq:clgammainclu}. 

We also define the  operators
\begin{align}
B_n&=\sum_{j\geq 1}a_j \cL(\hat{S}^{2j}(\Vod)_{n-1})\label{eq:bndefx}
\end{align}
Then we can rewrite the recursive definition~\eqref{eq:sndefinitiongeneral} of the operators $S_n$ as 
\begin{align*}
S_1 &=\cL(\Vod)\\
S_n&=\tVd(S_{n-1})+B_{n}=A_n+B_n\qquad\textrm{ for } n\geq 2\ ,
\end{align*}
where we also introduced 
\begin{align}
A_n=\tVd(S_{n-1})\qquad\textrm{ for }n\geq 2\ .\label{eq:anvtdef}
\end{align}
Similarly to Lemma~\ref{lem:auxiliaryexpressions}, 
we can show the following:
\begin{lemma}\label{lem:blpzerocompatibility}
Suppose  that $P_0\Gamma(n)P_0\subset\mathbb{C}P_0$ for all $n<L$. Then 
for any 
\begin{align*}
Y&=\begin{cases}
Z_0\Vd Z_1\Vd Z_2\cdots Z_{m-1}\Vd Z_m  &\textrm{ for }m>0\\
Z_0&\textrm{ for }m=0
\end{cases}
\end{align*}
where $Z_j\in \{P_0,Q_0\}\cup\{G^k\ | k\in\mathbb{N}\}$, we have 
\begin{align}
P_0 B_\ell Y\Vod P_0\in \mathbb{C}P_0\qquad\textrm{ and }\qquad P_0\Vod YB_\ell P_0\in\mathbb{C}P_0\ \label{eq:pb0yp0cl}
\end{align}
for all $\ell,m$ satisfying $\ell+m-1<L$. 
\end{lemma}
\begin{proof}
By definition~\eqref{eq:bndefx}, $B_\ell$ is a linear combination of terms of the form $\cL(\hat{S}^{2j}(\Vod)_{\ell-1})$ with $j\geq 1$, which in turn (cf.~\eqref{eq:skvodmdef}) is a linear combination of expressions of the form
\begin{align*}
\cL\left(\ad{S_{n_{1}}}\cdots\ad{S_{n_{2j}}}(\Vod)\right)\qquad\textrm{ where }\qquad 
\sum_{r=1}^{2j}n_r=\ell-1\ .
\end{align*}
It hence suffices to show that 
\begin{align}
P_0 \cL\left(\ad{S_{n_{1}}}\cdots\ad{S_{n_{2j}}}(\Vod)\right)Y\Vod P_0\in \mathbb{C}P_0\ .\label{eq:pzeroladsn}
\end{align}
(The proof of the second statement in~\eqref{eq:pb0yp0cl} is identical and omitted here.)

By definition of~$\cL$, the claim is true if $Z_0=P_0$, since in this case the lhs.~vanishes as $P_0\Vod P_0=0$.  Furthermore, for general $Z_0\in \{Q_0\}\cup \{G^k\ |\ k\in\mathbb{N}\}$, the claim~\eqref{eq:pzeroladsn} follows if
we can show that 
\begin{align*}
P_0 (\ad{S_{n_{1}}}\cdots\ad{S_{n_{2j}}}(\Vod))Y\Vod P_0\in \mathbb{C}P_0\ ,
\end{align*}
i.e., we can omit~$\cL$ from these considerations. This follows 
by inserting the expression~\eqref{eq:lpxg} for~$\cL$.

Observe that 
$\ad{S_{n_1}}\cdots\ad{S_{n_{2j}}}(\Vod)$ is a linear combination of products~$T_1\cdots T_{2j+1}$ of $2j+1$~operators $\{T_r\}_{r=1}^{2j+1}$, where $(T_1,\ldots,T_{2j+1})$ is a permutation of $(S_{n_1},\ldots,S_{n_{2j}},\Vod)$. That is, it suffices to show that
for each such $2j+1$-tuple of elements  $\{T_r\}_{r=1}^{2j+1}$, we have 
\begin{align}
P_0 T_1\cdots T_{2j}T_{2j+1}Y\Vod P_0\in\mathbb{C}P_0\ .\label{eq:pzerotsandwich}
\end{align}
By Lemma~\ref{lem:basicformsnops}, $T_r\in \Gamma^*(m_r)$ for some integers~$m_r\geq 1$ satisfying $\sum_{r=1}^{2j+1} m_r=1+\sum_{r=1}^{2r}n_r=1+ \ell-1$. 
This implies (by our assumption $\ell+m-1<L$) that 
\begin{align}
\sum_{r=1}^{2j}m_r=\ell-1<L\ ,\label{eq:summrdef}
\end{align} and thus $P_0T_1\cdots T_{2j}\in \mathbb{C}P_0$ according to Lemma~\ref{lem:producttrivialpert}. We conclude that
\begin{align}
P_0 T_1\cdots T_{2j}T_{2j+1}Y\Vod P_0&\in\mathbb{C}P_0T_{2j+1}Y\Vod P_0\\
&\in\mathbb{C} P_0\Gamma(m_{2j+1}+m+1)P_0\ .
\end{align}
But by~\eqref{eq:summrdef} and the because $j\geq 1$, we have 
\begin{align}
m_{2j+1}+(m+1)&=(\ell-1-\sum_{r=1}^{2j}m_r)+(m+1)\\
&\leq (\ell-1-2j)+(m+1)\leq \ell+m-2<L 
\end{align}
by assumption on $\ell,m$ and $L$, hence~\eqref{eq:pzerotsandwich} follows from 
the assumption~\eqref{eq:pzerotsandwich}. 
\end{proof}

\begin{lemma}\label{lem:bnadjointprojected}
Suppose  that $P_0\Gamma(n)P_0\subset\mathbb{C}P_0$ for all $n<L$. Then for any $\ell,m$ satisfying $\ell+m-1<L$, we have 
\begin{align}
P_0\tVd^{\circ m}(B_\ell)\Vod P_0&\in\mathbb{C}P_0\qquad\textrm{ and }\qquad
P_0\Vod\tVd^{\circ m}(B_\ell) P_0\in\mathbb{C}P_0\label{eq:vdmblprod}
\end{align}
In particular, for every $q<L$, we have 
\begin{align}
P_0\ad{\tVd^{\circ k+1}(B_{q-k})}(\Vod)P_0\in\mathbb{C}P_0\ \label{eq:adtvdbqclaim}
\end{align}
for all $k=0,\ldots,q-2$.
Furthermore, 
\begin{align}
P_0\ad{B_\ell}(\Vod)P_0\in\mathbb{C}P_0\qquad\textrm{ for all }\ell\leq L\ . \label{eq:vqozerodefa}
\end{align}
\end{lemma}
\begin{proof}
By definition of $\tVd$, the expression 
$\tVd^{\circ m}(B_\ell)$ is a linear combination of terms of the form
\begin{align*}
A^L B_\ell A^R\qquad\textrm{ where }\qquad
\begin{matrix}
A^L&=&Z_0\Vd Z_1\cdots Z_{r-1}\Vd Z_r \\
A^R&=&Z_{r+1}\Vd Z_{r+2}\cdots Z_m\Vd Z_{m+1}\ 
\end{matrix}
\end{align*}
and each $Z_j\in \{P_0,Q_0\}\cup \{G^m\ | m\in\mathbb{N}\}$. 
Since $A^L$ only involves diagonal operators and the number of
factors~$\Vd$ is equal to~$r<L$, we have $P_0A^{L}=P_0A^LP_0\in P_0\Gamma(r)P_0\in \mathbb{C}P_0$.
In particular,
\begin{align*}
P_0(A^L B_{\ell} A^R)\Vod P_0 &\in\mathbb{C} P_0B_{\ell}A^R\Vod P_0\ .
\end{align*}
But
\begin{align*}
P_0B_\ell A^R\Vod P_0\in\mathbb{C}P_0\ ,
\end{align*}
where we applied Lemma~\ref{lem:blpzerocompatibility} 
with $Y=A^R$ (note that $A^R$ involves $m-r$ factors~$\Vd$, 
and $\ell+(m-r)-1<L$ by assumption).  We conclude that 
\begin{align*}
P_0(A^LB_\ell A^R)\Vod P_0\in\mathbb{C}P_0\ ,
\end{align*}
and since
$P_0\tVd^{\circ m}(B_\ell)\Vod P_0$ is a linear combination of such terms, the first identity in~\eqref{eq:vdmblprod} follows. The second identity is shown in an analogous manner.

The claim~\eqref{eq:adtvdbqclaim} follows by setting $m=k+1$ and $\ell=q-k$, and observing that $\ell+m-1=q<L$.

Finally, consider the claim~\eqref{eq:vqozerodefa}. We have
\begin{align}
P_0B_{\ell}\Vod P_0&=P_0B_{\ell}Q_0\Vod P_0\in\mathbb{C}P_0\\
P_0\Vod B_\ell P_0&=P_0\Vod Q_0 B_\ell P_0
\end{align}
for all $\ell$ with $\ell-1<L$ by Lemma~\ref{lem:blpzerocompatibility}, hence the claim follows. 
\end{proof}

\begin{lemma}\label{lem:finaldjointpzeroaqvod}
Suppose  that $P_0\Gamma(n)P_0\subset\mathbb{C}P_0$ for all $n<L$. 
Then
\begin{align*}
P_0\ad{\tVd(A_q)}(\Vod)P_0&\in P_0\ad{\tVd^{\circ q}(\cL(\Vod))}(\Vod) P_0+\mathbb{C}P_0\ .
\end{align*}
for all $q<L$. 
\end{lemma}
\begin{proof}
We will show that for $k=1,\ldots,q-2$, we have the identity
\begin{align}
P_0\ad{\tVd^{\circ k}(A_{q+1-k})}(\Vod)P_0&\in
P_0\ad{\tVd^{\circ k+1}(A_{q-k})}(\Vod)P_0+\mathbb{C}P_0\ .\label{eq:kprovidentityv}
\end{align}
(Notice that the expression on the rhs.~is obtained from the lhs~by substituting $k+1$ for $k$.) Iteratively applying this implies 
\begin{align*}
P_0\ad{\tVd(A_q)}(\Vod)P_0&\in P_0\ad{\tVd^{\circ q-1}(A_2)}(\Vod)P_0+\mathbb{C}P_0\  ,
\end{align*}
from which the claim follows since $A_2=\tVd(\cL(\Vod))$.

To prove~\eqref{eq:kprovidentityv}, observe that by definition~\eqref{eq:anvtdef} of $A_n$, we have by linearity of $\tVd$  
\begin{align*}
\tVd^{\circ k}(A_{q+1-k})&=
\tVd^{\circ k+1}(S_{q-k})=\tVd^{\circ k+1}(A_{q-k})+\tVd^{\circ k+1}(B_{q-k})\ .
\end{align*}
By linearity of the map $X\mapsto P_0\ad{X}(\Vod)P_0$, it thus suffices to show that
\begin{align*}
P_0\ad{\tVd^{\circ k+1}(B_{q-k})}(\Vod)P_0\in\mathbb{C}P_0\ 
\end{align*}
for all $k=1,\ldots,q-2$. This follows from~\eqref{eq:adtvdbqclaim} of Lemma~\ref{lem:bnadjointprojected}.
\end{proof}

\begin{lemma}\label{lem:commutatorexprvex}
Suppose  that $P_0\Gamma(n)P_0\subset\mathbb{C}P_0$ for all $n<L$.  Then
\begin{align}
P_0\hat{S}^1(\Vod)_{n-1}P_0&= P_0\ad{\tVd^{\circ n-2}(\cL(\Vod))}(\Vod)P_0+\mathbb{C}P_0\qquad\textrm{ for all }n<L+2\ .\label{eq:firstclaimlemmacommut}
\end{align}
\end{lemma}
\begin{proof}
By definition~\eqref{eq:skvodmdef} and the linearity of $\ad{\cdot}$, we have 
\begin{align*}
\hat{S}^1(\Vod)_{n-1}&=\ad{S_{n-1}}(\Vod)=\ad{A_{n-1}}(\Vod)+\ad{B_{n-1}}(\Vod)\ .
\end{align*}
But by definition of $A_{n-1}$, we have if $n-2<L$
\begin{align*}
P_0\ad{A_{n-1}}(\Vod)P_0&=P_0\ad{\tVd(S_{n-2})}(\Vod)P_0\\
&= P_0\ad{\tVd(A_{n-2})}(\Vod)P_0+P_0\ad{\tVd(B_{n-2})}(\Vod)P_0\\
&\in P_0\ad{\tVd^{\circ n-2}(\cL(\Vod))}(\Vod)P_0+\mathbb{C}P_0
\end{align*}
where we again used the linearity of the involved operations in the second step and 
Lemma~\ref{lem:bnadjointprojected} and Lemma~\ref{lem:finaldjointpzeroaqvod} in the last step (with $k=0$ and $q=n-2$).

Similarly, we have (again by Lemma~\ref{lem:bnadjointprojected}) if $n-2<L$.
\begin{align*}
P_0 \ad{B_{n-1}}(\Vod)P_0&=\mathbb{C}P_0\ .
\end{align*}
The claim~\eqref{eq:firstclaimlemmacommut} follows.
\end{proof}

\begin{lemma}\label{lem:firstnontrivialorder}
Suppose  that $P_0\Gamma(n)P_0\subset\mathbb{C}P_0$ for all $n<L$. Then
\begin{align}
P_0\hat{S}^1(\Vod)_{L-1}P_0&=2P_0VGVG\cdots GVP_0\ ,\label{eq:toprovecommutator}
\end{align}
where there are $L$ factors~$V$ on the rhs.
\end{lemma}
\begin{proof}
We will first show inductively for $k=1,\ldots,n-2$ that 
\begin{align}
\tVd^{\circ k}(\cL(\Vod))=- ((GV)^{k+1} P_0 -h.c.)+T_k\qquad\textrm{ for some }T_k\in \Gamma^\star(k)\ .\label{eq:assumptioninducv}
\end{align}
By straightforward computation, we have 
\begin{align*}
\cL(\Vod)&= P_0 \Vod G - h.c.\\
[\cL(\Vod), \Vd]&=-G \Vod P_0 \Vd+ \Vd G \Vod P_0+h.c.\\
\tVd(\cL(\Vod)) &= -(G \Vd G \Vod P_0- h.c.) +  T_1,
\end{align*}
where $T_1=G^2 \Vod P_0 \Vd P_0-h.c.$. By assumption, $P_0 \Vd P_0=P_0\Vd P_0\in\mathbb{C}P_0$. Thus $T_1\in \Gamma^*(1)$, and the claim~\eqref{eq:assumptioninducv} is verified for $k=1$ (since $G\Vd G\Vod P_0=GVGVP_0$).

Now assume that~\eqref{eq:assumptioninducv} holds for some $k\leq n-1$. We will show that it is also valid for $k$ replaced by~$k+1$. With the assumption, we have
\begin{align}
\tVd^{\circ k+1}(\cL(\Vod))&=\tVd\left(\tVd^{\circ k}(\cL(\Vod))\right)\\
&=-\tVd\left((G V)^{k+1}P_0-h.c.\right)+\tVd(T_k)
\end{align}
But
\begin{align}
\tVd((G V)^{k+1}P_0)&=\cL((G V)^{k+1}P_0\Vd-\Vd(G V)^{k+1} P_0)\\
&=-G(GV)^{k+1}P_0\Vd P_0+G\Vd(G V)^{k+1}P_0\\
&=-G(GV)^{k+1}P_0V P_0+(G V)^{k+2}P_0
\end{align}
and by doing a similar computation for the Hermitian conjugate we find
\begin{align}
\tVd^{\circ k+1}(\cL(\Vod))&=-
\left((GV)^{k+2}P_0-h.c.\right)+T_{k+1}\qquad\textrm{ where }\\
T_{k+1}&=\left(G(GV)^{k+1}P_0VP_0-h.c.\right)+\tVd(T_k)\ .
\end{align}
We claim that $T_{k+1}\in \Gamma^\star(k+1)$. Indeed, by assumption we have $P_0VP_0\in P_0\Gamma(1)P_0\subset\mathbb{C}P_0$, hence  $G(GV)^{k+1}P_0VP_0\in\mathbb{C}G(GV)^{k+1}P_0\subset \Gamma^\star(k+1)$ and the same reasoning applies to the Hermitian conjugate. Furthermore, for $T_k\in \Gamma^\star(k)$, we have $\tVd(T_k)\in\Gamma^\star(k+1)$ by~\eqref{eq:gradingpreserv}.

This concludes the proof of~\eqref{eq:assumptioninducv}, which we now apply with $k=L-2$ to get
\begin{align}
P_0\ad{\tVd^{\circ L-2}(\cL(\Vod))}(\Vod)P_0&=P_0\tVd^{\circ L-2}(\cL(\Vod))\Vod P_0-P_0\Vod\tVd^{\circ L-2}(\cL(\Vod))P_0\\ 
&=P_0(VG)^{L-1}\Vod P_0+P_0\Vod (GV)^{L-1}P_0\\
&\qquad +P_0T_{L-2}\Vod P_0-P_0\Vod T_{L-2}P_0
\end{align}
Since $P_0T_{L-2}\Vod P_0$ and $P_0\Vod T_{L-2}P_0$ are elements of $P_0\Gamma(L-1)P_0$, we conclude that 
\begin{align}
P_0\ad{\tVd^{\circ L-2}(\cL(\Vod))}(\Vod)P_0&=P_0\tVd^{\circ L-2}(\cL(\Vod))\Vod P_0-P_0\Vod\tVd^{\circ L-2}(\cL(\Vod))P_0\\ 
&=P_0((VG)^{L-1}V+V(GV)^{L-1})P_0+\mathbb{C}P
\end{align}
Finally, with the expression obtained by Lemma~\ref{lem:commutatorexprvex} (with $n=L$), we get
\begin{align}
P_0\hat{S}^1(\Vod)_{L-1}P_0&=
 P_0\ad{\tVd^{\circ L-2}(\cL(\Vod))}(\Vod)P_0\\
&=2P_0(VG)^{L-1}VP_0+\mathbb{C}P\ ,
\end{align}
as claimed.

\end{proof}

\subsubsection{Equivalence of self-energy method and Schrieffer-Wolff transformation}
With Lemma~\ref{lem:bnadjointprojected}, Lemma~\ref{lem:finaldjointpzeroaqvod}
and Lemma~\ref{lem:firstnontrivialorder}, we now have the expressions necessary to obtain effective Hamiltonians.
\begin{theorem}[Theorem~\ref{thm:effectivehamiltoniaschrieffer} in the main text]
Suppose that
$P_0\Gamma(n)P_0\subset\mathbb{C}P_0$ for all $n< L$. 
Then the $n$-th order Schrieffer-Wolff effective Hamiltonian satisfies
\begin{align}
\Heff^{(n)}\in \mathbb{C}P_0\qquad\textrm{for all }n<L\ ,\label{eq:hefftrivialloworder}
\end{align}
i.e., the effective Hamiltonian is trivial for these orders, and
\begin{align}
\Heff^{(L)}=2 b_1 P_0(VG)^{L-1}VP_0+\mathbb{C}P_0\ ,\label{eq:hleffcomputed}
\end{align}
and where there are $L$ factors $V$ involved. 
\end{theorem}
\begin{proof}
Consider the definition~\eqref{eq:qthordereffectiveHamiltoniandef} of the $n$-th order term~${\Heff}_{,n}$ in the expansion~\eqref{eq:heffexpansion}: we have
 \begin{align*}
{\Heff}_{,n}&=\sum_{1\leq j\leq \lfloor n/2\rfloor}b_{2j-1}P_0\hat{S}^{2j-1}(\Vod)_{n-1}P_0\ .
\end{align*}
For $n<L$, each term~$P_0\hat{S}^{2j-1}(\Vod)_{n-1}P_0$ is proportional to $P_0$ (see~\eqref{eq:firstclaimimm}
of Lemma~\ref{lem:auxiliaryexpressions}), hence the claim~\eqref{eq:hefftrivialloworder} follows.

On the other hand, for $n=L$,  we have
\begin{align}
P_0\hat{S}^{2j-1}(\Vod)_{L-1}P_0&
\begin{cases}
\in \mathbb{C}P_0\qquad&\textrm{ if }j>1\\
P_0VGVGV\cdots GVP_0\qquad &\textrm{  if } j=1
\end{cases}
\end{align}
according to Lemma~\ref{lem:auxiliaryexpressions}
and Lemma~\ref{lem:firstnontrivialorder}, hence~\eqref{eq:hleffcomputed} follows. 
\end{proof}

\newpage

\section{On a class of single-qudit operators in the Levin-Wen model\label{sec:stringFibonacci}}
In this appendix, we consider the action of certain single-qudit operators and discuss how they affect states in the Levin-Wen model.  For simplicity, we will restrict our attention to models where each particle satisfies~$\bar{a}=a$, i.e., is its own antiparticle.  Similar local operators were previously considered (for example, in~\cite{burnell2011condensation}). We introduce the operators in Section~\ref{sec:defalgb} and compute the associated effective Hamiltonians in Section~\ref{sec:effectivepertbv}

\subsection{Definition and algebraic properties of certain local operators\label{sec:defalgb}}
Recall that for each qudit in the Levin-Wen model,
there is an orthonormal basis $\{\ket{a}\}_{a\in\cF}$ indexed by particle labels.
For each particle~$a\in \cF$, we define an operator acting diagonally in the orthonormal basis as
\begin{align}
O_a\ket{b}&=\frac{S_{ab}}{S_{1b}}\ket{b}\qquad\textrm{ for all }b\in\cF\ .\label{eq:singlequditspinopl}
\end{align}

As an example, consider the Pauli-$Z$ operator defined in Section~\ref{sec:doubledsemion} for the doubled semion model. Because
the $S$-matrix of the semion model is given by (see e.g.,~\cite[Section 2.4]{schulz2013topological})
\begin{align}
S&=\frac{1}{\sqrt{2}}\begin{bmatrix}
1 & 1\\
1 & -1
\end{bmatrix}
\end{align}
with respect to the (ordered) basis $\{\ket{\bfone},\ket{\bfs}\}$, 
the operator $O_{\bfs}$ takes the form
\begin{align}
O_{\bfs}&=\mathsf{diag}(1,-1)=Z\ \label{eq:Zobsfsem}
\end{align}
according to~\eqref{eq:singlequditspinopl}. 

As another example, we can use the fact that the Fibonacci model has 
$S$-matrix (with respect to the basis $\{\ket{\bfone},\ket{\bftau}\}$)
\begin{align*}
S=\frac{1}{\sqrt{1+\varphi^2}} \begin{bmatrix}
1 & \varphi \\
\varphi & -1
\end{bmatrix} \ 
\end{align*}
to obtain
\begin{align}
O_{\bftau}&=
\mathsf{diag}(\varphi,-1/\varphi)\ .
\end{align}
Therefore, the Pauli-$Z$-operator in the doubled Fibonacci model takes the form
\begin{align}
Z&=\frac{\varphi}{\varphi+2}\left(-I+2O_{\bftau}\right)\ ,\label{eq:Zfibdefa}
\end{align}
where $I$ is the identity matrix.

We will write $O_a^{(e)}=O_a$ for the operator $O_a$ applied to the qudit on the edge~$e$ of the lattice. To analyze the action of such an operator~$O_a^{(e)}$ on ground states of the Levin-Wen model, we used the ``fattened honeycomb'' description 
of (superpositions) of string-nets: this gives a compact representation of 
the action of certain operators (see the appendix of~\cite{levin2005string}), as well as a representation of ground states (see~\cite{koenig2010quantum}).
In this picture, states of the many-spin system are expressed as
superpositions of string-nets (ribbon-graphs) embedded in a surface where each plaquette is punctured. Coefficients in the computational basis of the qudits can be obtained by a process of ``reduction to the lattice'', i.e., the application of $F$-moves, removal of bubbles etc.~similar to the discussion in Section~\ref{sec:anyonchains}. Importantly, the order of reduction does not play a role in obtaining these coefficients as a result of MacLane's theorem (see the appendix of~\cite{kitaev2006anyons}). Note, however,  that this diagrammatic formalism only makes sense in the subspace
\begin{align}
\cH_{\textrm{valid}}=\span \{\ket{\psi}\ |\ A_v\ket{\psi}=\ket{\psi}\textrm{ for all vertices } v\}
\end{align}
spanned by valid string-net configurations, since otherwise reduction is not well-defined.

This provides a significant simplification for certain computations. For example, application of a plaquette operator~$B_p$ 
corresponds -- in this terminology -- to the insertion of a ``vacuum loop'' times a factor~$1/D$. The latter is itself a superposition of strings, where each string of particle type~$j$ carries a coefficient~$\frac{d_j}{D}$. 
We will represent such vacuum strings by dotted lines below:
\begin{align*}
\includegraphics[scale=0.5]{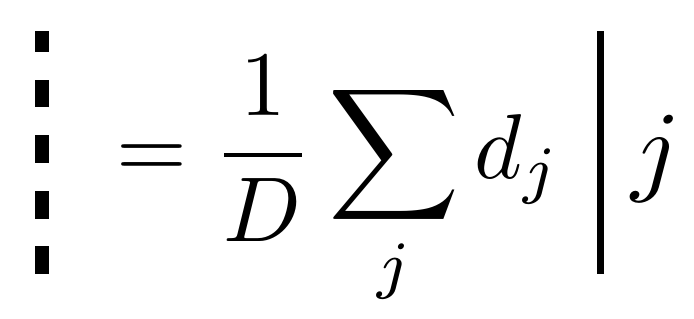}
\end{align*}
Crucial properties of this superposition are (see~\cite[Lemma A.1]{koenig2010quantum})
\begin{align*}
\includegraphics[scale=0.5]{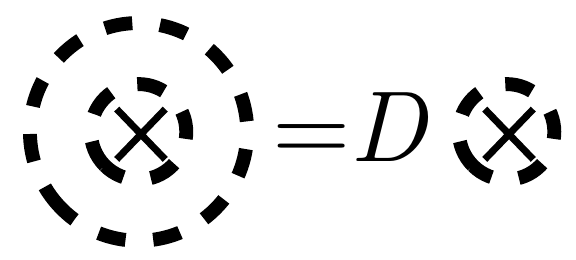}
\end{align*}
\begin{align*}
\includegraphics[scale=0.5]{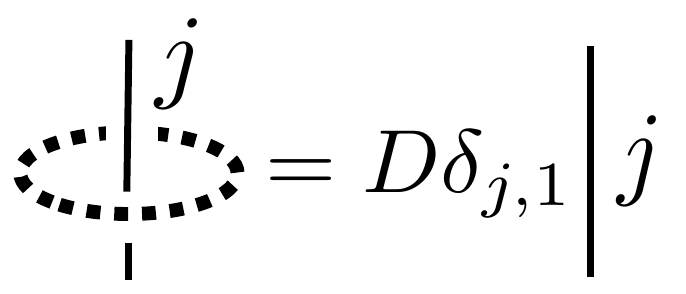}
\end{align*}
and the pulling-through rule
\begin{align*}
\includegraphics[scale=0.5]{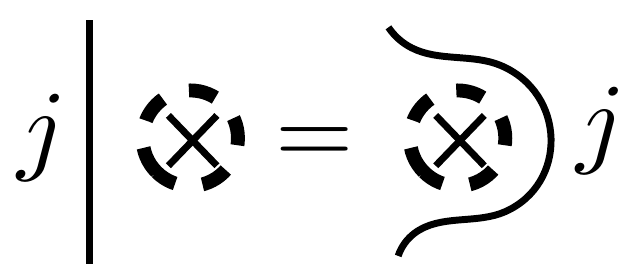}
\end{align*}

Similarly, a single-qudit operator~$O_a^{(e)}$  of the form~\eqref{eq:singlequditspinopl} can be expressed in this language, and takes the form of adding a ``ring'' around a line: we have
\begin{align}
O_a^{(e)}\ket{b}=\raisebox{-3mm}{\includegraphics[scale=0.4]{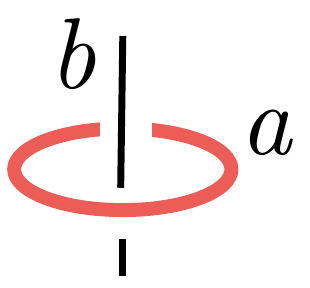}} \ .
\end{align}
(The color is only used to emphasize the application of the operator, but is otherwise of no significance.)

\begin{lemma}\label{lem:annihilationgroundspacex}
Let $a\neq 1$, and let $O_a^{(e)}$ be an operator of the form~\eqref{eq:singlequditspinopl} acting an an edge~$e$ of the qudit lattice. 
Let $p,p'$ be the two  plaquettes adjacent to the edge~$e$, and let  $B_{p}, B_{p'}$ be the associated operators. Then
for any $\ket{\psi}\in\cH_{\textrm{valid}}$, we have 
\begin{align}
B_p\ket{\psi}&=\ket{\psi}\qquad\Rightarrow\qquad B_p(O_a^{(e)}\ket{\psi})=0\\
B_{p'}\ket{\psi}&=\ket{\psi}\qquad\Rightarrow\qquad B_{p'}(O_a^{(e)}\ket{\psi})=0\\
\end{align}

\end{lemma}
For example, for any ground state~$\ket{\psi}$ of the Levin-Wen model~$\Htop$, $O_a^{(e)}\ket{\psi}$ is an eigenstate of~$\Htop$ with energy~$2$.
Furthermore, for any ground state~$\ket{\psi}$,
and any edges $e_1,\ldots,e_n$ which 
 (pairwise) do not belong to the same plaquette,
the state~$O_a^{(e_1)}\cdots O_a^{(e_n)}\ket{\psi}$ is an eigenstate (with energy~$2n$) of $\Htop$. The case where the edges belong to the same plaquette will be discussed below in Lemma~\ref{lem:productofsinglequdit}. 

\begin{proof}
For concreteness, consider the plaquette operator~$B_p$ ``on the left'' of the edge (the argument for the other operator is identical).
Because $\ket{\psi}$ is a ground state, we have $B_p\ket{\psi}=\ket{\psi}$. 
Using the graphical calculus (assuming that the state~$\ket{\psi}$ is expressed as a string-net embedded in the gray lattice), we obtain
\begin{align}
B_pO_a^{(e)} B_p\ket{\psi} & =\frac{1}{D^2} \raisebox{-8mm}{\includegraphics[scale=0.4]{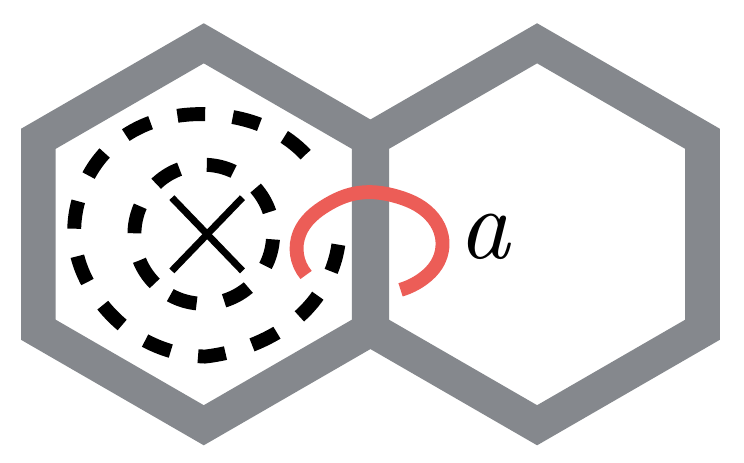}} \\
& =\frac{1}{D^2} \raisebox{-8mm}{\includegraphics[scale=0.4]{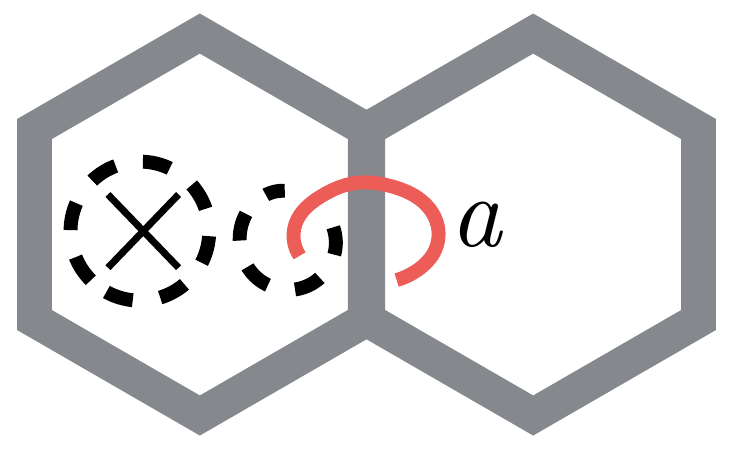}}\\
& =\frac{1}{D} \delta_{a,1} \raisebox{-8mm}{\includegraphics[scale=0.4]{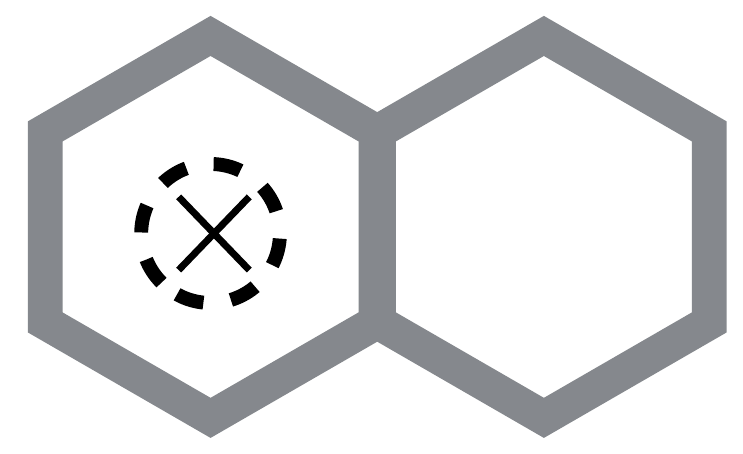}}\\
& =\delta_{a,1} \ket{\psi}
\end{align}
\end{proof}

\begin{lemma}\label{lem:productofsinglequdit}
Let
$e_1\neq e_2$ be two edges lying on the same plaquette~$p$, and let
 us assume that they lie on  opposite sides of the plaquette~$p$ (this
assumption is for concreteness only and can be dropped).
Let $O^{(e_1)}_a$ and $O^{(e_2)}_a$ be the associated single-qudit operators (with $a\neq 1$). 
Then for all $\ket{\psi}\in\cH_{\textrm{valid}}$, we have 
\begin{align}
B_p O^{(e_1)}_aO^{(e_2)}_aB_p\ket{\psi}=\frac{d_a}{D}B_p O^{(e_1e_2)}_aB_p\ket{\psi}\ ,
\end{align}
where the operator $O^{(e_1e_2)}_a$ is defined by
\begin{align}
O^{(e_1e_2)}_a\ket{\psi} &= \raisebox{-8mm}{\includegraphics[scale=0.4]{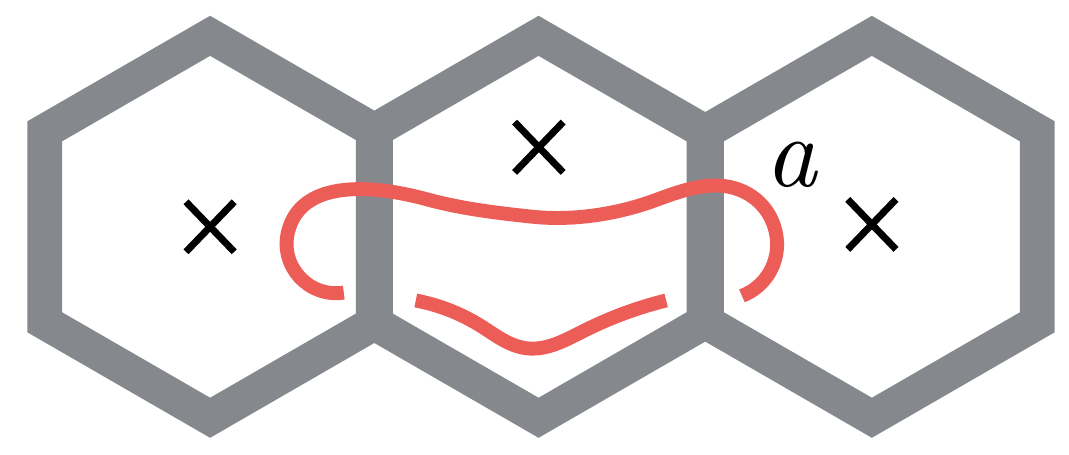}}
\end{align}
in the diagrammatic formalism. In other words, $O_a^{(e_1e_2)}$ adds a single loop of type~$a$ around the edges~$e_1,e_2$. 
\end{lemma}

\begin{proof}
Let $\ket{\psi}\in\cH_{\textrm{valid}}$. 
Then we have by a similar computation as before
\begin{align}
B_p(O^{(e_1)}_a O^{(e_2)}_a)B_p \ket{\psi}&= B_p \frac{1}{D} \raisebox{-8mm}{\includegraphics[scale=0.4]{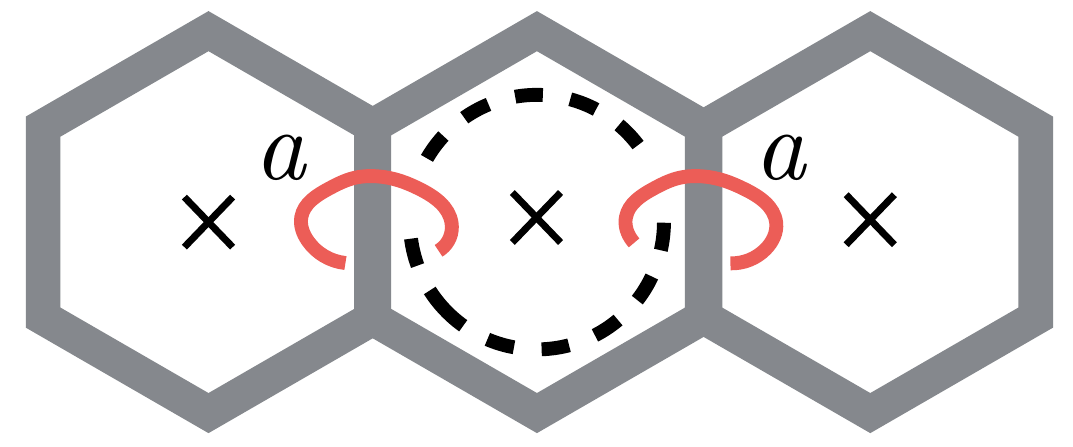}} \\
& = B_p \frac{1}{D}\sum_k F^{aa1}_{aak} \raisebox{-8mm}{\includegraphics[scale=0.4]{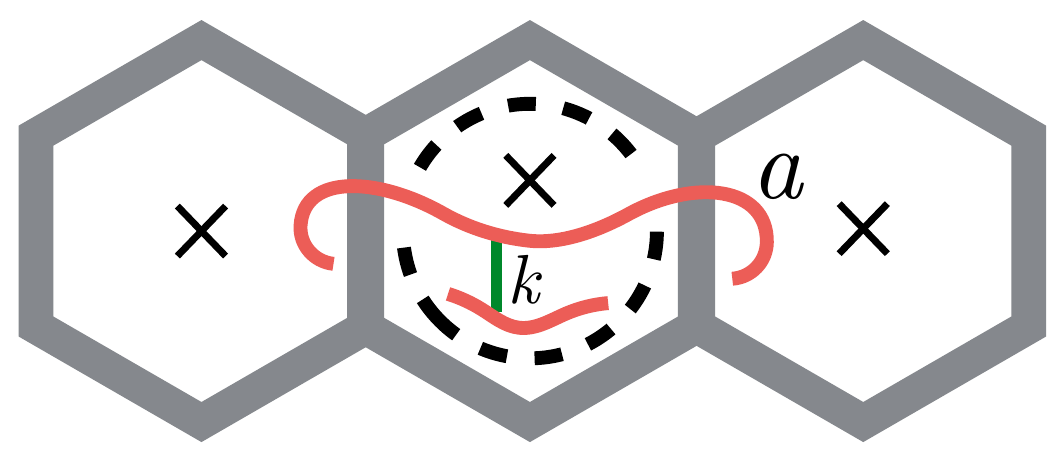}} \\
& = \frac{1}{D^2}\sum_k F^{aa1}_{aak} \raisebox{-8mm}{\includegraphics[scale=0.4]{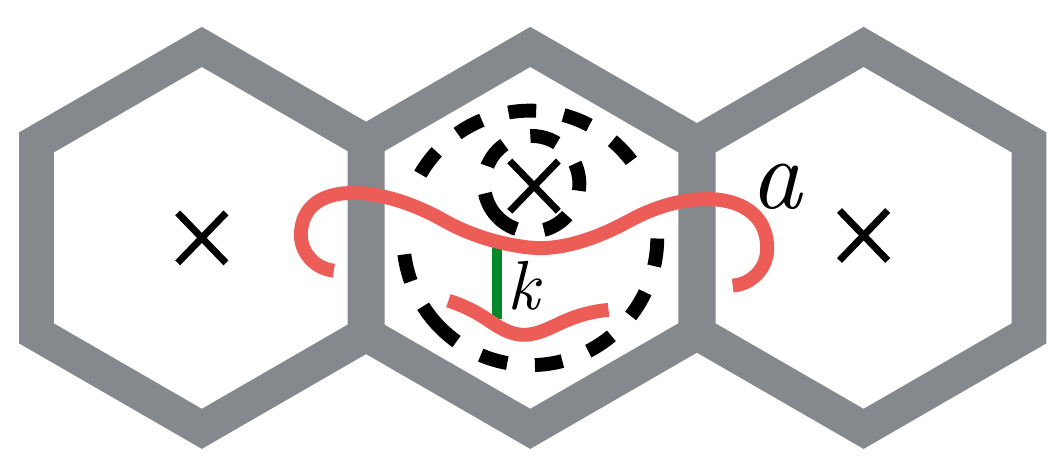}} \\
& = \frac{1}{D^2}\sum_k F^{aa1}_{aak} \raisebox{-8mm}{\includegraphics[scale=0.4]{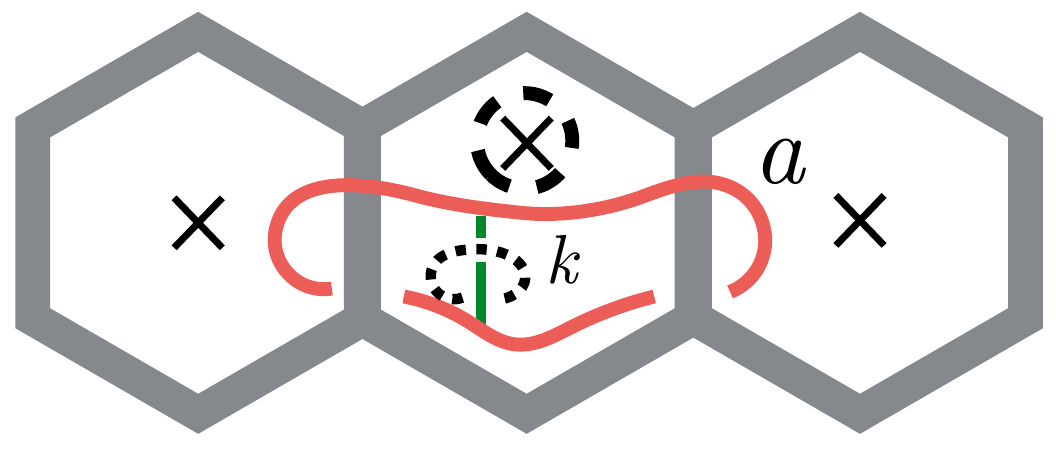}} \\
& = \frac{1}{D^2}\sum_k F^{aa1}_{aak} D \delta_{k1} \raisebox{-8mm}{\includegraphics[scale=0.4]{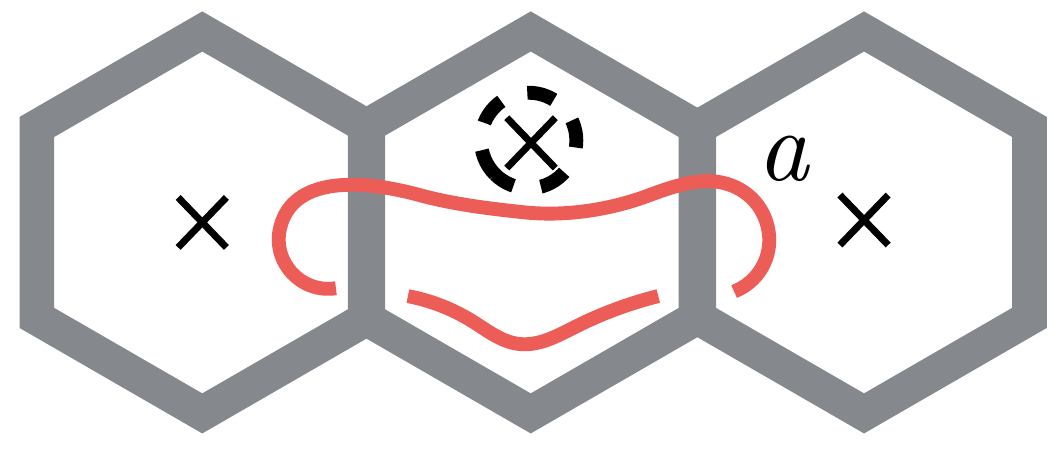}} \\
& = B_p  \frac{1}{d_a} \raisebox{-8mm}{\includegraphics[scale=0.4]{./figures/app_comp2_6}} \\
&= \frac{1}{d_a} B_p O^{(e_1e_2)}_aB_{p}\ket{\psi}\ ,
\end{align}
as claimed. 
\end{proof}
Clearly, the reasoning  of Lemma~\ref{lem:productofsinglequdit}
can be applied inductively to longer sequences of products~$O^{(e_1)}_a O^{(e_2)}_a \cdots O^{(e_k)}_a$ if the edges
$\{e_1,\ldots,e_k\}$ correspond to  a path on the dual lattice, giving rise to certain operators $O^{(e_1\cdots e_k)}_a$ with a nice graphical representation: we have for example
\begin{align}
P_0O^{(e_1)}_a O^{(e_2)}_a \cdots O^{(e_k)}_a P_0 =c\cdot P_0O^{(e_1\cdots e_k)}_aP_0
\end{align}
for some constant~$c$, where $P_0$ is the projection onto the ground space of the Levin-Wen model and where $O^{(e_1\cdots e_k)}_a$ is the operator given in the diagrammatic formalism as 
 \begin{align}
O^{(e_1\cdots e_k)}_a\ket{\psi}&= \raisebox{-8mm}{\includegraphics[scale=0.4]{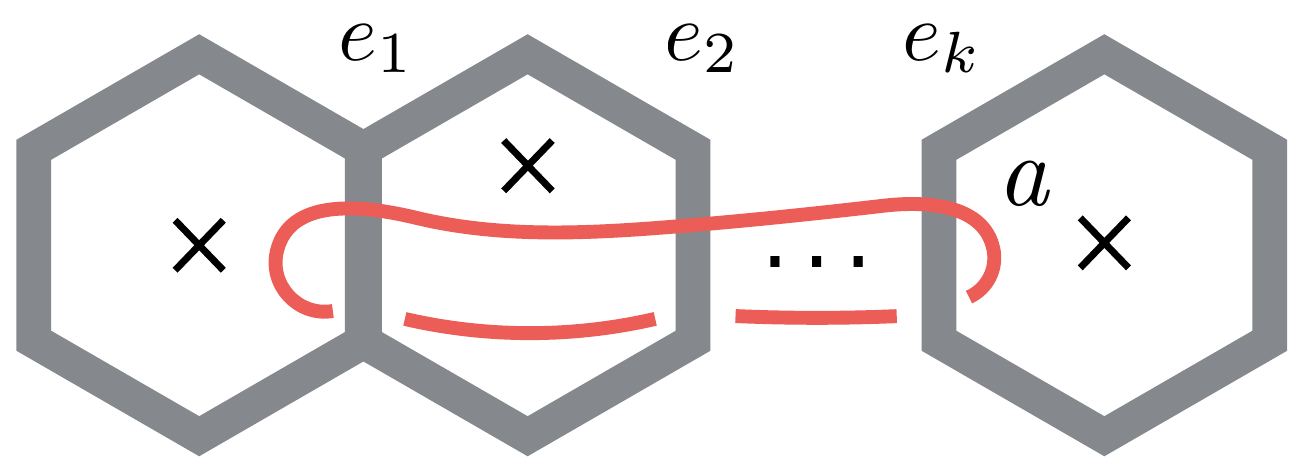}} \label{eq:oonetokfop}
\end{align}
\begin{figure}
\begin{subfigure}	{0.3\textwidth}
\includegraphics[width=\textwidth]{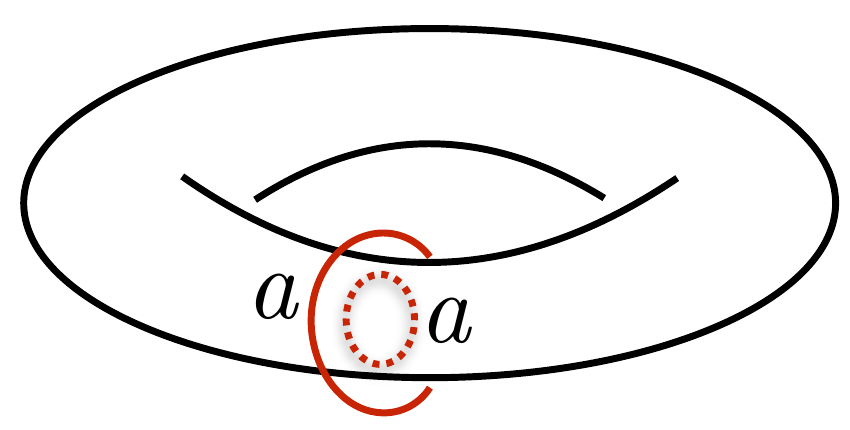}
\caption{$F_{(a,a)}(C)$}
\label{fig_twoloops}
\end{subfigure}
\begin{subfigure}	{0.3\textwidth}
\includegraphics[width=\textwidth]{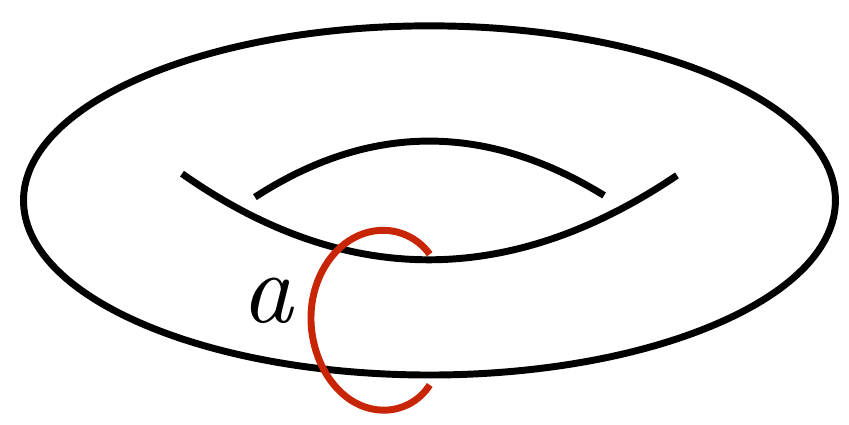}
\caption{$F_{(1,a)}(C)$}
\end{subfigure}
\begin{subfigure}	{0.3\textwidth}
\includegraphics[width=\textwidth]{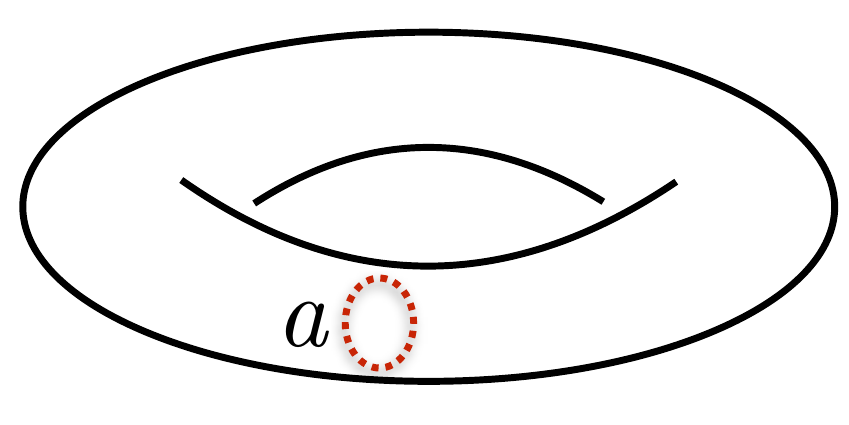}
\caption{$F_{(a,1)}(C)$}
\end{subfigure}
\caption{The graphical representation of certain anyonic string-operators in the doubled model. The dashed line is inside the torus.}
\label{fig_graphical_rep_Fib}
\end{figure}

Using this fact, we can relate certain products of operators to the string-operator~$F_{(a,a)}(C)$ associated with the (doubled) anyon~$(a,a)$. That is, assume that  the edges~$\{e_1,\ldots,e_L\}$ cover a topologically non-trivial loop~$C$ on the (dual) lattice (e.g., $\{10,12\}$ in the $12$-qudit torus of Fig.~\ref{fig_lattice}). Then we have 
\begin{align}
P_0(O^{(e_1)}_a O^{(e_2)}_a\cdots O^{(e_L)}_a)P_0=c\cdot F_{a,a}(C)\label{eq:productooperators}
\end{align}
for some constant~$c$. This follows by comparing~\eqref{eq:oonetokfop} with the graphical representation of the string-operators of the doubled model as discussed in~\cite{levin2005string}, see Fig.~\ref{fig_graphical_rep_Fib}. Note also that by the topological order condition, operators of the form~$P_0O_a^{(e_1)}O_a^{(e_2)}\cdots O_a^{(e_k)}P_0$ are  proportional to~$P_0$ if~$k<L$.

\subsection{Effective Hamiltonians for translation-invariant perturbation\label{sec:effectivepertbv}}

According to~\eqref{eq:Zobsfsem} and~\eqref{eq:Zfibdefa}, a translation-invariant perturbation of the form $V=\sum_j Z_j$ for the doubled semion or Fibonacci models (as considered in Section~\ref{sec:numerics}) is, up to a global energy shift and a proportionality constant, equivalent to a perturbation of the form
\begin{align}
V&=\sum_e O^{(e)}_a\ ,\label{eq:translationinvariantVoap}
\end{align}
where $a\neq 1$  and the sum is over all edges~$e$ of the lattice (Here $a=\bfs$ in the doubled semion model and $a=\bftau$ in the Fibonacci model). We show the following:
\begin{lemma}
For the perturbation~\eqref{eq:translationinvariantVoap} to the Levin-Wen model~$H_0$, the $L$-th order effective Hamiltonian is given by 
 \begin{align}
 \Heff^{(L)}&= c_1 \left( \sum_C   F_{(a,a)}(C)\right)+c_2P_0\ ,\label{eq:heffLlevw}
 \end{align}
where $c_1$ and $c_2$ are constants, and the sum is over all topologically non-trivial loops~$C$ of length~$L$.
\end{lemma}

\begin{proof}
According to Theorem~\ref{thm:effectivehamiltoniaschrieffer}, the $L$-th order effective Hamiltonian is proportional to 
\begin{align}
P_0 (VG)^{L-1}VP_0&=\sum_{e_1,\ldots,e_L} P_0O^{(e_1)}_a G O^{(e_2)}_a G\cdots GO^{(e_L)}_aP_0
\end{align} 
up to an energy shift. By the topological order constraint,
the only summands on the rhs.~which can have a non-trivial action on the ground space are those associated with edges~$\{e_1,\ldots,e_L\}$ constituting a non-trivial loop~$C$ on the (dual) lattice. Note  that for such a collection of edges,
every plaquette~$p$ has at most two edges $e_j,e_k\in\{e_1,\ldots,e_L\}$
as its sides, a fact we will use below.
Our claim follows if we show that for any such collection of edges, we have 
\begin{align}
P_0O^{(e_1)}_a G O^{(e_2)}_a G\cdots GO^{(e_L)}_aP_0&= c F_{(a,a)}(C)\label{eq:pzeroprodv}
\end{align}
for some constant~$c$. 

We show~\eqref{eq:pzeroprodv}
by showing that the resolvent operators~$G$ only 
contribute a global factor; the claim then follows from~\eqref{eq:productooperators}. The reason is that the local operators~$O_a^{(e_\ell)}$ create localized excitations, and these cannot be removed unless operators acting on the edges of neighboring plaquettes are applied. Thus 
a process as the one on the lhs.~\eqref{eq:productooperators} is equivalent to one which goes through a sequence of eigenstates of the unperturbed Hamiltonian~$H_0$.

The proof of this statement is a bit more involved since 
operators~$O_a^{(e_\ell)}$ can also create superpositions of excited and ground states. We proceed inductively.
Let us set
\begin{align}
\begin{matrix}
\Lambda_1&=& P_0 \qquad\qquad & \Gamma_1&=& O_a^{(e_1)}GO^{(e_2)}_aGO^{(e_3)}_aG\cdots GO^{(e_L)}_a P_0\\
\Lambda_2&=& P_0 O_a^{(e_1)}G\qquad\qquad & \Gamma_2&=&O_a^{(e_2)}GO^{(e_3)}_aG\cdots GO^{(e_L)}_a P_0\\
\Lambda_k&=&P_0 O_a^{(e_1)}GO_a^{(e_2)}\cdots O_a^{(e_{k-1})} G  & \Gamma_k &=& O_a^{(e_{k})} GO^{(e_{k+1})}_a G\cdots GO^{(e_L)}_aP_0\textrm{ for }k=3,\ldots,L-1\\
\Lambda_L&=&P_0 O_a^{(e_1)}GO_a^{(e_2)}\cdots O_a^{(e_{L-1})} G  & \Gamma_L &=& O^{(e_L)}_aP_0\ .
\end{matrix}
\end{align}
such that
\begin{align}
P_0O^{(e_1)}_a G O^{(e_2)}_a G\cdots GO^{(e_L)}_aP_0&=\Lambda_k\Gamma_k\qquad\textrm{ for }k=1,\ldots,L-1\ .\label{eq:lambdakgammak}
\end{align}
Let~$\ket{\psi}$ be a ground state of the Levin-Wen model~$H_0$. 
We claim that for every $k=1,\ldots,L-1$, there is a set of plaquettes~$\cP_k$ and a constant~$c_k$ (independent of the chosen ground state) such that 
\begin{enumerate}[(i)]
\item\label{it:firstind}
$\Lambda_k\Gamma_k\ket{\psi}= c_k\cdot \Lambda_k \left(\prod_{p\in\cP_k}B_p\right)
O_a^{(e_{k})}\cdots O_a^{(e_L)}\ket{\psi}$.
\item\label{it:secondind}
The (unnormalized) state 
$\left(\prod_{p\in\cP_k}B_p\right)
O_a^{(e_{k})}\cdots O_a^{(e_L)}\ket{\psi}$ is an eigenstate of~$H_0$. Its energy~$\epsilon_k$ is independent of the state~$\ket{\psi}$. 
\item\label{it:thirdind}
The set $\cP_k$ only contains plaquettes
which  have two edges in common with~$\{e_k,\ldots,e_L\}$. 
\end{enumerate}
Note that for $k=1$, this implies 
  $P_0O^{(e_1)}_a G O^{(e_2)}_a G\cdots GO^{(e_L)}_aP_0=c_1\cdot 
P_0 O^{(e_1)}_a  O^{(e_2)}_a \cdots O^{(e_L)}_aP_0$
because $P_0 B_p=P_0$, and the claim~\eqref{eq:pzeroprodv} follows with~\eqref{eq:productooperators}

Properties~\eqref{it:firstind},~\eqref{it:secondind} hold for
$k=L$, with $\cP_{L}=\emptyset$ and $\epsilon_{L}=2$: we have for any ground state~$\ket{\psi}$
\begin{align}
\Gamma_{L}\ket{\psi}&=O_a^{(e_L)}\ket{\psi}
\end{align}
and this is an eigenstate of $H_0$ with energy~$2$ according to Lemma~\ref{lem:annihilationgroundspacex}. 

Assume now that~\eqref{it:firstind},~\eqref{it:secondind}
hold for some~$k\in \{2,L\}$. Then we have according to~\eqref{eq:lambdakgammak}
\begin{align}
\Lambda_{k-1}\Gamma_{k-1}\ket{\psi}&= \Lambda_k \Gamma_k\ket{\psi}\\
&=c_k \Lambda_k \left(\prod_{p\in\cP_k}B_p\right)O_a^{(e_{k})}\cdots O_a^{(e_L)}\ket{\psi}\\
&=c_k (\Lambda_{k-1}O_a^{(e_{k-1})}G)\left(\prod_{p\in\cP_k}B_p\right)O_a^{(e_{k})}\cdots O_a^{(e_L)}\ket{\psi}\\
&=c_{k-1}\cdot \Lambda_{k-1}O_a^{(e_{k-1})}\left(\prod_{p\in\cP_k}B_p\right)O_a^{(e_{k})}\cdots O_a^{(e_L)}\ket{\psi}\ ,
\end{align}
where
\begin{align}
c_{k-1}&=\begin{cases}
\frac{c_k}{E_0-\epsilon_k}\qquad& \textrm{ if }\epsilon_k>E_0\\
0 & \textrm{ otherwise}
\end{cases}
\end{align}

It hence suffices to show that
for some choice of plaquettes~$\cP_{k-1}$, we have
\begin{enumerate}[(a)]
\item $\Lambda_{k-1}O_a^{(e_{k-1})}\left(\prod_{p\in\cP_k}B_p\right)
O_a^{(e_{k})}\cdots O_a^{(e_L)}\ket{\psi}= \Lambda_{k-1} \left(\prod_{p\in\cP_{k-1}}B_p\right) O_a^{(e_{k-1})}\cdots O_a^{(e_L)}\ket{\psi}$\label{it:firstlamb}
\item
$\left(\prod_{p\in\cP_{k-1}}B_p\right) O_a^{(e_{k-1})}\cdots O_a^{(e_L)}\ket{\psi}$
is an eigenstate of $H_0$ with energy $\epsilon_{k-1}$ (independent of~$\ket{\psi}$). 
\item
that the set $\cP_{k-1}$ only contains plaquettes sharing two edges with $\{e_{k-1},\ldots,e_L\}$. \label{it:thirdlmab}
\end{enumerate}
By assumption~\eqref{it:thirdind}
and the particular choice of~$\{e_1,\ldots,e_L\}$, none of the plaquettes~$p\in\cP_k$ contains the edge~$e_{k-1}$.  Therefore, we can commute 
the operator $O_a^{(e_{k-1})}$ through, getting
\begin{align}
O_a^{(e_{k-1})}\left(\prod_{p\in\cP_k}B_p\right)
O_a^{(e_{k})}\cdots O_a^{(e_L)}\ket{\psi}
&=
\left(\prod_{p\in\cP_k}B_p\right)
O_a^{(e_{k-1})}O_a^{(e_{k})}\cdots O_a^{(e_L)}\ket{\psi}\label{eq:statevbx} 
\end{align}
We then consider two cases:
\begin{itemize}
\item
If $e_{k-1}$ does not lie on the same plaquette
as any of the edges~$\{e_{k},\ldots,e_L\}$,
then application of $O_a^{(e_{k-1})}$ creates a pair of excitations according to Lemma~\ref{lem:annihilationgroundspacex}
and
the state~\eqref{eq:statevbx}
is an eigenstate of~$H_0$ with energy~$\epsilon_{k-1}=\epsilon_k+2>E_0$. In particular, setting $\cP_{k-1}=\cP_{k}$, properties~\eqref{it:firstlamb}--\eqref{it:thirdlmab} follow.
\item
If there is an edge~$e_\ell$, $\ell\geq k$ such that 
$e_{k-1}$ and $e_\ell$ belong to the same plaquette~$\tilde{p}$, then
the state~\eqref{eq:statevbx} 
is a superposition of states with $B_{\tilde{p}}$ excited/not excited, that is, we have
\begin{align}
\ket{\varphi}=\left(\prod_{p\in\cP_k}B_p\right)
O_a^{(e_{k-1})}O_a^{(e_{k})}\cdots O_a^{(e_L)}\ket{\psi}
&=(I-B_{\tilde{p}})\ket{\varphi}+B_{\tilde{p}}\ket{\varphi}\ .
\end{align}
However, an excitation at~$\tilde{p}$
cannot disappear by applying the operators $O^{(e_1)}_a,\ldots,O^{(e_{k-2})}_a$
since these do not share an edge with~$\tilde{p}$, hence $\Lambda_{k-1}(I-B_{\tilde{p}})\ket{\varphi}=0$ (recall that $\Lambda_{k-1}=P_0\Lambda_{k-1}$
includes a projection onto the ground space).  Thus setting $\cP_{k-1}=\cP_{k}\cup \{\tilde{p}\}$, we  can verify that~$(a)$--$(c)$ indeed are satisfied.
(The case where there are two such plaquettes~$\tilde{p}$ can be treated analogously.)
\end{itemize}
\end{proof}

Let us compute the effective Hamiltonian~\eqref{eq:heffLlevw}
for the case of the rhombic torus, or more specifically, the 
lattice we use in the numerical simulation, Fig.~\ref{fig_lattice}.
It  has three inequivalent weight-2 loops: $\{10,12\},\{1,2\}, \{5,7\}$. 
Follow the recipe in Section~\ref{sec:stringoperatorstqft}, respectively Section~\ref{sec:symmetryv}, these three loops are related by a $120^{\circ}$ rotation. The corresponding unitary transformation for this rotation
is given by the product of matrices~$A=TS$
when expressed in the flux basis discussed in Section~\ref{sec:stringoperatorstqft} (for the doubled Fibonacci model, the latter two matrices are given by~\eqref{eq:doubledfibst}). Similarly, we can express the action of $F_{(a,a)}(C)$ in
this basis using~\eqref{eq_idempotents_to_F}, getting a matrix~$F$. By~\eqref{eq_Heff_rhombic}, the effective Hamiltonian for the perturbation $-\epsilon \sum_j Z_j$ is then proportional to (when expressed in the same basis)
\begin{equation}
\Heff\sim -(F+A^{-1}FA+A^{-2}FA^{2}).
\end{equation}
Note that the overall sign of the effective Hamiltonian is not specified in~\eqref{eq:hefflepstqft}, but can be determined to be negative here by explicit calculation. For example, substituting in the $S$ matrix (Eq.~\eqref{eq:doubledfibst}) of the doubled Fibonacci model, we have $F=\mathsf{diag}(\varphi+1,-1,-1,\varphi+1)$
for the Fibonacci model. It is then straightforward to obtain the ground state of~$\Heff$, which is
\begin{align}
\label{eq_analytical_perturbed_ground_state}
0.715\ket{(1,1)}+(0.019-0.057i)\ket{(\tau,1)}+(0.019+0.057i)\ket{(1,\tau)}+0.693\ket{(\tau,\tau)} \ ,
\end{align}
where $\ket{(a,b)}$ is a flux basis vector, i.e., the image of $P_{(a,b)}(C)$  (see Section~\ref{sec:stringoperatorstqft}) up to some phase.

\clearpage
\bibliographystyle{alpha}
\bibliography{main}

\end{document}